\documentclass[11pt]{article}

\usepackage{graphicx,amssymb,amsmath}
\usepackage{amsfonts}
\usepackage{booktabs}
\usepackage{cite}
\usepackage{color}
\usepackage{enumerate}
\usepackage{float}
\usepackage{hyperref}
\usepackage{mathrsfs}
\usepackage{subfig}
\usepackage{tabularx}
\usepackage{authblk}

\usepackage[in]{fullpage}
\usepackage[top=0.8in, bottom=0.9in, left=0.9in, right=0.9in]{geometry}

\newcommand{\spm}{\mbox{$S\!P\!M$}}

\def\calC{\mathcal{C}}
\def\calV{\mathcal{V}}
\def\calP{\mathcal{P}}
\def\calF{\mathcal{F}}
\def\calM{\mathcal{M}}
\def\calS{\mathcal{S}}
\def\calU{\mathcal{U}}
\newcommand{\Tri}{\mbox{$T\!r\!i$}}
\def\bay{bay(\overline{cd})}
\def\canal{canal(x,y)}
\def\td{\tilde{d}}
\def\st{$s$-$t$}
\newtheorem{observation}{Observation}

\newtheorem{lemma}{Lemma}
\newtheorem{theorem}{Theorem}
\newtheorem{corollary}{Corollary}
\newenvironment{proof}{\noindent {\textbf{Proof:}}\rm}{\hfill $\Box$ \rm\bigskip}

\title{A New Algorithm for Euclidean Shortest Paths in the Plane\thanks{A preliminary version will appear in {\em Proceedings of the 53rd Annual ACM Symposium on Theory of Computing (STOC 2021)}. This research was supported in part by NSF under Grant CCF-2005323.}}

\author{
Haitao Wang
}
\affil{Department of Computer Science \\
Utah State University, Logan, UT 84322, USA
\\ {\tt haitao.wang@usu.edu}}

\begin{document}

\pagestyle{plain}
\pagenumbering{arabic}
\setcounter{page}{1}
\date{}

\thispagestyle{empty}
\maketitle

\vspace{-0.35in}
\begin{abstract}
Given a set of pairwise disjoint polygonal obstacles in the plane,
finding an obstacle-avoiding Euclidean shortest path between two
points is a classical problem in computational geometry and has been
studied extensively. Previously, Hershberger and Suri [SIAM J. Comput. 1999] gave an algorithm of $O(n\log n)$ time and $O(n\log n)$ space, where $n$ is the total number of vertices of all obstacles. Recently, by modifying Hershberger and Suri's algorithm, Wang [SODA 2021] reduced the space to $O(n)$ while the runtime of the algorithm is still $O(n\log n)$.
In this paper, we present a new algorithm of $O(n+h\log h)$ time and $O(n)$ space, provided that a triangulation of the free space is given, where $h$ is the number of obstacles. The algorithm, which improves the previous work when $h=o(n)$, is optimal in both time and space as $\Omega(n+h\log h)$ is a lower bound on the runtime.
Our algorithm builds a shortest path map for a source point $s$, so
that given any query point $t$, the shortest path length from $s$
to $t$ can be computed in $O(\log n)$ time and a shortest \st\ path can be
produced in additional time linear in the number of edges of the
path.
\end{abstract}


\section{Introduction}
\label{sec:intro}

Let $\calP$ be a set of $h$ pairwise disjoint polygonal obstacles with
a total of $n$ vertices in the plane. Let $\calF$ denote the {\em free space}, i.e., the plane
minus the interior of the obstacles. Given two points $s$ and $t$ in $\calF$, we consider the problem of
finding a Euclidean shortest path from $s$ to $t$ in $\calF$. This is a
classical problem in computational geometry and has been studied
extensively,
e.g.,~\cite{ref:ChenCo13,ref:GhoshAn91,ref:HershbergerAn99,ref:MitchellA91,ref:MitchellSh96,ref:SharirOn86,ref:StorerSh94,ref:RohnertSh86,ref:GuibasOp89,ref:GuibasLi87,ref:HershbergerA91,ref:HershbergerCo94,ref:LeeEu84,ref:WangSh21}.

To solve the problem, two
methods are often used in the literature: the visibility graph and the continuous
Dijkstra. The visibility graph method is to first construct the
visibility graph of the vertices of $\calP$ along with $s$ and $t$,
and then run Dijkstra's shortest path algorithm on the graph to find a
shortest \st\ path. The best algorithms for constructing the
visibility graph run in $O(n\log n+K)$ time~\cite{ref:GhoshAn91} or in
$O(n+h\log^{1+\epsilon}h +K)$ time~\cite{ref:ChenA15} for any constant
$\epsilon>0$, where $K$ is the number of edges of the visibility graph. Because
$K=\Omega(n^2)$ in the worst case, the visibility graph method
inherently takes quadratic time. To deal with the case where $h$ is relatively small
comparing to $n$, a variation of the
visibility graph method was proposed that is to first construct a
so-called {\em tangent graph} and then find a shortest \st\ path in
the graph. Using this method, a shortest \st\ path can be found in $O(n+h\log h+K')$ time~\cite{ref:ChenCo15} after the free space $\calF$ is triangulated, where $K'$ may be considered as the number of tangents among obstacles of
$\calP$ and $K'=O(h^2)$. Note that triangulating $\calF$ can be done
in $O(n\log n)$ time or in $O(n+h\log^{1+\epsilon}h)$ time for any small
$\epsilon>0$~\cite{ref:Bar-YehudaTr94}. Hence, the running time of
the above algorithm in~\cite{ref:ChenCo15} is still quadratic in the worst case.

Using the continuous Dijkstra method, Mitchell~\cite{ref:MitchellSh96} made a breakthrough and achieved the first subquadratic algorithm of $O(n^{3/2+\epsilon})$ time for any constant $\epsilon>0$. Also using the continuous Dijkstra approach plus a novel conforming subdivision of the free space, Hershberger and Suri~\cite{ref:HershbergerAn99} presented an algorithm of $O(n\log n)$ time and $O(n\log n)$ space; the running time is optimal when $h=\Theta(n)$ as $\Omega(n+h\log h)$ is a lower bound in the algebraic computation tree model (which can be obtained by a reduction from sorting; e.g., see Theorem~3~\cite{ref:deRezendeRe89} for a similar reduction). Recently, by modifying Hershberger and Suri's algorithm, Wang~\cite{ref:WangSh21} reduced the space to $O(n)$ while the running time is still $O(n\log n)$.\footnote{An unrefereed report
\cite{ref:InkuluA10} announced an algorithm based on the
continuous Dijkstra approach with $O(n+h\log h\log n)$ time and $O(n)$ space.}

All three continuous Dijkstra algorithms~\cite{ref:HershbergerAn99,ref:MitchellSh96,ref:WangSh21}
construct the {\em shortest path map}, denoted by $\spm(s)$,
for a source point $s$. $\spm(s)$ is of $O(n)$ size and can be used to
	answer shortest path queries. By building a point location data
	structure on $\spm(s)$ in additional $O(n)$ time~\cite{ref:EdelsbrunnerOp86,ref:KirkpatrickOp83}, given a query point $t$, the
	shortest path length from $s$ to $t$ can be computed in $O(\log
	n)$ time and a shortest \st\ path can be output in time linear
	in the number of edges of the path.

The problem setting for $\calP$ is usually referred to as {\em
polygonal domains} or {\em polygons with holes} in the literature. The
problem in simple polygons is relatively
easier~\cite{ref:GuibasOp89,ref:GuibasLi87,ref:HershbergerA91,ref:HershbergerCo94,ref:LeeEu84}.
Guibas et al.~\cite{ref:GuibasLi87} presented an algorithm that can
construct a shortest path map in
linear time. For two-point shortest path query problem where both $s$
and $t$ are query points, Guibas and
Hershberger~\cite{ref:GuibasOp89,ref:HershbergerA91} built a data
structure in linear time such that each query can be answered in
$O(\log n)$ time. In contrast, the two-point query problem in
polygonal domains is much more challenging: to achieve $O(\log n)$
time queries, the current best result uses $O(n^{11})$
space~\cite{ref:ChiangTw99}; alternatively Chiang and
Mitchell~\cite{ref:ChiangTw99} gave a data structure of $O(n+h^5)$
space with $O(h\log n)$ query time. Refer to~\cite{ref:ChiangTw99} for
other data structures with trade-off between space and query time.

The $L_1$ counterpart of the problem where the path length is measured
in the $L_1$ metric also attracted much attention, e.g.,
\cite{ref:BaeL119,ref:ClarksonRe87,ref:ClarksonRe88,ref:MitchellAn89,ref:MitchellL192,ref:ChenCo19}.
For polygons with holes,
Mitchell~\cite{ref:MitchellAn89,ref:MitchellL192} gave an algorithm
that can build a shortest path map for a source point in
$O(n\log n)$ time and $O(n)$ space; for small $h$, Chen and Wang~\cite{ref:ChenCo19}
proposed an algorithm of $O(n+h\log
h)$ time and $O(n)$ space, after the free space is triangulated.
For simple polygons, Bae and Wang~\cite{ref:BaeL119} built a data
structure in linear time that can answer each two-point $L_1$ shortest path
query in $O(\log n)$ time. The two-point query problem in polygons with holes has also been studied~\cite{ref:ChenTw16,ref:ChenSh00,ref:WangA20}. To achieve $O(\log n)$ time queries, the current best result uses $O(n+h^2\log^3h/\log\log h)$ space~\cite{ref:WangA20}.

\subsection{Our result}


In this paper, we show that the problem of finding an Euclidean shortest path among obstacles in $\calP$ is solvable in $O(n+h\log h)$ time and $O(n)$ space, after a triangulation of the free space $\calF$ is given.
If the time for triangulating $\calF$ is included and the triangulation algorithm in~\cite{ref:Bar-YehudaTr94} is used, then the total time of the algorithm is $O(n+h\log^{1+\epsilon}h)$, for any constant $\epsilon>0$.\footnote{If randomization is allowed, the algorithm of Clarkson, Cole, and  Tarjan~\cite{ref:ClarksonRa92} can compute a triangulation in $O(n\log^* n+h\log h)$ expected time.}
With the assumption that the triangulation could be done in $O(n+h\log h)$ time, which has been an open problem and is beyond the scope of this paper, our result settles Problem 21 in The Open Problem Project~\cite{ref:TOPP}. Our algorithm actually constructs the shortest path map $\spm(s)$ for the source point $s$ in $O(n+h\log h)$ time and $O(n)$ space. We give an overview of our approach below.

The high-level scheme of our algorithm is similar to that for the $L_1$ case~\cite{ref:WangA20} in the sense that we first solve the {\em convex case} where all obstacles of $\calP$ are convex and then extend the algorithm to the general case with the help of the extended corridor structure of $\calP$~\cite{ref:ChenCo19,ref:ChenTw16,ref:ChenCo17,ref:ChenA15,ref:KapoorAn97,ref:MitchellSe95}.

\paragraph{The convex case.}
We first discuss the convex case. Let $\calV$ denote the set of topmost, bottommost, leftmost, and rightmost vertices of all obstacles. Hence, $|\calV|\leq 4h$. Using the algorithm of Hershberger and Suri~\cite{ref:HershbergerAn99}, we build a conforming subdivision $\calS$ on the points of $\calV$, without considering the obstacle edges. Since $|\calV|=O(h)$, the size of $\calS$ is $O(h)$. Then, we insert the obstacle edges into $\calS$ to build a conforming subdivision $\calS'$ of the free space.
The subdivision $\calS'$ has $O(h)$ cells (in contrast, the conforming subdivision of the free space in~\cite{ref:HershbergerAn99} has $O(n)$ cells).
Unlike the subdivision in~\cite{ref:HershbergerAn99} where each cell is of constant size, here the size of each cell of $\calS'$ may not be constant but its boundary consists of $O(1)$ transparent edges and $O(1)$ convex chains (each of which belongs to the boundary of an obstacle of $\calP$). Like the subdivision in~\cite{ref:HershbergerAn99}, each transparent edge $e$ of $\calS'$ has a well-covering region $\calU(e)$. In particular, for each transparent edge $f$ on the boundary of $\calU(e)$, the shortest path distance between $e$ and $f$ is at least $2\cdot\max\{|e|,|f|\}$. Using $\calS'$ as a guidance, we run the continuous Dijkstra algorithm as in~\cite{ref:HershbergerAn99} to expand the wavefront, starting from the source point $s$.
A main challenge our algorithm needs to overcome (which is also a main difference between our algorithm and that in~\cite{ref:HershbergerAn99}) is that each cell in our subdivision $\calS'$ may not be of constant size. One critical property our algorithm relies on is that the boundary of each cell of $\calS'$ has $O(1)$ convex chains. Our strategy is to somehow treat each such convex chain as a whole. We also borrow some idea from the algorithm of Hershberger, Suri, and Y{\i}ld{\i}z~\cite{ref:HershbergerA13} for computing shortest paths among curved obstacles. To guarantee the $O(n+h\log h)$ time, some global charging analysis is used. In addition, the tentative prune-and-search technique of Kirkpatrick and Snoeyink~\cite{ref:KirkpatrickTe95} is applied to perform certain operations related to bisectors, in logarithmic time each. Finally, the techniques of Wang~\cite{ref:WangSh21} are utilized to reduce the space to $O(n)$. All these efforts lead to an $O(n+h\log h)$ time and $O(n)$ space algorithm to construct the shortest path map $\spm(s)$ for the convex case.

\paragraph{The general case.}
We extend the convex case algorithm to the general case where obstacles may not be convex. To this end, we resort to  the extended corridor structure of $\calP$, which was used before for reducing the time complexities from $n$ to $h$, e.g.,~\cite{ref:ChenCo19,ref:ChenTw16,ref:ChenCo17,ref:ChenA15,ref:KapoorAn97,ref:MitchellSe95}. The structure partitions the free space $\calF$ into an ocean $\calM$, $O(n)$ bays, and $O(h)$ canals.
Each bay is a simple polygon that shares an edge with $\calM$. Each canal is a simple polygon that shares two edges with $\calM$. But two bays or two canals, or a bay and a canal do not share any edge. A common edge of a bay (or canal) with $\calM$ is called a gate. Thus each bay has one gate and each canal has two gates. Further,
$\calM$ is bounded by $O(h)$ convex chains (each of which is on the boundary of an obstacle).
An important property related to shortest paths is that if both $s$ and $t$ are in $\calM$, then any shortest \st\ path must be in the union of $\calM$ and all corridor paths, each of which is contained in a canal.
As the boundary of $\calM$ consists of $O(h)$ convex chains, by incorporating all corridor paths, we can easily extend our convex case algorithm to computing $\spm(\calM)$,
the shortest path map $\spm(s)$ restricted to $\calM$, i.e.,
$\spm(\calM)=\spm(s)\cap \calM$. To compute the entire map $\spm(s)$,
we expand $\spm(\calM)$ to all bays and canals through their gates.
For this, we process each bay/canal individually. For each
bay/canal $C$, expanding the map into $C$ is actually a special case
of the (additively)-weighted geodesic Voronoi diagram problem on a
simple polygon where all sites are outside $C$ and can influence $C$
only through its gates. In summary, after a triangulation of $\calF$
is given, building $\spm(\calM)$ takes $O(n+h\log h)$ time, and
expanding $\spm(\calM)$ to all bays and canals takes additional
$O(n+h\log h)$ time. The space of the algorithm is bounded by $O(n)$.

\paragraph{Outline.}
The rest of the paper is organized as follows. Section~\ref{sec:pre}
defines notation and introduces some concepts. Section~\ref{sec:convex} presents the algorithm for the convex case.
The general case is discussed in Section~\ref{sec:general}.

\section{Preliminaries}
\label{sec:pre}

For any two points $a$ and $b$ in the plane, denote by $\overline{ab}$
the line segment with $a$ and $b$ as endpoints; denote by
$|\overline{ab}|$ the length of the segment.

For any two points $s$ and $t$ in the free space $\calF$, we use
$\pi(s,t)$ to denote a shortest path from $s$ to $t$ in $\calF$. In the case
where shortest paths are not unique, $\pi(s,t)$ may refer to an
arbitrary one. Denote by $d(s,t)$ the length of $\pi(s,t)$;
we call $d(s,t)$ the {\em geodesic distance} between $s$ and $t$. For
two line segments $e$ and $f$ in $\calF$, their {\em geodesic
distance} is defined to be the minimum geodesic distance between any
point on $e$ and any point on $f$, i.e., $\min_{s\in e, t\in
f}d(s,t)$; by slightly abusing the notation, we use $d(e,f)$ to denote
their geodesic distance.
For any path $\pi$ in the plane, we use $|\pi|$ to denote its
length.

For any compact region $A$ in the plane, let $\partial A$ denote its boundary. We use $\partial \calP$ to denote the union of the boundaries of all obstacles of $\calP$.

Throughout the paper, we use $s$ to refer to the source point. For
convenience, we consider $s$ as a degenerate obstacle in $\calP$. We
often refer to the vertices of $\calP$ as {\em obstacle vertices} and
refer to the edges of $\calP$ as {\em obstacle edges}.
For any point $t\in \calF$, we call the adjacent vertex of $t$ in $\pi(s,t)$ the {\em anchor} of $t$ in $\pi(s,t)$\footnote{Usually ``predecessor'' is used in the literature instead of ``anchor'', but here we reserve ``predecessor'' for other purpose.}.

\begin{figure}[t]
\begin{minipage}[t]{\textwidth}
\begin{center}
\includegraphics[height=2.5in]{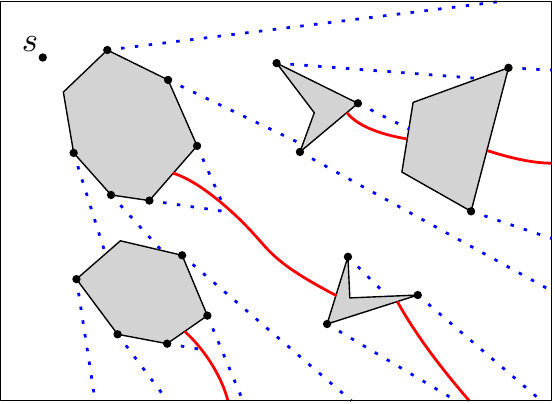}
\caption{\footnotesize Illustrating the shortest path map $\spm(s)$. The solid (red) curves are walls and the
(blue) dotted segments are windows. The anchor of each cell is also shown with a black point.}
\label{fig:spm}
\end{center}
\end{minipage}
\vspace{-0.15in}
\end{figure}

The {\em shortest path map} $\spm(s)$ of $s$ is a decomposition of the free space $\calF$ into maximal regions such that all points in each region $R$ have the same anchor~\cite{ref:HershbergerAn99,ref:MitchellA91} in their shortest paths from $s$; e.g., see Fig.~\ref{fig:spm}. Each edge of $\spm(s)$ is either an obstacle edge fragment or a {\em bisecting-curve}\footnote{This is usually called {\em bisector} in the literature. Here we reserve the term ``bisector'' to be used later.}, which is the locus of points $p$ with $d(s,u)+|\overline{pu}|=d(s,v)+|\overline{pv}|$ for two obstacle vertices $u$ and $v$. Each bisecting-curve is in general a hyperbola; a special case happens if one of $u$ and $v$ is the anchor of the other, in which case their bisecting-curve is a straight line. Following the notation in~\cite{ref:Eriksson-BiqueGe15}, we differentiate between two types of bisecting curves: {\em walls} and {\em windows}. A bisecting curve of $\spm(s)$ is a {\em wall} if there exist two topologically different shortest paths from $s$ to each point of the edge; otherwise (i.e., the above special case) it is a {\em window} (e.g., see Fig.~\ref{fig:spm}).


We make a general position assumption that for each obstacle vertex $v$, there is a unique shortest path from $s$ to $v$, and for any point $p$ in the plane, there are at most three different shortest paths from $s$ to $p$. The assumption assures that each vertex of $\spm(s)$ has degree at most three, and there are at most three bisectors of $\spm(s)$ intersecting at a common point, which is sometimes called a {\em triple point} in the literature~\cite{ref:Eriksson-BiqueGe15}.

A curve in the plane is {\em $x$-monotone} if its intersection with any vertical line is connected; the {\em $y$-monotone} is defined similarly. A curve is $xy$-monotone if it is both $x$- and $y$-monotone.

The following observation will be used throughout the paper without explicitly mentioning it again.

\begin{observation}\label{obser:hlogh}
$n+h\log n=O(n+h\log h)$.
\end{observation}
\begin{proof}
Indeed, if $h<n/\log n$, then $n+h\log n=\Theta(n)$, which is $O(n+h\log h)$; otherwise, $\log n=O(\log h)$ and $n+h\log n=O(n+h\log h)$.
\end{proof}

\section{The convex case}
\label{sec:convex}

In this section, we present our algorithm for the convex case where all obstacles of $\calP$ are convex. The algorithm will be extended to the general case in Section~\ref{sec:general}.

For each obstacle $P\in \calP$, the topmost, bottommost, leftmost, and rightmost vertices of $P$ are called {\em rectilinear extreme vertices}. The four rectilinear extreme vertices partition $\partial P$ into four portions and each portion is called an {\em elementary chain}, which is convex and $xy$-monotone. For technical reason that will be clear later, we assume that each rectilinear extreme vertex $v$ belongs to the elementary chain counterclockwise of $v$ with respect to the obstacle (i.e., $v$ is the clockwise endpoint of the chain; e.g., see Fig.~\ref{fig:elechain}).
We use {\em elementary chain fragment} to refer to a portion of an elementary chain.

\begin{figure}[t]
\begin{minipage}[t]{0.46\textwidth}
\begin{center}
\includegraphics[height=1.4in]{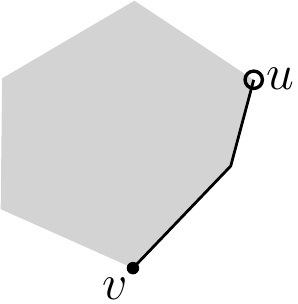}
\caption{\footnotesize Illustrating an elementary chain (the thick segments), which contains the vertex $v$ but not $u$.}
\label{fig:elechain}
\end{center}
\end{minipage}
\begin{minipage}[t]{0.54\textwidth}
\begin{center}
\includegraphics[height=1.4in]{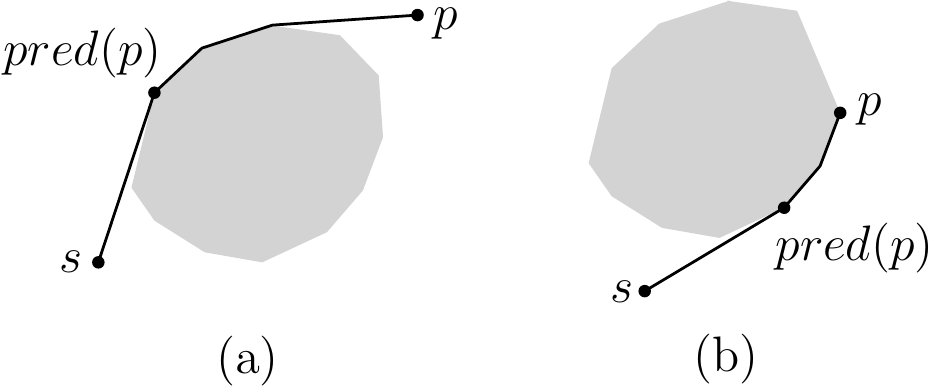}
\caption{\footnotesize Illustrating the predecessor.}
\label{fig:pred}
\end{center}
\end{minipage}
\vspace{-0.15in}
\end{figure}

We introduce some notation below that is similar in spirit to those
from~\cite{ref:HershbergerA13} for shortest paths among curved
obstacles.


Consider a shortest path $\pi(s,p)$ from $s$ to a point $p$ in the free space $\calF$.
It is not difficult to see that $\pi(s,p)$ is a sequence of
elementary chain fragments and common tangents between obstacles of
$\calP\cup \{p\}$. We define the {\em predecessor} of $p$, denoted
by $pred(p)$, to be the initial vertex of the last elementary
chain fragment in $\pi(s,p)$ (e.g., see Fig.~\ref{fig:pred} (a)). Note that since each rectilinear extreme vertex belongs to a single elementary chain, $pred(p)$ in $\pi(s,p)$ is unique. A special case happens if $p$ is a rectilinear extreme vertex and $\pi(s,p)$ contains a portion of an elementary chain $A$ clockwise of $p$.
In this case, we let $pred(p)$ be endpoint of the fragment of $A$ in $\pi(s,p)$ other than $p$ (e.g., see Fig.~\ref{fig:pred}~(b)); in this way, $pred(p)$ is unique in $\pi(s,p)$. Note that $p$ may still have multiple
predecessors if there are multiple shortest paths from $s$ to $p$.
Intuitively, the reason we define predecessors as above is to treat each elementary chain somehow as a whole, which is essential for reducing the runtime of the algorithm from $n$ to $h$.


The rest of this section is organized as follows. In Section~\ref{sec:subdivision}, we compute a conforming subdivision $\calS'$ of the free space $\calF$. Section~\ref{sec:notation} introduces some basic concepts and notation for our algorithm. The wavefront expansion algorithm is presented in Section~\ref{sec:algo}, with two key subroutines of the algorithm described in Section~\ref{sec:merge} and Section~\ref{sec:propagation}, respectively. Section~\ref{sec:time} analyzes the time complexity of the algorithm, where a technical lemma is proved separately in Section~\ref{sec:lemnumgen}. Using the information computed by the wavefront expansion algorithm, Section~\ref{sec:spm} constructs the shortest path map $\spm(s)$. The overall algorithm runs in $O(n+h\log h)$ time and $O(n+h\log h)$ space. Section~\ref{sec:space} reduces the space to $O(n)$ while keeping the same runtime, by using the techniques from Wang~\cite{ref:WangSh21}.

\subsection{Computing a conforming subdivision of the free space}
\label{sec:subdivision}

Let $\calV$ denote the set of the rectilinear extreme vertices of all
obstacles of $\calP$. Hence, $|\calV|=O(h)$. Using the algorithm of
algorithm of Hershberger and Suri~\cite{ref:HershbergerAn99} (called
the {\em HS} algorithm), we build a conforming subdivision $\calS$
with respect to the vertices of $\calV$, without considering the
obstacle edges.

The subdivision $\calS$, which is of size $O(h)$, is a quad-tree-style subdivision of the plane into $O(h)$ cells. Each cell of $\calS$ is a square or a square annulus  (i.e., an outer square with an inner square hole). Each vertex of $\calV$ is contained in the interior of a square cell and each square cell contains at most one vertex of $\calV$.
Each edge $e$ of $\calS$ is axis-parallel and {\em well-covered},
i.e., there exists a set $\calC(e)$ of $O(1)$ cells of $\calS$ such
that their union $\calU(e)$ contains $e$ with the following
properties: (1) the total complexity of all cells of $\calC(e)$ is
$O(1)$ and thus the size of $\calU(e)$ is $O(1)$; (2) for any edge $f$
of $\calS$ that is on $\partial \calU(e)$ or outside $\calU(e)$, the
Euclidean distance between $e$ and $f$ (i.e., the minimum
$|\overline{pq}|$ among all points $p\in e$ and $q\in f$) is at least
$2\cdot\max\{|e|,|f|\}$; (3) $\calU(e)$, which is called the {\em
well-covering region} of $e$, contains at most one vertex of $\calV$.
In addition, each cell $c$ of $\calS$ has $O(1)$ edges on its boundary with the
following {\em uniform edge property}:
the lengths of the edges on the boundary of $c$ differ by at most a
factor of $4$, regardless of whether $c$ is a square or square annulus.

The subdivision $\calS$ can be computed in $O(h\log h)$ time and $O(h)$ space~\cite{ref:HershbergerAn99}.

Next we insert the obstacle edges into $\calS$ to produce a conforming
subdivision $\calS'$ of the free space $\calF$. In $\calS'$, there are two
types of edges: those introduced by the subdivision construction
(which are in the interior of $\calF$ except possibly their endpoints)
and the obstacle edges; we call the former the {\em transparent edges}
(which are axis-parallel) and the latter the {\em opaque edges}. The
definition of $\calS'$ is similar to the conforming subdivision of the
free space used in the HS algorithm. A main difference is that here
endpoints of each obstacle edge may not be in $\calV$, a consequence of which is that each cell of
$\calS'$ may not be of constant size (while each cell in the
subdivision of the HS algorithm is of constant size). However, each cell $c$
of $\calS'$ has the following property that is critical to our algorithm: The boundary $\partial c$ consists of $O(1)$ transparent edges and $O(1)$ convex chains (each of which is a portion of an elementary chain).

More specifically, $\calS'$ is a subdivision of $\calF$ into $O(h)$ cells. Each cell of $\calS'$ is one of the connected components formed by intersecting $\calF$ with an axis-parallel rectangle (which is the union of a set of adjacent cells of $\calS$) or a square annulus of $\calS$. Each cell of $\calS'$ contains at most one vertex of $\calV$. Each vertex of $\calV$ is incident to a transparent edge. Each transparent edge $e$ of $\calS'$ is {\em well-covered}, i.e., there exists a set $\calC(e)$ of $O(1)$ cells whose union $\calU(e)$ contains $e$ with the following property: for each transparent edge $f$ on $\partial \calU(e)$, the geodesic distance $d(e,f)$ between $e$ and $f$ is at least $2\cdot\max\{|e|,|f|\}$. The region $\calU(e)$ is called the {\em well-covering region} of $e$ and contains at most one vertex of $\calV$. Note that $\calS'$ has $O(h)$ transparent edges.

Below we show how $\calS'$ is produced from $\calS$. The procedure is similar to that in the HS algorithm. We overlay the obstacle edges on top of $\calS$ to obtain a subdivision $\calS_{overlay}$. Because each edge of $\calS$ is axis-parallel and all obstacle edges constitute a total of $O(h)$ elementary chains, each of which is $xy$-monotone, $\calS_{overlay}$ has $O(h^2)$ faces. We say that a face of $\calS_{overlay}$ is {\em interesting} if its boundary contains a vertex of $\calV$ or a vertex of $\calS$. We keep intact the interesting faces of $\calS_{overlay}$ while deleting every edge fragment of $\calS$ not on the boundary of any interesting cell. Further, for each cell $c$ containing a vertex $v\in \calV$, we partition $c$ by extending vertical edges from $v$ until the boundary of $c$. This divides $c$ into at most three subcells. Finally we divide each of the two added edges incident to $v$ into segments of length at most $\delta$, where $\delta$ is the length of the shortest edge on the boundary of $c$. By the uniform edge property of $\calS$, $\partial c$ has $O(1)$ edges, whose lengths differ by at most a factor of $4$; hence dividing the edges incident to $v$ as above produces only $O(1)$ vertical edges. The resulting subdivision is $\calS'$.

As mentioned above, the essential difference between our subdivision $\calS'$ and the one in the HS algorithm is that the role of an obstacle edge in the HS algorithm is replaced by an elementary chain. Therefore, each opaque edge in the subdivision of the HS algorithm becomes an elementary chain fragment in our case. Hence, by the same analysis as in the HS algorithm (see Lemma 2.2~\cite{ref:HershbergerAn99}), $\calS'$ has the properties as described above and the well-covering region $\calU(e)$ of each transparent edge $e$ of $\calS'$ is defined in the same way as in the HS algorithm.

The following lemma computes $\calS'$. It should be noted that although $\calS'$ is defined with the help of $\calS_{overlay}$, $\calS'$ is constructed directly without computing $\calS_{overlay}$ first.

\begin{lemma}\label{lem:subalgo}
The conforming subdivision $\calS'$ can be constructed in $O(n+h\log h)$ time and $O(n)$ space.
\end{lemma}
\begin{proof}
We first construct $\calS$ in $O(h\log h)$ time and $O(h)$ space~\cite{ref:HershbergerAn99}. In
the following we construct $\calS'$ by inserting the obstacle edges
into $\calS$. The algorithm is similar to that in the HS algorithm
(see Lemma 2.3~\cite{ref:HershbergerAn99}). The difference is that we
need to handle each elementary chain as a whole.

We first build a data structure so that for any query horizontal ray with origin in $\calF$, the first obstacle edge of $\calP$ hit by it can be computed in $O(\log n)$ time. This can be done by building a horizontal decomposition of $\calF$, i.e., extend a horizontal segment from each vertex until it hits $\partial \calP$. As all obstacles of $\calP$ are convex, the horizontal decomposition can be computed in $O(n+h\log h)$ time and $O(n)$ space~\cite{ref:HertelFa85}. By building a point location data 	structure~\cite{ref:EdelsbrunnerOp86,ref:KirkpatrickOp83} on the horizontal decomposition in additional $O(n)$ time, each horizontal ray shooting can be answered in $O(\log n)$ time. Similarly, we can construct the vertical decomposition of $\calF$ in $O(n+h\log h)$ time and $O(n)$ space so that each vertical ray shooting can be answered in $O(\log n)$ time.

The edges of $\calS'$ are obstacle edges, transparent edges incident to the vertices of $\calS$, and transparent edges subdivided on the vertical segments incident to the vertices of $\calV$. To identify the second type of edges, we trace the boundary of each interesting cell separately. Starting from a vertex $v$ of $\calS$, we trace along each edge incident to $v$. Using the above ray-shooting data structure, we determine whether the next cell vertex is a vertex of $\calS$ or the first point on $\partial \calP$ hit by the ray. As $\calS$ has $O(h)$ edges and vertices, this tracing takes $O(h\log n)$ time in total. Tracing along obstacle edges is done by starting from each vertex of $\calV$ and following each of its incident elementary chains. For each elementary chain, the next vertex is either the next obstacle vertex on $e$, where $e$ is the current tracing edge of the elementary chain, or the intersection of $e$ with a transparent edge of the current cell. Hence, tracing all elementary chains takes $O(n)$ time in total. The third type of edges can be computed in linear time by local operations on each cell containing a vertex of $\calV$.

Finally, we assemble all these edges together to obtain an adjacency
structure for $\calS'$. For this, we could use a plane sweep algorithm as in the
HS algorithm. However, that would take $O(n\log n)$ time as there are
$O(n)$ edges in $\calS'$. To obtain an $O(n+h\log h)$ time algorithm,
we propose the following approach. During the above tracing of
elementary chains, we record the fragment of the chain that lies in a
single cell. Since each such portion may not be of constant size, we
represent it by a line segment connecting its two endpoints; note that
this segment does not intersect any other cell edges because the
elementary chain fragment is on the boundary of an obstacle (and thus
the segment is inside the obstacle). Then, we apply the plane sweep
algorithm to stitch all edges with each elementary chain fragment in a
cell replaced by a segment as above. The algorithm takes $O(h\log h)$ time and
$O(h)$ space. Finally, for each segment representing an elementary
chain portion, we locally replace it by the chain fragment in linear
time. Hence, this step takes $O(n)$ time altogether for all such
segments. As such, the total time for computing the adjacency
information for $\calS'$ is $O(n+h\log n)$, which is $O(n+h\log h)$.
Clearly, the space complexity of the algorithm is $O(n)$.
\end{proof}



\subsection{Basic concepts and notation}
\label{sec:notation}

Our shortest path algorithm uses the continuous Dijkstra method. The algorithm initially generates a wavefront from $s$, which is a circle centered at $s$. During the algorithm, the wavefront consists of all points of $\calF$ with the same geodesic distance from $s$ (e.g., see Fig.~\ref{fig:wavefront}). We expand the wavefront until all points of the free space are covered. The conforming subdivision $\calS'$ is utilized to guide the wavefront expansion. Our wavefront expansion algorithm follows the high-level scheme as the HS algorithm. The difference is that our algorithm somehow considers each elementary chain as a whole, which is in some sense similar to the algorithm of Hershberger, Suri, and Y{\i}ld{\i}z~\cite{ref:HershbergerA13} (called the {\em HSY} algorithm). Indeed, the HSY algorithm considers each $xy$-monotone convex arc as a whole, but the difference is that each arc in the HSY algorithm is of constant size while in our case each elementary chain may not be of constant size. As such, techniques from both the HS algorithm and the HSY algorithm are borrowed.

\begin{figure}[t]
\begin{minipage}[t]{\textwidth}
\begin{center}
\includegraphics[height=2.0in]{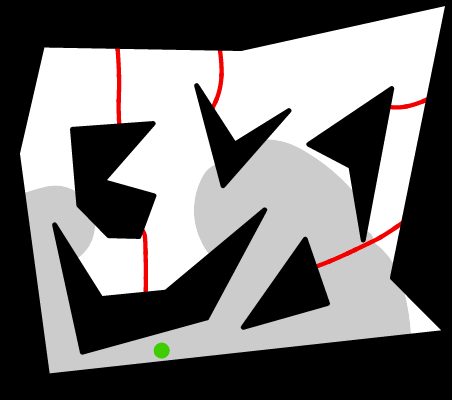}
\caption{\footnotesize Illustrating the wavefront. The black region are obstacles. The green point is $s$. The red curves are bisecting curves of $\spm(s)$. The gray region is the free space that has been covered by the wavefront. The boundary between the white region and the grey region is the wavefront. The figure is generated using the applet at~\cite{ref:HershbergerGe14}.}
\label{fig:wavefront}
\end{center}
\end{minipage}
\vspace{-0.15in}
\end{figure}

We use $\tau$ to denote the geodesic distance from $s$ to all points in the wavefront. One may also think of $\tau$ as a parameter representing time. The algorithm simulates the expansion of the wavefront as time increases from $0$ to $\infty$.
The wavefront comprises a sequence of {\em wavelets}, each emanating from a {\em generator}. In the HS algorithm, a generator is simply an obstacle vertex. Here, since we want to treat an elementary chain as a whole, similar to the HSY algorithm, we define a generator as a couple $\alpha=(A,a)$, where $A$ is an elementary chain and $a$ is an obstacle vertex on $A$, and further a clockwise or counterclockwise direction of $A$ is designated for $\alpha$; $a$ has a weight $w(a)$ (one may consider $w(a)$ as the geodesic distance between $s$ and $a$).
We call $a$ the {\em initial vertex} of the generator $\alpha$.

We say a point $q$ is {\em reachable} by a generator $\alpha=(A,a)$ if one can draw a path in $\calF$ from $a$ to $q$ by following $A$ in the designated direction to a vertex $v$ on $A$ such that $\overline{vq}$ is tangent to $A$ and then following the segment $\overline{vq}$ (e.g., see Fig.~\ref{fig:generator}). The (weighted) distance between the generator $\alpha$ and $q$ is the length of this path plus $w(a)$;
by slightly abusing the notation, we use $d(\alpha,q)$ to denote the distance.
From the definition of reachable points, the vertex $a$ partitions $A$ into two portions and only the portion following the designated direction is relevant (e.g., in Fig.~\ref{fig:generator}, only the portion containing the vertex $v$ is relevant). Henceforth, unless otherwise stated, we use $A$ to refer to its relevant portion only and we call $A$ the {\em underlying chain} of $\alpha$.
In this way the initial vertex $a$ becomes an endpoint of $A$.
For convenience, sometimes we do not differentiate $\alpha$ and $A$. For example, when we say ``the tangent from $q$ to $\alpha$'', we mean ``the tangent from $q$ to $A$''; also, when we say ``a vertex of $\alpha$'', we mean ``a vertex of $A$''.

\begin{figure}[h]
\begin{minipage}[t]{0.51\textwidth}
\begin{center}
\includegraphics[height=1.4in]{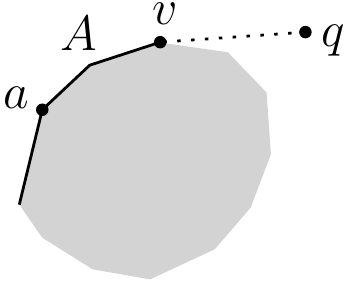}
\caption{\footnotesize Illustrating a generator $\alpha=(A,a)$. $A$ consists of the thick segments, and is designated clockwise direction around the obstacle. $q$ is a reachable point through vertex $v$.}
\label{fig:generator}
\end{center}
\end{minipage}
\hspace{0.02in}
\begin{minipage}[t]{0.47\textwidth}
\begin{center}
\includegraphics[height=1.4in]{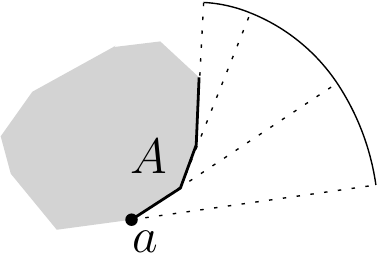}
\caption{\footnotesize Illustrating the wavelet generated by $\alpha=(A,q)$, designated counterclockwise direction. The wavelet has three pieces each of which is a circular arc. }
\label{fig:wavelet}
\end{center}
\end{minipage}
\vspace{-0.15in}
\end{figure}


The wavefront can thus be represented by the sequence of generators of
its wavelets. A wavelet generated by a generator $\alpha=(A,a)$ at time $\tau$ is a
contiguous set of reachable points $q$ such that $d(\alpha,q)=\tau$ and
$d(\alpha',q)\geq \tau$ for all other generators $\alpha'$ in the
wavefront; we also say that $q$ is claimed by $\alpha$. Note that as
$A$ may not be of constant size, a wavelet may
not be of constant size either; it actually consists of a
contiguous sequence of circular arcs centered at the obstacle vertices
$A$ (e.g., see Fig.~\ref{fig:wavelet}). If a point $q$ is claimed by $\alpha$, then
$d(s,q)=d(\alpha,q)=\tau$ and the predecessor $pred(q)$ of $q$ is $a$; sometimes for convenience we also say that the generator $\alpha$ is the predecessor of $q$. If $q$ is on an elementary
chain $A'$ and the tangent from $q$ to $a$ is also tangent to $A'$,
then a new generator $(A',q)$ is added to the wavefront (e.g., see Fig.~\ref{fig:newgenerator}(a)). A special
case happens when $q$ is the counterclockwise endpoint of $A$ (and thus
$q$ does not belong to $A$); in this case, a new generator $\alpha'=(A',q)$
is also added, where $A'$ is the elementary chain that contains $q$ (e.g., see Fig.~\ref{fig:newgenerator}(b)).

\begin{figure}[h]
\begin{minipage}[t]{\textwidth}
\begin{center}
\includegraphics[height=2.2in]{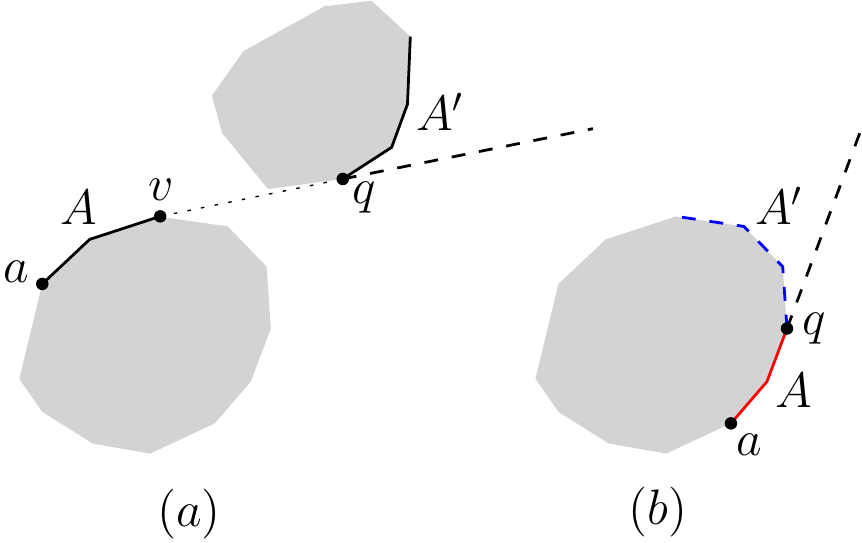}
\caption{\footnotesize Illustrating a new generator $\alpha'=(A',q)$. (a) A general case where both $A$ and $A'$ are marked with thick segments. (b) A special case where $A$ is marked with solid (red) segments and $A'$ is marked with dashed (blue) segments. }
\label{fig:newgenerator}
\end{center}
\end{minipage}
\vspace{-0.15in}
\end{figure}

As $\tau$ increases, the points bounding the adjacent wavelets trace out the bisectors that form the edges of the shortest path map $\spm(s)$. The {\em bisector} between the wavelets of two generators $\alpha$ and $\alpha'$, denoted
by $B(\alpha,\alpha')$, consists of points $q$ with $d(\alpha,q)=d(\alpha',q)$. Note that since $\alpha$ and $\alpha'$ may not be of constant size, $B(\alpha,\alpha')$ may not be of constant size either. More specifically, $B(\alpha,\alpha')$ has multiple pieces each of which is a hyperbola defined by two obstacle vertices $v\in \alpha$ and $v'\in \alpha'$ such that the hyperbola consists of all points that have two shortest paths from $s$ with $v$ and $v'$ as the anchors in the two paths, respectively (e.g., see Fig.~\ref{fig:bisector}).
A special case happens if $\alpha'$ is a generator created by the wavelet of $\alpha$, such as that illustrated in Fig.~\ref{fig:newgenerator}(a), then $B(\alpha,\alpha')$ is the half-line extended from $q$ along the direction from $v$ to $q$ (the dashed segment in Fig.~\ref{fig:newgenerator}(a)); we call such a bisector an {\em extension bisector}.
Note that in the case illustrated in Fig.~\ref{fig:newgenerator}(b), $B(\alpha,\alpha')$, which is also an extension bisector, is the half-line extended from $q$ along the direction from $v$ to $q$ (the dashed segment in Fig.~\ref{fig:newgenerator}(b)), where $v$ is the obstacle vertex adjacent to $q$ in $A$.

\begin{figure}[h]
\begin{minipage}[t]{\textwidth}
\begin{center}
\includegraphics[height=2.4in]{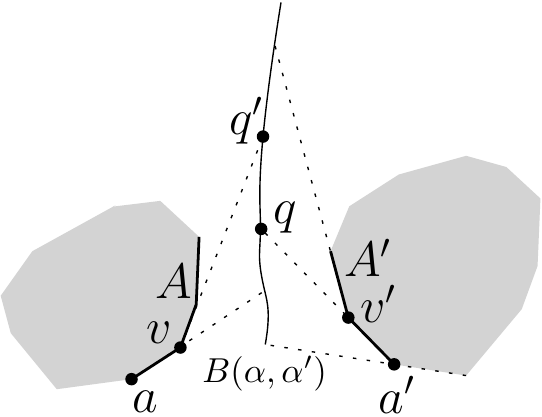}
\caption{\footnotesize Illustrating the bisector $B(\gamma,\gamma')$ defined by two generators $\gamma=(A,q)$ and $\gamma=(A',q')$. The portion between $q$ and $q'$ is a hyperbola defined by $v$ and $v'$. }
\label{fig:bisector}
\end{center}
\end{minipage}
\vspace{-0.15in}
\end{figure}


A wavelet gets eliminated from the wavefront if the two bisectors bounding it intersect, which is called a {\em bisector event}. Specifically, if $\alpha_1$, $\alpha$, and $\alpha_2$ are three consecutive generators of the wavefront, the wavelet generated by $\alpha$ will be eliminated when $B(\alpha_1,\alpha)$ intersects $B(\alpha,\alpha_2)$; e.g., see Fig.~\ref{fig:bisectorintersection}.
Wavelets are also eliminated by collisions with obstacles and other wavelets in front of it.
If a bisector $B(\alpha,\alpha')$ intersects an obstacle, their intersection is also called a {\em bisector event}.

\begin{figure}[t]
\begin{minipage}[t]{\textwidth}
\begin{center}
\includegraphics[height=1.4in]{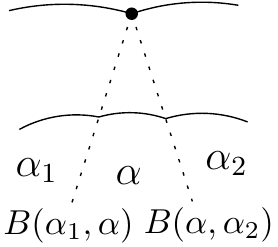}
\caption{\footnotesize Illustrating the intersection of two bisectors. }
\label{fig:bisectorintersection}
\end{center}
\end{minipage}
\vspace{-0.15in}
\end{figure}

Let $\spm'(s)$ be the subdivision of $\calF$ by the bisectors of all generators (e.g. see Fig.~\ref{fig:map}). The intersections of bisectors and intersections between bisectors and obstacle edges are vertices of $\spm'(s)$. Each bisector connecting two vertices is an edge of $\spm'(s)$, called a {\em bisector edge}.
As discussed before, a bisector, which consists of multiple hyperbola pieces, may not be of constant size. Hence, a bisector edge $e$ of $\spm'(s)$ may not be of constant size. For differentiation, we call each hyperbola piece of $e$, a {\em hyperbolic-arc}.
In addition, if the boundary of an obstacle $P$ contains more than one vertex of $\spm'(s)$, then the chain of edges of $\partial P$ connecting two adjacent vertices of $\spm'(s)$ also forms an edge of $\spm'(s)$, called a {\em convex-chain edge}.

\begin{figure}[h]
\begin{minipage}[t]{\textwidth}
\begin{center}
\includegraphics[height=2.5in]{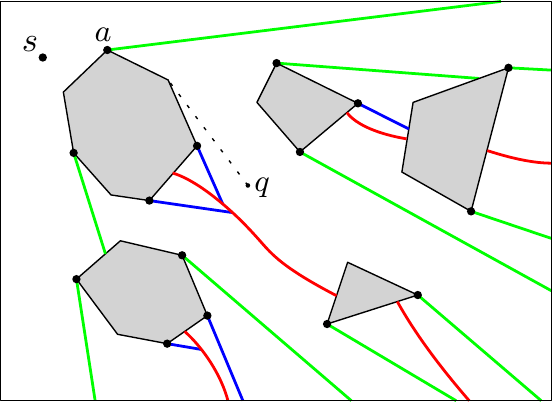}
\caption{\footnotesize Illustrating the map $\spm'(s)$. The red curves are non-extension bisectors. The green and blue segments are extension bisectors, where the green and blue ones are Type~(a) and (b) as illustrated in Fig.~\ref{fig:newgenerator}. The predecessor in each cell is also shown with a black point. For example, all points in the cell containing $q$ have $a$ as their predecessor. }
\label{fig:map}
\end{center}
\end{minipage}
\end{figure}

\begin{figure}[h]
\begin{minipage}[t]{\textwidth}
\begin{center}
\includegraphics[height=2.5in]{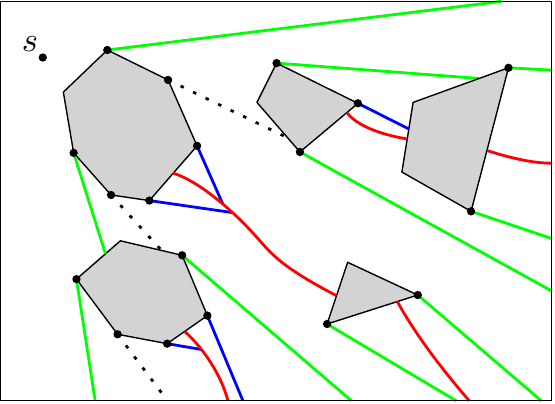}
\caption{\footnotesize Illustrating the map $\spm(s)$. The dashed segments are windows that are not in $\spm'(s)$; removing them becomes $\spm'(s)$. The anchor of each cell is also shown.}
\label{fig:fullmap}
\end{center}
\end{minipage}
\vspace{-0.15in}
\end{figure}

The following lemma shows that $\spm'(s)$ is very similar to $\spm(s)$. Refer to Fig.~\ref{fig:fullmap} for $\spm(s)$.

\begin{lemma}\label{lem:spm}
Each extension bisector edge of $\spm'(s)$ is a window of $\spm(s)$. The union of all non-extension bisector edges of $\spm'(s)$ is exactly the union of all walls of $\spm(s)$. $\spm'(s)$ can be obtained from $\spm(s)$ by removing all windows except those that are extension bisectors of $\spm'(s)$.
\end{lemma}
\begin{proof}
We first prove that each extension bisector edge $e$ of $\spm'(s)$ is a window of $\spm(s)$.
By definition, one endpoint of $e$, denoted by $u$, is an obstacle vertex, such that $e$ is a half-line extended from $u$ in the direction from $v$ to $u$ for another obstacle vertex $v$. Hence, for each point $p\in e$, there is a shortest $s$-$p$ path $\pi(s,p)$ such that $\overline{vp}$ is a segment of $\pi(s,p)$. As $u\in \overline{vp}$, $p$ must be on a window $w$ of $\spm(s)$. This proves that $e\subseteq w$. On the other hand, for any point $p\in w$, there is a shortest $s$-$p$ path $\pi(s,p)$ that contains $\overline{vp}$. Since $u\in \overline{vp}$, $p$ must be on the extension bisector edge $e$ of $\spm'(s)$, and thus $w\subseteq e$. Therefore, $e=w$. This proves that each extension bisector edge of $\spm'(s)$ is a window of $\spm(s)$.

We next prove that the union of all non-extension bisector edges of $\spm'(s)$ is exactly the union of all walls of $\spm(s)$.

\begin{itemize}
\item
We first show that the union of all non-extension bisector edges of $\spm'(s)$ is a subset of the union of all walls of $\spm(s)$.

Let $q$ be a point on a non-extension bisector $B(\alpha,\alpha')$ of two generators $\alpha=(A,a)$ and $\alpha=(A',a')$. Let $v$ and $v'$ be the tangent points on $A$ and $A'$ from $q$, respectively.
Hence, $q$ has two shortest paths from $s$, one containing $\overline{qv}$ and the other containing $\overline{qv'}$. Since $B(\alpha,\alpha')$ is not an extension bisector, $v'\not\in \overline{qv}$ and $v\not\in\overline{qv'}$. This means that the two paths are combinatorially different and thus $q$ is on a wall of $\spm(s)$. This proves that the union of all non-extension bisector edges of $\spm'(s)$ is a subset of the union of all walls of $\spm(s)$.

\item
We then show that the union of all walls of $\spm(s)$ is a subset of the union of all non-extension bisector edges of $\spm'(s)$.

Let $q$ be a point on a wall of $\spm(s)$. By definition, there are two obstacle vertices $v$ and $v'$ such that there are two combinatorially different shortest paths from $s$ to $q$ whose anchors are $v$ and $v'$, respectively. Further, since the two paths are combinatorially different and also due to the general position assumption, $v'\not\in \overline{qv}$ and $v\not\in \overline{qv'}$. Therefore, $q$ must be on the bisector edge of the two generators $\alpha$ and $\alpha'$ in $\spm'(s)$ whose underlying chains containing $v$ and $v'$, respectively. Further, since $v\not\in \overline{qv'}$ and $v'\not\in \overline{qv}$, the bisector of $\alpha$ and $\alpha'$ is not an extension bisector. This proves that the union of all walls of $\spm(s)$ is a subset of the union of all non-extension bisector edges of $\spm'(s)$.
\end{itemize}

The above proves that the first and second statements of the lemma. The third statement immediately follows the first two.
\end{proof}

\begin{corollary}\label{coro:size}
The combinatorial size of $\spm'(s)$ is $O(n)$.
\end{corollary}
\begin{proof}
By Lemma~\ref{lem:spm}, $\spm'(s)$ is a subset of $\spm(s)$. As the combinatorial size of $\spm(s)$ is $O(n)$~\cite{ref:HershbergerAn99,ref:MitchellA91}, the combinatorial size of $\spm'(s)$ is also $O(n)$.
\end{proof}

\begin{lemma}\label{lem:sizespm}
The subdivision $\spm'(s)$ has $O(h)$ faces, vertices, and edges. In particular, $\spm'(s)$ has $O(h)$ bisector intersections and $O(h)$ intersections between bisectors and obstacle edges.
\end{lemma}
\begin{proof}
By Lemma~\ref{lem:spm}, for any vertex of $\spm'(s)$ that is the intersection of two non-exenstion bisectors, it is also the intersection of two walls of $\spm(s)$, which is a {\em triple point}.
Let $\spm''(s)$ refer to $\spm'(s)$ excluding all extension bisectors.
We define a planar graph $G$ corresponding to $\spm''(s)$ as follows. Each obstacle defines a node of $G$ and each triple point also defines a node of $G$. For any two nodes of $G$, if they are connected by a chain of bisector edges of $\spm'(s)$ such that the chain does not contain any other node of $G$, then $G$ has an edge connecting the two nodes.
It is proved in~\cite{ref:WangQu19} that $G$ has $O(h)$ vertices, faces, and edges.

It is not difficult to see that a face of $\spm''(s)$ corresponds to a face of $G$ and thus the number of faces of $\spm''(s)$ is $O(h)$. For each bisector intersection of $\spm''(s)$, it must be a triple point and thus it is also a node of $G$. As $G$ has $O(h)$ nodes, the number of bisector intersections of $\spm''(s)$ is $O(h)$.
For each intersection $v$ between a bisector and an obstacle edge in $\spm''(s)$, $G$ must has an edge $e$ corresponding a chain of bisector edges and $v$ is the endpoint of the chain; we charge $v$ to $e$. As $e$ as two incident nodes of $G$, $e$ can be charged at most twice. Since $G$ has $O(h)$ edges, the number of intersections between bisectors and obstacle edges in $\spm''(s)$ is $O(h)$.

We next prove that the total number of extension bisector edges of $\spm'(s)$ is $O(h)$. For each extension bisection edge $e$, it belongs to one of the two cases $(a)$ and $(b)$ illustrated in Fig.~\ref{fig:newgenerator}. For Case~(b), one endpoint of $e$ is a rectilinear extreme vertex. Since each rectilinear extreme vertex can define at most two extension bisector edges and there are $O(h)$ rectilinear extreme vertices, the total number of Case (b) extension bisectors is $O(h)$. For Case (a), $e$ is an extension of a common tangent of two obstacles; we say that $e$ is {\em defined} by the common tangent. Note that all such common tangents are interior disjoint as they belong to shortest paths from $s$ encoded by $\spm'(s)$. We now define a planar graph $G'$ as follows. The node set of $G'$ consists of all obstacles of $\calP$. Two obstacles have an edge in $G'$ if they have at least one common tangent that defines an extension bisector. Since all such common tangents are interior disjoint, $G'$ is a planar graph. Apparently, no two nodes of $G'$ share two edges and no node of $G'$ has a self-loop. Therefore, $G'$ is a simple planar graph. Since $G'$ has $h$ vertices, the number of edges of $G'$ is $O(h)$. By the definition of $G'$, each pair of obstacles whose common tangents define extension bisectors have an edge in $G'$. Because each pair of obstacles can define at most four common tangents and thus at most four extension bisectors and also because $G'$ has $O(h)$ edges, the total number of Case (a) extension bisectors is $O(h)$.

The bisector intersections of $\spm'(s)$ that are not vertices of $\spm''(s)$ involve extension bisectors of $\spm'(s)$. The intersections between bisectors and obstacle edges in $\spm'(s)$ that are not in $\spm''(s)$ also involve extension bisectors. Since all extension bisectors and all edges of $\spm''(s)$ are interior disjoint, each extension bisector can involve in at most two bisector intersections and at most two intersections between bisectors and obstacle edges. Because the total number of extension bisector edges of $\spm'(s)$ is $O(h)$, the number of bisector intersections involving extension bisectors is $O(h)$ and the number of intersections between extension bisectors and obstacle edges is also $O(h)$.
Therefore, comparing to $\spm''(s)$, $\spm'(s)$ has $O(h)$ additional vertices and $O(h)$ additional edges. Note that the number of convex-chain edges of $\spm'(s)$ is $O(h)$, for $\spm'(s)$ has $O(h)$ vertices.
The lemma thus follows.
\end{proof}


\begin{corollary}\label{coro:complexity}
There are $O(h)$ bisector events and $O(h)$ generators in $\spm'(s)$.
\end{corollary}
\begin{proof}
Each bisector event is either a bisector intersection or an intersection between a bisector and an obstacle edge. By Lemma~\ref{lem:sizespm}, there are $O(h)$ bisector intersections and $O(h)$ intersections between bisectors and obstacle edges. As such, there are $O(h)$ bisector events in $\spm'(s)$.

By definition, each face of $\spm'(s)$ has a unique generator.
Since $\spm'(s)$ has $O(h)$ faces by Lemma~\ref{lem:sizespm}, the total number of generators is $O(h)$.
\end{proof}

By the definition of $\spm'(s)$, each cell $C$ of $\spm'(s)$ has a unique generator $\alpha=(A,a)$, all points of the cell are reachable from the generator, and $a$ is the predecessor of all points of $C$ (e.g., see Fig.~\ref{fig:map}; all points in the cell containing $q$ have $a$ as their predecessor).
Hence, for any point $q\in C$, we can compute $d(s,q)=d(\alpha,q)$ by computing the tangent from $q$ to $A$. Thus, $\spm'(s)$ can also be used to answer shortest path queries. In face, given $\spm'(s)$, we can construct $\spm(s)$ in additional $O(n)$ time by inserting the windows of $\spm(s)$ to $\spm'(s)$, as shown in the lemma below.

\begin{lemma}\label{lem:spmconstruction}
Given $\spm'(s)$, $\spm(s)$ can be constructed in $O(n)$ time (e.g., see Fig.~\ref{fig:map} and Fig.~\ref{fig:fullmap}).
\end{lemma}
\begin{proof}
It suffices to insert all windows of $\spm(s)$ (except those that are also extension bisectors of $\spm'(s)$) to $\spm'(s)$. We consider each cell $C$ of $\spm'(s)$ separately. Let $\alpha=(A,a)$ be the generator of $C$. 
Our goal is to extend each obstacle edge of $A$ along the designated direction of $A$ until the boundary of $C$. To this end, since all points of $C$ are reachable from $\alpha$, the cell $C$ is ``star-shaped'' with respect to the tangents of $A$ (along the designated direction) and
the extension of each obstacle edge of $A$ intersects the boundary of $C$ at most once. Hence, the order of the endpoints of these extensions is consistent with the order of the corresponding edges of $A$. Therefore,  these extension endpoints can be found in order by traversing the boundary of $C$, which takes linear time in the size of $C$. Since the total size of all cells of $\spm'(s)$ is $O(n)$ by Corollary~\ref{coro:size}, the total time of the algorithm is $O(n)$.
\end{proof}

In light of Lemma~\ref{lem:spmconstruction}, we will focus on computing the map $\spm'(s)$.

\subsection{The wavefront expansion algorithm}
\label{sec:algo}

To simulate the wavefront expansion, we compute the wavefront passing
through each transparent edge of the conforming subdivision $\calS'$.
As in the HS algorithm, since it is difficult to compute a true
wavefront for each transparent edge $e$ of $\calS'$, a key idea
is to compute two one-sided wavefronts (called {\em
approximate wavefronts}) for $e$, each representing the wavefront
coming from one side of $e$. Intuitively, an approximate wavefront
from one side of $e$ is what the true wavefront would be if the
wavefront were blocked off at $e$ by considering $e$ as an
opaque segment (with open endpoints).

In the following, unless otherwise stated, a wavefront at a
transparent edge $e$ of $\calS'$ refers to an approximate wavefront.
Let $W(e)$ denote a wavefront at $e$. To make the description concise,
as there are two wavefronts at $e$, depending on the context, $W(e)$
may refer to both wavefronts, i.e., the discussion on $W(e)$ applies
to both wavefronts. For example, ``compute the wavefronts $W(e)$''
means ``compute both wavefronts at $e$''.

For each transparent edge $e$ of $\calS'$,
define $input(e)$ as the set of transparent edges on the boundary of
the well-covering region $\calU(e)$, and define $output(e)=\{g\ |\ e\in
input(g)\}\cup input(e)$\footnote{We include $input(e)$ in $output(e)$
mainly for the proof of Lemma~\ref{lem:markcell}, which relies on
$output(e)$ having a cycle enclosing $e$.}. Because $\partial
\calU(e')$ for each transparent edge $e'$ of $\calS'$ has $O(1)$
transparent edges, both $|input(e)|$ and $|output(e)|$ are $O(1)$.

\paragraph{The wavefront propagation and merging  procedures.}
The wavefronts $W(e)$ at $e$ are computed from the wavefronts at edges of
$input(e)$; this guarantees the correctness because $e$ is in $\calU(e)$ (and thus
any shortest path $\pi(s,p)$ must cross some edge $f\in input(e)$ for any point $p\in e$).
After the wavefronts $W(e)$ at $e$ are computed, they will pass
to the edges of $output(e)$. Also, the geodesic distances
from $s$ to both endpoints of $e$ will be computed. Recall that
$\calV$ is the set of rectilinear extreme vertices of all obstacles
and each
vertex of $\calV$ is incident to a transparent edge of $\calS'$. As such,
after the algorithm is finished, geodesic distances from $s$ to all
vertices of $\calV$ will be available. The process of passing the
wavefronts $W(e)$ at $e$ to all edges $g\in output(e)$ is called the
{\em wavefront propagation procedure}, which will compute the
wavefront $W(e,g)$, where $W(e,g)$ is
the portion of $W(e)$ that passes to $g$ through the well-covering
region $\calU(g)$ of $g$ if $e\in input(g)$ and through $\calU(e)$
otherwise (in this case $g\in input(e)$);
whenever the procedure is invoked on $e$, we say that $e$ is {\em
processed}. The wavefronts $W(e)$ at $e$ are constructed by
merging the wavefronts $W(f,e)$ for edges $f\in input(e)$; this procedure is called the
{\em wavefront merging procedure}.

\paragraph{The main algorithm.}
The transparent edges of $\calS'$ are processed in a rough time order.
The wavefronts $W(e)$ of each transparent edge $e$ are constructed at the time $\td(s,e)+|e|$, where $\td(s,e)$ is the minimum geodesic distance from $s$ to the two endpoints of $e$.
Define $covertime(e)=\td(s,e)+|e|$. The value $covertime(e)$ will be
computed during the algorithm. Initially, for each edge $e$ whose
well-covering region $\calU(e)$ contains $s$, $W(e)$ and $covertime(e)$ are computed
directly (and set $covertime(e)=\infty$ for all other edges);
we refer to this step as the {\em initialization step}, which will be elaborated below.
The algorithm maintains a timer $\tau$ and processes the
transparent edges $e$ of $\calS'$ following the order of
$covertime(e)$.

The main loop of the algorithm works as follows.
As long as $\calS'$ has an unprocessed transparent edge, we do the following.
First, among all unprocessed transparent edges, choose the one $e$
with minimum $covertime(e)$ and set $\tau=covertime(e)$. Second, call
the
wavefront merging procedure to construct the wavefronts $W(e)$ from
$W(f,e)$ for all edges $f\in input(e)$ satisfying
$covertime(f)<covertime(e)$; compute $d(s,v)$ from $W(e)$ for each
endpoint $v$ of $e$. Third, process $e$, i.e., call the wavefront
propagation procedure on $W(e)$ to compute $W(e,g)$ for all edges
$g\in output(e)$; in particular, compute the time $\tau_g$ when the
wavefronts $W(e)$ first encounter an endpoint of $g$ and set
$covertime(g)= \min\{covertime(g),\tau_g+|g|\}$.

The details of the wavefront merging procedure and the wavefront propagation procedure will be described in Section~\ref{sec:merge} and Section~\ref{sec:propagation}, respectively.

\paragraph{The initialization step.} We provide the details on the initialization step. Consider a transparent edge $e$ whose well-covering region $\calU(e)$ contains $s$. To compute $W(e)$, we only consider straight-line paths from $s$ to $e$ inside $\calU(e)$. If $\calU(e)$ does not have any island inside, then the points of $e$ that can be reached from $s$ by straight-line paths in $\calU(e)$ form an interval of $e$ and the interval can be computed by considering the tangents from $s$ to the elementary chains on the boundary of $\calU(e)$. Since $\calU(e)$ has $O(1)$ cells of $\calS'$, $\partial \calU(e)$ has $O(1)$ elementary chain fragments and thus computing the interval on $e$ takes $O(\log n)$ time. If $\calU(e)$ has at least one island inside, then $e$ may have multiple such intervals. As $\calU(e)$ is the union of $O(1)$ cells of $\calS'$, the number of islands inside $\calU(e)$ is $O(1)$. Hence, $e$ has $O(1)$ such intervals, which can be computed in $O(\log n)$ time altogether. These intervals form the wavefront $W(e)$, i.e., each interval corresponds to a wavelet with generator $s$. From $W(e)$, the value $covertime(e)$ can be immediately determined. More specifically, for each endpoint $v$ of $e$, if one of the wavelets of $W(e)$ covers $v$, then the segment $\overline{sv}$ is in $\calU(e)$ and thus $d(s,v)=|\overline{sv}|$; otherwise, we set $d(s,v)=\infty$. In this way, we find an upper bound for $\td(s,e)$ and set $covertime(e)$ to this upper bound plus $|e|$.


\paragraph{The algorithm correctness.}
At the time $\td(s,e)+|e|$, all edges $f\in input(e)$ whose wavefronts
$W(f)$ contribute a wavelet to $W(e)$ must have already been
processed. This is due to the property of the well-covering regions of
$\calS'$ that $d(e,f)\geq 2\cdot \max\{|e|,|f|\}$ since $f$ is on
$\partial\calU(e)$. The proof is the same as that of
Lemma~4.2~\cite{ref:HershbergerAn99}, so we omit it. Note that this
also implies that the geodesic distance $d(s,v)$ is correctly computed
for each endpoint $v$ of $e$. Therefore, after the algorithm is
	finished, geodesic distances from $s$ to endpoints of all
	transparent edges of $\calS'$ are correctly computed.

\paragraph{Artificial wavelets.}
As in the HS algorithm, to limit the interaction between wavefronts from different sides of each transparent edge $e$, when a wavefront propagates across $e$, i.e., when $W(e)$ is computed in the wavefront merging procedure, an {\em artificial wavelet} is generated at each endpoint $v$ of $e$, with weight $d(s,v)$. This is to eliminate a wavefront from one side of $e$ if it arrives at $e$ later than the wavefront from the other side of $e$.

\paragraph{Topologically different paths.}
In the wavefront propagation procedure to pass $W(e)$ to all edges $g\in output(e)$, $W(e)$ travels through the cells of the well-covering region $\calU$ of either $g$ or $e$. Since $\calU$ may not be simply connected (e.g., the square-annulus), there may be multiple topologically different shortest paths between $e$ and $g$ inside $\calU$; the number of such paths is $O(1)$ as $\calU$ is the union of $O(1)$ cells of $\calS'$. We propagate $W(e)$ in multiple components of $\calU$, each corresponding to a topologically different shortest path and defined by the particular sequence of transparent edges it crosses. These wavefronts are later combined in the wavefront merging step at $g$. This also happens in the initialization step, as discussed above.

\paragraph{Claiming a point.}
During the wavefront merging procedure at $e$, we have a set of wavefronts $W(f,e)$ that reach $e$ from the same side for edges $f\in input(e)$. We say that a wavefront $W(f,e)$ {\em claims} a point $p\in e$ if $W(f,e)$ reaches $p$ before any other wavefront. Further, for each wavefront $W(f,e)$, the set of points on $e$ claimed by it forms an interval (the proof is similar to that of Lemma~4.4~\cite{ref:HershbergerAn99}).
Similarly, a wavelet of a wavefront {\em claims} a point $p$ of $e$ if the wavelet reaches $p$ before any other wavelet of the wavefront, and the set of points on $e$ claimed by any wavelet of the wavefront also forms an interval. For convenience, if a wavelet claims a point, we also say that the generator of the wavelet claims the point.

\bigskip

Before presenting the details of the wavefront merging procedure and the wavefront propagation procedure in the next two subsections, we first discuss the data structure (for representing elementary chains, generators, and wavefronts) and a monotonicity property of bisectors, which will be used later in our algorithm.

\subsubsection{The data structure}
We use an array to represent each elementary chain. Then, a generator $\alpha=(A,a)$ can be represented by recording the indices of the two end vertices of its underlying chain $A$. In this way, a generator takes $O(1)$ additional space to record and binary search on $A$ can be supported in $O(\log n)$ time (e.g., given a point $p$, find the tangent from $p$ to $A$).
In addition, we maintain the lengths of the edges in each elementary chain so that given any two vertices of the chain, the length of the sub-chain between the two vertices can be obtained in $O(1)$ time. All these preprocessing can be done in $O(n)$ time and space for all elementary chains.

For a wavefront $W(e)$ of one side of $e$, it is a list of generators $\alpha_1,\alpha_2,\ldots$ ordered by the intervals of $e$ claimed by these generators. Note that these generators are in the same side of $e$. Formally, we say that a generator $\alpha$ is in one side of $e$ if the initial vertex of $\alpha$ lies in that side of the supporting line of $e$.
We maintain these generators by a balanced binary search tree so that the following operations can be supported in logarithmic time each. Let $W$ be a wavefront with generators $\alpha_1,\alpha_2,\ldots,\alpha_k$.

\begin{description}
\item[Insert] Insert a generator $\alpha$ to $W$.
In our algorithm, $\alpha$ is inserted either in the front of $\alpha_1$ or after $\alpha_k$.

\item[Delete] Delete a generator $\alpha_i$ from $W$, for any $1\leq i\leq k$.

\item[Split] Split $W$ into two sublists at some generator $\alpha_i$ so that the first $i$ generators form a wavefront and the rest form another wavefront.

\item[Concatenate] Concatenate $W$ with another list $W'$ of generators so that all generators of $W$ are in the front of those of $W'$ in the new list.
\end{description}

We will show later that each wavefront involved in our algorithm has $O(h)$ generators. Therefore, each of the above operation can be performed in $O(\log h)$ time.
We make the tree fully persistent by path-copying~\cite{ref:DriscollMa89}. In this way, each operation on the tree will cost $O(\log h)$ additional space but the older version of the tree will be kept intact (and operations on the old tree can still be supported).




\subsubsection{The monotonicity property of bisectors}

\begin{lemma}\label{lem:intersection}
Suppose $\alpha_1=(A_1,a_1)$ and $\alpha_2=(A_2,a_2)$ are two generators on the same side of an axis-parallel line $\ell$ (i.e., $a_1$ and $a_2$ are on the same side of $\ell$; it is possible that $\ell$ intersects $A_1$ and $A_2$). Then, the bisector $B(\alpha_1,\alpha_2)$ intersects $\ell$ at no more than one point.
\end{lemma}
\begin{proof}
Consider a generator $\alpha=(A,a)$. Recall that $a$ is an endpoint of $A$. Let $a'$ be the other endpoint of $A$.
We define $R(\alpha)$ as the set of points in the plane that are reachable from $\alpha$ (without considering any other obstacles), and we call it the {\em reachable region} of $\alpha$.
The reachable region $R(\alpha)$ is bounded by $A$, a ray $\rho(a)$ with origin $a$, and another ray $\rho(a')$ with origin $a'$ (e.g., see Fig.~\ref{fig:reachableregion}).
Specifically, $\rho(a)$ is along the direction from $z$ to $a$, where $z$ is the anchor of $a$ in the shortest path from $s$ to $a$;
$\rho(a')$ is the ray from the adjacent obstacle vertex of $a'$ on $A'$ to $a'$.

We claim that the intersection $\ell\cap R(\alpha)$ must be an interval of $\ell$. Indeed, since $A$ is $xy$-monotone, without loss of generality, we assume that $a'$ is to the northwest of $a$. Let $\rho'(a)$ be the vertically downward ray from $a$ (e.g., see Fig.~\ref{fig:reachableregion10}). Observe that the concatenation of $\rho(a')$, $A$, and $\rho'(a)$ is $xy$-monotone; let
$R'$ be the region of the plane to the right of their concatenation. Since the boundary of $R'$ is $xy$-monotone and $\ell$ is axis-parallel, $R'\cap \ell$ must be an interval of $\ell$. Notice that $\rho(a)$ must be in $R'$ and thus it partitions $R'$ into two subregions, one of which is $R(\alpha)$. Since $\ell$ intersects $\rho(a)$ at most once and $R'\cap \ell$ is an interval, $R(\alpha)\cap \ell$ must also be an interval.

\begin{figure}[t]
\begin{minipage}[t]{0.49\textwidth}
\begin{center}
\includegraphics[height=1.2in]{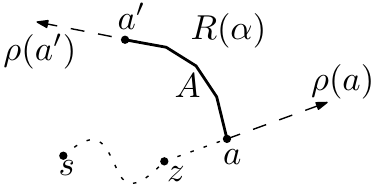}
\caption{\footnotesize Illustrating the reachable region $R(\alpha)$ of a generator $\alpha=(A,a)$.}
\label{fig:reachableregion}
\end{center}
\end{minipage}
\hspace{0.02in}
\begin{minipage}[t]{0.49\textwidth}
\begin{center}
\includegraphics[height=1.5in]{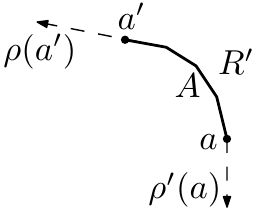}
\caption{\footnotesize Illustrating the ray $\rho'(a)$ and the region $R'$.}
\label{fig:reachableregion10}
\end{center}
\end{minipage}
\vspace{-0.15in}
\end{figure}

According to the above claim, both $\ell\cap R(\alpha_1)$ and $\ell\cap R(\alpha_2)$ are intervals of $\ell$.
Let $I$ denote the common intersection of the two intervals. Since $B(\alpha_1,\alpha_2)\subseteq R(\alpha_1)$ and $B(\alpha_1,\alpha_2)\subseteq R(\alpha_2)$, we obtain that $B(\alpha_1,\alpha_2)\cap \ell\subseteq I$.

\begin{figure}[h]
\begin{minipage}[t]{\textwidth}
\begin{center}
\includegraphics[height=1.4in]{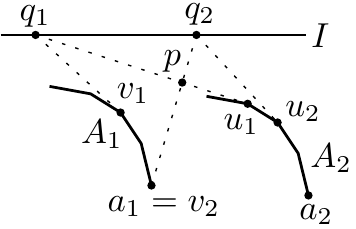}
\caption{\footnotesize Illustrating the proof of Lemma~\ref{lem:intersection}.}
\label{fig:monotone}
\end{center}
\end{minipage}
\vspace{-0.15in}
\end{figure}

Assume to the contrary that $B(\alpha_1,\alpha_2)\cap \ell$ contains two points, say, $q_1$ and $q_2$. Then, both points are in $I$. Let $v_1$ and $u_1$ be the tangents points of $q_1$ on $A_1$ and $A_2$, respectively.
Let $v_2$ and $u_2$ be the tangents points of $q_2$ on $A_1$ and $A_2$, respectively (e.g., see Fig.~\ref{fig:monotone}). As $I=\ell\cap R(\alpha_1)\cap R(\alpha_2)$ and $I$ contains both $q_1$ and $q_2$, if we move a point $q$ from $q_1$ to $q_2$ on $\ell$, the tangent from $q$ to $A_1$ will continuously change from $\overline{q_1v_1}$ to $\overline{q_2v_2}$ and the tangent from $q$ to $A_2$ will continuously change from $\overline{q_1u_1}$ to $\overline{q_2u_2}$. Since $A_1$ and $A_2$ are in the same side of $\ell$, either $\overline{q_1u_1}$ intersects $\overline{q_2v_2}$ in their interiors or $\overline{q_1v_1}$ intersects $\overline{q_2u_2}$ in their interiors; without loss of generality, we assume that it is the former case. Let $p$ be the intersection of $\overline{q_1u_1}$ and $\overline{q_2v_2}$ (e.g., see Fig.~\ref{fig:monotone}). Since $q_1\in B(\alpha_1,\alpha_2)$, points of $\overline{q_1u_1}$ other than $q_1$ have only one predecessor, which is $a_2$. As $p\in \overline{q_1u_1}$ and $p\neq q_1$, $p$ has only one predecessor $a_2$. Similarly, since $q_2\in B(\alpha_1,\alpha_2)$ and $p\in \overline{q_2v_2}$, $a_1$ is also $p$'s predecessor. We thus obtain a contradiction.
\end{proof}


\begin{corollary}\label{coro:monotone}
Suppose $\alpha_1=(A_1,a_1)$ and $\alpha_2=(A_2,a_2)$ are two generators both below a horizontal line $\ell$. Then, the portion of the bisector $B(\alpha_1,\alpha_2)$ above $\ell$ is $y$-monotone.
\end{corollary}
\begin{proof}
For any horizontal line $\ell'$ above $\ell$, since both generators are below $\ell$, they are also below $\ell'$. By Lemma~\ref{lem:intersection}, $B(\alpha_1,\alpha_2)\cap \ell'$ is either empty or a single point. The corollary thus follows.
\end{proof}

\subsection{The wavefront merging procedure}
\label{sec:merge}

In this section, we present the details of the wavefront merging procedure. Given all contributing wavefronts $W(f,e)$ of $f\in input(e)$ for $W(e)$, the goal of the procedure is to compute $W(e)$.
The algorithm follows the high-level scheme of the HS algorithm (i.e., Lemma 4.6~\cite{ref:HershbergerAn99}) but the implementation details are quite different.


We only consider the wavefronts $W(f,e)$ and $W(e)$ for one side of $e$ since the algorithm for the other side is analogous. Without loss of generality, we assume that $e$ is horizontal and all wavefronts $W(f,e)$ are coming from below $e$. We describe the algorithm for computing the interval of $e$ claimed by $W(f,e)$ if only one other wavefront $W(f',e)$ is present. The common intersection of these intervals of all such $f'$ is the interval of $e$ claimed by $W(f,e)$. Since $|input(e)|=O(1)$, the number of such $f'$ is $O(1)$.


We first determine whether the claim of $W(f,e)$ is to the left or right of that of $W(f',e)$. To this end,
depending on whether both $W(f,e)$ and $W(f',e)$ reach the left endpoint $v$ of $e$, there are two cases. Note that the intervals of $e$ claimed by $W(f,e)$ and $W(f',e)$ are available from $W(f,e)$ and $W(f',e)$; let $I_f$ and $I_{f'}$ denote these two intervals, respectively.

\begin{itemize}
\item
If both $I_f$ and $I_{f'}$ contain $v$, i.e., both $W(f,e)$ and $W(f',e)$ reach $v$, then we compute the (weighted) distances from $v$ to the two wavefronts. This can be done as follows. Since $v\in I_f$, $v$ must be reached by the leftmost generator $\alpha$ of $W(f,e)$. We compute the distance $d(\alpha,v)$ by computing the tangent from $v$ to $\alpha$ in $O(\log n)$ time. Similarly, we compute $d(\alpha',v)$, where $\alpha'$ is leftmost generator of $W(f',e)$. If $d(\alpha, v)\leq d(\alpha',v)$, then the claim of $W(f,e)$ is to the left of that of $W(f',e)$; otherwise, the claim of $W(f,e)$ is to the right of that of $W(f',e)$.

\item
If not both $I_f$ and $I_{f'}$ contain $v$, then the order of the left endpoints of $I_f$ and $I_{f'}$ will give the answer.
\end{itemize}

Without loss of generality, we assume that the claim of $W(f,e)$ is to the left of that of $W(f',e)$. We next compute the interval $I$ of $e$ claimed by $W(f,e)$ with respect to $W(f',e)$. Note that the left endpoint of $I$ is the left endpoint of $I_f$. Hence, it remains to find the right endpoint of $I$, as follows.

Let $\ell_e$ be the supporting line of $e$.
Let $\alpha$ be the rightmost generator of $W(f,e)$ and let $\alpha'$ be the leftmost generator of $W(f',e)$.
Let $q_1$ be the left endpoint of the interval on $e$ claimed by $\alpha$ in $W(f,e)$, i.e., $q_1$ is the intersection of the bisector $B(\alpha_1,\alpha)$ and $e$, where $\alpha_1$ is the left neighboring generator of $\alpha$ in $W(f,e)$. Similarly, $q_2$ be the right endpoint of the interval on $e$ claimed by $\alpha'$ in $W(f',e)$, i.e., $q_2$ is the intersection of $e$ and the bisector $B(\alpha',\alpha_1')$, where $\alpha_1'$ is the right neighboring generator of $\alpha'$ in $W(f',e)$.
Let $q_0$ be the intersection of the bisector $B(\alpha,\alpha')$ and $e$. We assume that the three points $q_i$, $i=0,1,2$ are available and we will discuss later that each of them can be computed in $O(\log n)$ time by a {\em bisector-line intersection operation} given in Lemma~\ref{lem:bl-intersection}. If $q_0$ is between $q_1$ and $q_2$, then $q_0$ is the right endpoint of $I$ and we can stop the algorithm. If $q_0$ is to the left of $q_1$, then we delete $\alpha$ from $W(f,e)$. If $q_0$ is to the right of $q_2$, then we delete $\alpha'$ from $W(f',e)$. In either case, we continue the same algorithm by redefining $\alpha$ or $\alpha'$ (and recomputing $q_i$ for $i=0,1,2$).

Clearly, the above algorithm takes $O((1+k)\log n)$ time, where $k$ is the number of generators that are deleted. We apply the algorithm on $f$ and other $f'$ in $input(e)$ to compute the corresponding intervals for $f$. The common intersection of all these intervals is the interval of $e$ claimed by $W(f,e)$. We do so for each $f\in input(e)$, after which $W(e)$ is obtained. Since the size of $input(e)$ is $O(1)$, we obtain the following lemma.

\begin{lemma}\label{lem:merge}
Given all contributing wavefronts $W(f,e)$ of edges $f\in input(e)$ for $W(e)$, we can compte the interval of $e$ claimed by each $W(f,e)$ and thus construct $W(e)$ in $O((1+k)\log n)$ time, where $k$ is the total number of generators in all wavefronts $W(f,e)$ that are absent from $W(e)$.
\end{lemma}


\begin{lemma}{\em\bf (Bisector-line intersection operation)}\label{lem:bl-intersection}
Each bisector-line intersection operation can be performed in $O(\log n)$ time.
\end{lemma}
\begin{proof}
Given two generators $\alpha_1=(A_1,a_1)$ and $\alpha_2=(A_2,a_2)$ below a horizontal line $\ell$, the goal of the operation is to compute the intersection between the bisector $B(\alpha_1,\alpha_2)$ and $\ell$. By Lemma~\ref{lem:intersection}, $B(\alpha_1,\alpha_2)\cap \ell$ is either empty or a single point.


We apply a prune-and-search technique from Kirkpatrick and Snoeyink~\cite{ref:KirkpatrickTe95}. To avoid the lengthy background explanation, we follow the notation in~\cite{ref:KirkpatrickTe95} without definition.
We will rely on Theorem~3.6 in~\cite{ref:KirkpatrickTe95}. To do so,
we need to define a decreasing function $f$ and an increasing function $g$.



We first compute the intersection of $\ell$ and the reachable region $R(\alpha_1)$ of $\alpha_1$, as defined in the proof of Lemma~\ref{lem:intersection}. This can be easily done by computing the intersections between $\ell$ and the boundary of $R(\alpha_1)$ in $O(\log n)$ time. Let $I_1$ be their intersection, which is an interval of $\ell$ as proved in Lemma~\ref{lem:intersection}. Similarly, we compute the intersection $I_2$ of $\ell$ and the reachable region $R(\alpha_2)$ of $\alpha_2$. Let $I=I_1\cap I_2$.
For each endpoint of $I$, we compute its tangent point on $A_1$ (along its designated direction) to determine the portion of $A_1$ whose tangent rays intersect $I$ and let $A$ denote the portion.  Similarly, we determine the portion $B$ of $A_2$ whose tangent rays intersect $I$. All above takes $O(\log n)$ time.

\begin{figure}[t]
\begin{minipage}[t]{\textwidth}
\begin{center}
\includegraphics[height=1.7in]{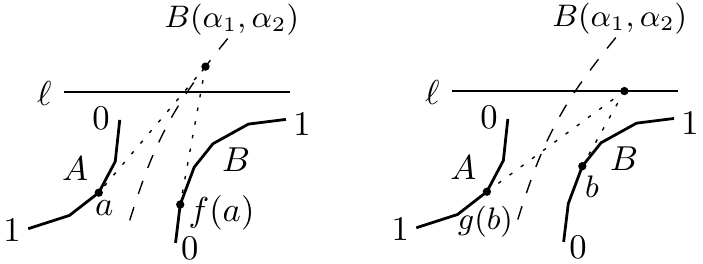}
\caption{\footnotesize Illustrating the definitions of $f(a)$ and $g(b)$.}
\label{fig:bisectorline}
\end{center}
\end{minipage}
\vspace{-0.15in}
\end{figure}

We parameterize over $[0,1]$ each of the two convex chains $A$ and
$B$ in clockwise order, i.e., each value of $[0,1]$ corresponds to a slope of a tangent at a point on
$A$ (resp., $B$). For each point $a$ of $A$, we define $f(a)$ to be the parameter of the point $b\in B$ such that the tangent ray of $A$ at $a$ along the designated direction and the tangent ray of $B$ at $b$ along the designated direction intersect at a point on the bisector $B(\alpha_1,\alpha_2)$ (e.g., see Fig.~\ref{fig:bisectorline} left); if the tangent ray at $a$ does not intersect $B(\alpha_1,\alpha_2)$, then define $f(a)=1$. For each point $b$ of $B$, we define $g(b)$ to be the parameter of the point $a\in A$ such that the tangent ray of $A$ at $a$ and the tangent ray of $B$ at $b$ intersect at a point on the line $\ell$ (e.g., see Fig.~\ref{fig:bisectorline} right).
One can verify that $f$ is a continuous decreasing function while $g$ is a continuous increasing function (the tangent at an obstacle vertex of $A$ and $B$ is not unique but the issue can be handled~\cite{ref:KirkpatrickTe95}). The fixed-point of the
composition of the two functions $g \cdot f$ corresponds to
the intersection of $\ell$ and $B(\alpha_1,\alpha_2)$,
which can be computed by applying the
prune-and-search algorithm of Theorem 3.6~\cite{ref:KirkpatrickTe95}.

As both chains $A$ and $B$ are represented by arrays (of size $O(n)$),
to show that the algorithm can be implemented in $O(\log n)$ time, it suffices to show that given
any $a\in A$ and any $b\in B$, we can determine whether $f(a)\geq b$ in $O(1)$
time and determine whether $g(b)\geq a$ in $O(1)$ time.

To determine whether $f(a)\geq b$, we do the following.
We first find the intersection $q$ of the tangent ray of $A$ at
$a$ and the tangent ray of $B$ at $b$. Then, $f(a)\geq b$ if and only if $d(\alpha_1,q)\leq  d(\alpha_2,q)$. Notice that $d(\alpha_1,q)=w(a_1)+|\widehat{a_1a}|+|\overline{aq}|$, where $|\widehat{a_1a}|$ is the length of the portion of $A_1$ between $a_1$ and $a$. Recall that we have a data structure on each elementary chain $C$ such that given any two vertices on $C$, the length of the sub-chain of $C$ between the two vertices can be computed in $O(1)$ time. Using the data structure,
$|\widehat{a_1a}|$ can be computed in constant time. Since $w(a_1)$ is already known, $d(\alpha_1,q)$ can be computed in constant time. So is $d(\alpha_2,q)$.
In the case where the two tangent rays do not intersect, either the tangent ray of $A$ at $a$ intersects the backward extension of the tangent ray of $B$ at $b$ or the tangent ray of $B$ at $b$ intersects the backward extension of the tangent ray of $A$ at $a$. In the former case $f(a)\leq b$ holds while in the latter case $f(a)\geq b$ holds. Hence, whether $f(a)\geq b$ can be determined in constant time.

To determine whether $g(b)\geq a$, we do the following. Find the intersection $p_a$ between $\ell$ and the tangent ray of $A$ at $a$ and the intersection $p_b$ between $\ell$ and the tangent ray of $B$ at $b$. If $p_a$ is to the left of $p_b$, then $g(b)\geq a$; otherwise $g(b)\leq  a$.
Note that by the definition of $A$, the tangent ray at any point of $A$ intersects $I$; the same is true for $B$. Hence, whether $g(b)\geq a$ can be determined in constant time.

The above algorithm returns a point $q$ in $O(\log n)$ time. If the intersection of $\ell$ and $B(\alpha_1,\alpha_2)$ exists, then $q$ is the intersection. Because we do not know whether the intersection exists, we finally validate $q$ by computing $d(\alpha_1,q)$ and $d(\alpha_2,q)$ in $O(\log n)$ time as well as checking whether $q\in \ell$. The point $q$ is valid if and only if $d(\alpha_1,q)=d(\alpha_2,q)$ and $q\in \ell$.
\end{proof}

\subsection{The wavefront propagation procedure}
\label{sec:propagation}
In this section, we discuss the wavefront propagation procedure, which
is to compute the wavefront $W(e,g)$ for all transparent edges $g\in output(e)$ based on
$W(e)$. Consider a transparent edge $g\in output(e)$. The wavefront
$W(e,g)$ refers to the portion of $W(e)$ that arrives at $g$ through
the well-covering region $\calU(g)$ of $g$ if $e\in
input(g)$ and through $\calU(e)$ otherwise (in the latter case $g\in
input(e)$). We will need to handle the bisector events, i.e., the
intersections between bisectors and the intersections between
bisectors and obstacle edges.
The HS algorithm processes the bisector events in {\em temporal} order,
i.e., in order of the simulation time $\tau$. The HSY algorithm
instead proposes a simpler
approach that processes the events in {\em spatial order}, i.e., in
order of their geometric locations. We will adapt the HSY's
spacial-order method.

Recall that each wavefront $W(e)$ is represented by a list of generators, which are
maintained in the leaves of a fully-persistent balanced binary search
tree $T(e)$. We further assign each generator a ``next bisector event'', which
is the intersection of its two bounding bisectors (it is set to null if the two bisectors do not intersect). More
specifically, for each bisector $\alpha$, we assign it
the intersection of the two bisectors $B(\alpha_l,\alpha)$ and
$B(\alpha,\alpha_r)$, where $\alpha_l$ and $\alpha_r$ are $\alpha$'s
left and right neighboring generators in $W(e)$, respectively; we store the intersection at the leaf for $\alpha$. Our algorithm
maintains a variant that the next bisector event for each generator in
$W(e)$ has already been computed and stored in $T(e)$.
We further endow the tree $T(e)$ with additional node-fields so that each internal node stores a value
that is equal to the minimum (resp., maximum) $x$-coordinate (resp., $y$-coordinate) among all bisector
events stored at the leaves of the subtree rooted at the node. Using these extra values, we
can find from a query range of generators the generator whose bisector event has the minimum/maximum $x$- or $y$-coordinate in logarithmic time.


The propagation from $W(e)$ to $g$ through $\calU$ is done cell by cell, where $\calU$ is either $\calU(e)$ or $\calU(g)$. We start propagating $W(e)$ to the adjacent cell $c$ of $e$ in $\calU$ to compute the wavefront passing through all edges of $c$. Then by using the computed wavefronts on the edges of $c$, we recursively run the algorithm on cells of $\calU$ adjacent to $c$. As $\calU$ has $O(1)$ cells, the propagation passes through $O(1)$ cells. Hence, the essential ingredient of the algorithm is to propagate a single wavefront, say, $W(e)$, across a single cell $c$ with $e$ on its boundary.
Depending on whether $c$ is an empty rectangle, there are two cases.

\subsubsection{$c$ is an empty rectangle}

We first consider the case where $c$ is an empty rectangle, i.e., there is no island inside $c$ and $c$ does not intersect any obstacle. Without loss of generality, we assume that $e$ is an edge on the bottom side of $c$, and thus all generators of $W(e)$ are below $e$. Our goal is to compute $W(e,g)$, i.e., the generators of $W(e)$ claiming $g$, for all other edges $g$ of $c$. Our algorithm is similar to the HSY algorithm in the high level but the low level implementations are quite different. The main difference is that each bisector in the HSY algorithm is of constant size while this is not the case in our problem. Due to this, it takes constant time to compute the intersection of two bisectors in the HSY algorithm while in our problem this costs $O(\log n)$ time.

The technical crux of the algorithm is to process the intersections in $c$ among the bisectors of generators of $W(e)$. Since all generators of $W(e)$ are below $e$, their bisectors in $c$ are $y$-monotone by Corollary~\ref{coro:monotone}. This is a critical property our algorithm relies on. Due to the property, we only need to compute $W(e,g)$ for all edges $g$ on the left, right, and top sides of $c$.
Another helpful property is that since we propagate $W(e)$ through $e$ inside $c$, if a generator of $\alpha$ of $W(e)$ claims a point $q\in c$, then the tangent from $q$ to $\alpha$ must cross $e$; we refer it as the {\em propagation property}. Due to this property, the points of $c$ claimed by $\alpha$ must be to the right of the tangent ray from the left endpoint of $e$ to $\alpha$ (the direction of the ray is the from the tangent point to the left endpoint of $e$), as well as to the left of the tangent ray from the right endpoint of $e$ to $\alpha$ (the direction of the ray is the from the tangent point to the right endpoint of $e$). We call the former ray the {\em left bounding ray} of $\alpha$ and the latter the {\em right bounding ray} of $\alpha$.
As such, for the leftmost generator of $W(e)$, we consider its left bounding ray as its left bounding bisector; similarly, for the rightmost generator of $W(e)$, we consider its right bounding ray as its right bounding bisector.

Starting from $e$, we use a horizontal line segment $\ell$ to sweep $c$ upwards until its top side. At any moment during the algorithm, the algorithm maintains a subset $W(\ell)$ of generators of $W(e)$ for $\ell$ by a balanced binary search tree $T(\ell)$; initially $W(\ell)=W(e)$ and $T(\ell)=T(e)$.
Let $[x_1,x_2]\times[y_1,y_2]$ denote the coordinates of $c$. Using the extra fields on the nodes of the tree $T(\ell)$, we compute a maximal prefix $W_1(\ell)$ (resp., $W_2(\ell)$) of generators of $W(\ell)$ such that the bisector events assigned to all generators in it have $x$-coordinates less than $x_1$ (resp., larger than $x_2$). Let $W_m(\ell)$ be the remaining elements of $W(\ell)$. By definition, the first and last generators of $W_m(\ell)$ have their bisector events with $x$-coordinates in $[x_1,x_2]$. As all bisectors are $y$-monotone in $c$, the lowest bisector intersection in $c$ above $\ell$ must be the ``next bisector event'' $b$ associated with a generator in $W_m(\ell)$,
which can be found in $O(\log n)$ time using the tree $T(\ell)$. We advance $\ell$ to the $y$-coordinate of $b$ by removing the generator $\alpha$ associated with the event $b$. Finally, we recompute the next bisector events for the two neighbors of $\alpha$ in $W(\ell)$.
Specifically, let $\alpha_l$ and $\alpha_r$ be the left and right neighboring generators of $\alpha$ in $W(\ell)$, respectively. We need to compute the intersection of the two bounding bisectors of $\alpha_l$,  
and update the bisector event of $\alpha_l$ to this intersection. Similarly, we need to compute the intersection of the bounding bisectors of $\alpha_r$,
and update the bisector event of $\alpha_r$ to this intersection. Lemma~\ref{lem:bb-intersection} below shows that each of these bisector intersections can be computed in $O(\log n)$ time by a bisector-intersection operation, using the tentative prune-and-search technique of Kirkpatrick and Snoeyink~\cite{ref:KirkpatrickTe95}.
Note that if $\alpha$ is the leftmost generator, then $\alpha_r$ becomes the leftmost after $\alpha$ is deleted, in which case we compute the left bounding ray of $\alpha_r$ as its left bounding generator.
If $\alpha$ is the rightmost generator, the process is symmetric.

\begin{lemma}{\em\bf (Bisector-bisector intersection operation)}\label{lem:bb-intersection}
Each bisector-bisector intersection operation can be performed in $O(\log n)$ time.
\end{lemma}
\begin{proof}
We are given a horizontal line $\ell$ and three generators $\alpha_1$, $\alpha_2$, and $\alpha_3$ below $\ell$ such that they claim points on $\ell$ in this order. The goal is to compute the intersection of the two bisectors $B(\alpha_1,\alpha_2)$ and $B(\alpha_2,\alpha_3)$ above $\ell$, or determine that such an intersection does not exist. Using the tentative prune-and-search technique of Kirkpatrick and Snoeyink~\cite{ref:KirkpatrickTe95}, we present an $O(\log n)$ time algorithm.

To avoid the lengthy background explanation, we follow the notation
in~\cite{ref:KirkpatrickTe95} without definition.
We will rely on Theorem 3.9 in~\cite{ref:KirkpatrickTe95}. To this
end, we need to define three continuous and decreasing functions $f$, $g$,
and $h$. We define them in a way similar to Theorem~4.10 in~\cite{ref:KirkpatrickTe95} for finding a point equidistant to three convex polygons. Indeed, our problem may be considered as a weighted
case of their problem because each point in the underlying chains of the generators has a
weight that is equal to its weighted distance from its generator.

Let $A$, $B$, and $C$ be the underlying chains of $\alpha_1$, $\alpha_2$, and $\alpha_3$, respectively.

We parameterize over $[0,1]$ each of the three convex chains $A$,
$B$, and $C$ from one end to the other in
clockwise order, i.e.,
each value of $[0,1]$ corresponds to a slope of a tangent at a point on
the chains. For each point $a$ of $A$, we define
$f(a)$ to be the parameter of the point $b\in B$ such that the tangent ray
of $A$ at $a$ (following the designated direction) and the tangent ray of $B$ at $b$ intersect at a point on the bisector $B(\alpha_1,\alpha_2)$ (e.g., see Fig.~\ref{fig:bisectorline} left); if the tangent ray at $a$ does not intersect $B(\alpha_1,\alpha_2)$, then define $f(a)=1$.
In a similar manner, we define $g(b)$ for $b\in B$ with respect to $C$ and define $h(c)$ for $c\in C$
with respect to $A$. One can verify that all three functions are
continuous and decreasing (the tangent at an obstacle vertex of the chains is not unique but the issue can be handled~\cite{ref:KirkpatrickTe95}). The fixed-point of the
composition of the three functions $h\cdot g \cdot f$ corresponds to
the intersection of $B(\alpha_1,\alpha_2)$ and $B(\alpha_2,\alpha_3)$, which can be computed by applying the tentative prune-and-search algorithm of Theorem~3.9~\cite{ref:KirkpatrickTe95}.

To guarantee that the algorithm can be implemented in $O(\log n)$ time, since each of the chains $A$, $B$, and $C$ is represented by an array, we need to show that given
any $a\in A$ and any $b\in B$, we can determine whether $f(a)\geq b$ in $O(1)$
time. This can be done in the same way as in the proof of Lemma~\ref{lem:bl-intersection}.
Similarly, given any $b\in B$ and $c\in C$, we can determine whether $g(b)\geq c$ in $O(1)$ time, and given any $c\in C$ and $a\in A$, we can determine whether $h(c)\geq a$ in $O(1)$ time.
Therefore, applying Theorem 3.9~\cite{ref:KirkpatrickTe95} can compute the intersection of $B(\alpha_1,\alpha_2)$ and $B(\alpha_2,\alpha_3)$ in $O(\log n)$ time.

The above algorithm is based on the assumption that the intersection of the two bisectors exists. However, we do not know whether that is true or not. To determine it, we finally validate the solution as follows. Let $q$ be the point returned by the algorithm. We compute the distances $d(\alpha_i,q)$ for $i=1,2,3$. The point $q$ is a true bisector intersection if and only if the three distances are equal. Finally, we return $q$ if and only if $q$ is above $\ell$.
\end{proof}

The algorithm finishes once $\ell$ is at the top side of $c$. At this moment, no bisector events of $W(\ell)$ are in $c$. Finally, we run the following {\em wavefront splitting step} to split $W(\ell)$ to obtain $W(e,g)$ for all edges $g$ on the left, right, and top sides $c$. Our algorithm relies on the following observation.
Let $\zeta$ be the union of the left, top, and right sides of $c$.

\begin{lemma}
The list of generators of $W(\ell)$ are exactly those in $W(\ell)$ claiming $\zeta$ in order.
\end{lemma}
\begin{proof}
It suffices to show that during the sweeping algorithm whenever a bisector $B(\alpha_1,\alpha_2)$ of two generators $\alpha_1$ and $\alpha_2$ intersects $\zeta$, it will never intersect $\partial c$ again. Let $q$ be such an intersection. Let $\zeta_l$, $\zeta_t$, and $\zeta_r$ be the left, top, and right sides of $c$, respectively.

If $q$ is on $\zeta_t$, then since both $\alpha_1$ and $\alpha_2$ are below $e$, they are also below $\zeta_t$. By Lemma~\ref{lem:intersection}, $B(\alpha_1,\alpha_2)$ will not intersect the supporting line of $\zeta_t$ again and thus will not intersect $\partial c$ again.

If $q$ is on $\zeta_l$, then we claim that both generators $\alpha_1$ and $\alpha_2$ are to the right of the supporting line $\ell_l$ of $\zeta_l$. Indeed, since both generators claim $q$, the bounding rays (i.e., the left bounding ray of the leftmost generator of $W(\ell)$ and the right bounding ray of the rightmost generator of $W(\ell)$ during the sweeping algorithm) guarantee the propagation property: the tangents from $q$ to the two generators must cross $e$. Therefore, both generators must be to the right of $\ell_l$. By Lemma~\ref{lem:intersection}, $B(\alpha_1,\alpha_2)$ will not intersect the supporting line of $\zeta_l$ again and thus will not intersect $\partial c$ again.

If $q$ is on $\zeta_r$, the analysis is similar to the above second case.
\end{proof}

In light of the above lemma, our wavefront splitting step algorithm for computing $W(e,g)$ of all edges $g\in \zeta$ works as follows. Consider an edge $g\in \zeta$. Without loss of generality, we assume that the points of $\zeta$ are clockwise around $c$ so that we can talk about their relative order.

Let $p_l$ and $p_r$ be the front and rear endpoints of
$g$, respectively. Let $\alpha_l$ and $\alpha_r$ be the generators of
$W(\ell)$ claiming $p_l$ and $p_r$, respectively. Then all generators
of $W(\ell)$ to the left of $\alpha_l$ including $\alpha_l$ form
the wavefront for all edges of $\zeta$ in the front of $g$; all generators of $W(\ell)$ to
the right of $\alpha_r$ including $\alpha_r$ form the wavefront for all edges of $\zeta$ after $g$; all generators of $W(\ell)$ between $\alpha_l$ and $\alpha_r$ including $\alpha_l$ and $\alpha_r$ form
$W(e,g)$. Hence, once $\alpha_l$ and $\alpha_r$ are known, $W(e,g)$ can be obtained
by splitting $W(\ell)$ in $O(\log n)$ time. It remains to compute $\alpha_l$ and $\alpha_r$.
Below, we only discuss how to compute the generator $\alpha_l$ since $\alpha_r$ can be computed analogously.

Starting from the root $v$ of $T(\ell)$, we determine the intersection $q$ between $B(\alpha_1,\alpha_2)$ and $\zeta$, where $\alpha_1$ is the rightmost generator in the left subtree of $v$ and $\alpha_2$ is the leftmost generator of the right subtree of $v$.
If $q$ is in the front of $g$ on $\zeta$, then we proceed to the right subtree of $v$; otherwise, we proceed to the left subtree of $v$.

It is easy to see that the runtime of the algorithm is bounded by $O(\eta\cdot \log n)$ time, where $\eta$ is the time for computing $q$. In the HSY algorithm, each bisector is of constant size and an oracle is assumed to exist that can compute $q$ in $O(1)$ time. In our problem, since a bisector may not be of constant size, it is not clear how to bound $\eta$ by $O(1)$. But $\eta$ can be bounded by $O(\log n)$ using the bisector-line intersection operation in Lemma~\ref{lem:bl-intersection}. Thus, $\alpha_l$ can be computed in $O(\log^2 n)$ time. However, this is not sufficient for our purpose, as this would lead to an overall $O(n+h\log^2 h)$ time algorithm. We instead use the following {\em binary search plus bisector tracing} approach.

During the wavefront expansion algorithm, for each pair of neighboring generators $\alpha=(A,a)$ and $\alpha'=(A',a')$ in a wavefront (e.g., $W(e)$), we maintain a special point $z(\alpha,\alpha')$ on the bisector $B(\alpha,\alpha')$. For example, in the above sweeping algorithm, whenever a generator $\alpha$ is deleted from $W(\ell)$ at a bisector event $b=B(\alpha_l,\alpha)\cap B(\alpha,\alpha_r)$, its two neighbors $\alpha_l$ and $\alpha_r$ now become neighboring in $W(\ell)$. Then, we initialize $z(\alpha_l,\alpha_r)$ to $b$ (the tangent points from $b$ to $\alpha_l$ and $\alpha_r$ are also associated with $b$). During the algorithm, the point $z(\alpha_l,\alpha_r)$ will move on $B(\alpha,\alpha')$ further away from the two defining generators $\alpha$ and $\alpha'$ and the movement will trace out the hyperbolic-arcs of the bisector. We call $z(\alpha,\alpha')$ the {\em tracing-point} of $B(\alpha,\alpha')$.
Our algorithm maintains a variant that the tracing point of each bisector of $W(\ell)$ is below the sweeping line $\ell$ (initially, the tracing point of each bisector of $W(e)$ is below $e$).

With the help of the above $z$-points, we compute the generator $\alpha_l$ as follows. Like the above algorithm, starting from the root $v$ of $T(\ell)$, let $\alpha_1$ and $\alpha_2$ be the two generators as defined above.
To compute the intersection $q$ between $B(\alpha_1,\alpha_2)$ and $\zeta$, we trace out the bisector $B(\alpha_1,\alpha_2)$ by moving its tracing-point $z(\alpha_1,\alpha_2)$ upwards (each time trace out a hyperbolic-arc of $B(\alpha_1,\alpha_2)$) until the current tracing hyperbolic-arc intersects $\zeta$ at $q$.
If $q$ is in the front of $e$ on $\zeta$, then we proceed to the right subtree of $v$; otherwise, we proceed to the left subtree of $v$.

After $W(e,g)$ is obtained, we compute $W(e,g')$ for other edges $g'$ on $\zeta$ using the same algorithm as above. For the time analysis, observe that each bisector hyperbolic-arc will be traced out at most once in the wavefront splitting step for all edges of $\zeta$ because the tracing point of each bisector will never move ``backwards''.

This finishes the algorithm for propagating $W(e)$ through the cell $c$. Except the final wavefront splitting step, the algorithm runs in $O((1+h_c)\log n)$ time, where $h_c$ is the number of bisector events in $c$. Because $c$ has $O(1)$ edges, the wavefront splitting step takes $O(\log n + n_c)$ time, where $n_c$ is the number of hyperbolic-arcs of bisectors that are traced out.

\subsubsection{$c$ is not an empty rectangle}

We now discuss the case in which the cell $c$ is not an empty rectangle. In this case, $c$ has a square hole inside or/and the boundary of $c$ contains obstacle edges.
Without loss of generality, we assume that $e$ is on the bottom side of $c$.

If $c$ contains a square hole, then we partition $c$ into four subcells by cutting $c$ with two lines parallel to $e$, each passing through an edge of the hole. If $c$ has obstacle edges on its boundary, recall that these obstacles edges belong to $O(1)$ convex chains (each of which is a fragment of an elementary chain); we further partition $c$ by additional edges parallel to $e$, so that each resulting subcell contains at most two convex chains, one the left side and the other on the right side. Since $\partial c$ has $O(1)$ convex chains, $O(1)$ additional edges are sufficient to partition $c$ into $O(1)$ subcells as above. Then, we propagate $e$ through the subcells of $c$, one by one. In the following, we describe the algorithm for one such subcell. By slightly abusing the notation, we still use $c$ to denote the subcell with $e$ on its bottom side.

Since $\partial c$ has obstacle edges, the propagation algorithm becomes more complicated. As in the HSY algorithm, comparing with the algorithm for the previous case, there are two new bisector events.

\begin{itemize}
\item
First, a bisector may intersect a convex chain (and thus intersect an obstacle). The HSY algorithm does not explicitly compute these bisector events because such an oracle is not assumed to exist. In our algorithm, however, because the obstacles in our problem are polygonal, we can explicitly determine these events without any special assumption. This is also a reason that the high-level idea of our algorithm is simpler than the HSY algorithm.


\begin{figure}[t]
\begin{minipage}[t]{\textwidth}
\begin{center}
\includegraphics[height=1.7in]{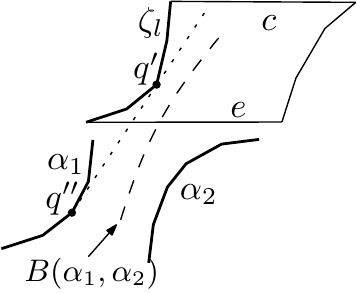}
\caption{\footnotesize Illustrating the creation of a new generator at $q'$.}
\label{fig:creategenerator10}
\end{center}
\end{minipage}
\vspace{-0.15in}
\end{figure}

\item
Second, new generators may be created at the convex chains. We still sweep a horizontal line $\ell$ from $e$ upwards. Let $W(\ell)$ be the current wavefront at some moment during the algorithm. Consider two neighboring generators $\alpha_1$ and $\alpha_2$ in $W(\ell)$ with $\alpha_1$ on the left of $\alpha_2$. We use $\zeta_l$ to denote the convex chain on the left side of $c$. Let $q'$ be the tangent point on $\zeta_l$ of the common tangent between $\zeta_l$ and $\alpha_1$ and let $q''$ be the tangent point on $\alpha_1$ (e.g., see Fig.~\ref{fig:creategenerator10}). If $d(\alpha_1,q')<d(\alpha_2,q')$, then a new generator $\alpha$ on $\zeta_l$ with initial vertex $q'$ and weight equal to $d(\alpha_1,q')$ is created (designated counterclockwise direction) and inserted into $W(\ell)$ right before $\alpha_1$. The bisector $B(\alpha,\alpha_1)$ is the ray emanating from $q'$ and extending away from $q''$. The region to the left of the ray has $\alpha$ as its predecessor. When the sweeping line $\ell$ is at $q'$, all wavelets in $W(\ell)$ to the left of $\alpha_1$ have already collided with $\zeta_l$ and thus the first three generators of $W(\ell)$ are $\alpha$, $\alpha_1$, and $\alpha_2$.
\end{itemize}

In what follows, we describe our sweeping algorithm to propagate $W(e)$ through $c$.
We begin with an easier case where only the left side of $c$ is a convex chain, denoted by $\zeta_l$ (and the right side is a vertical transparent edge, denoted by $\zeta_r$). We use $\zeta_t$ to denote the top side of $c$, which is a transparent edge.
As in the previous case, we sweep a line $\ell$ from $e$ upwards until the top side $\zeta_t$.
During the algorithm, we maintain a list $W(\ell)$ of generators by a balanced binary search tree $T(\ell)$. Initially, $W(\ell)=W(e)$ and $T(\ell)=T(e)$.

We compute the intersection $q$ of the convex chain $\zeta_l$ and the bisector $B(\alpha_1,\alpha_2)$, for the leftmost bisectors of $\alpha_1$ and $\alpha_2$ of $W(\ell)$. We call it the {\em bisector-chain intersection operation}. The following lemma shows that this operation can be performed in $O(\log n)$ time.

\begin{lemma}{\em\bf (Bisector-chain intersection operation)}\label{lem:bc-intersection}
Each bisector-chain intersection operation can be performed in $O(\log n)$ time.
\end{lemma}
\begin{proof}
We are given a convex chain $\zeta_l$ above a horizontal line $\ell$ and two generators $\alpha_1=(A_1,a_1)$ and $\alpha_2=(A_2,a_2)$ below $\ell$ such that they claim points on $\ell$ in this order. The goal is to compute the intersection of $B(\alpha_1,\alpha_2)$ and $\zeta_l$, or determine that they do not intersect. In the following, using the tentative prune-and-search technique of Kirkpatrick and Snoeyink~\cite{ref:KirkpatrickTe95}, we present an $O(\log n)$ time algorithm.

To avoid the lengthy background explanation, we follow the notation
in~\cite{ref:KirkpatrickTe95} without definition.
We will rely on Theorem 3.9 in~\cite{ref:KirkpatrickTe95}. To this
end, we need to define three continuous and decreasing functions $f$, $g$,
and $h$.

Suppose $q$ is the intersection of $B(\alpha_1,\alpha_2)$ and $\zeta_l$. Let $p_1$ and $p_2$ be the tangent points from $q$ to $A_1$ and $A_2$, respectively. Then, $\overline{qp_1}$ (resp., $\overline{qp_2}$) does not intersect $\zeta_l$ other than $q$. We determine the portion $C$ of $\zeta_l$ such that for each point $p\in C$, its tangent to $A_1$ does not intersect $\zeta_l$ other than $p$. Hence, $q\in C$. $C$ can be determined by computing the common tangents between $\zeta_l$ and $A_1$, which can be done in $O(\log n)$ time~\cite{ref:GuibasCo91,ref:OvermarsMa81}. Also, we determine the portion $B$ of $A_2$ such that the tangent ray at any point of $B$ must intersect $C$. This can be done by computing the common tangents between $C$ and $A_2$ in $O(\log n)$ time~\cite{ref:GuibasCo91,ref:OvermarsMa81}. Let $A=A_1$.

We parameterize over $[0,1]$ each of the three convex chains $A$, $B$, and $C$ from one end to the other in
clockwise order, i.e.,
each value of $[0,1]$ corresponds to a slope of a tangent at a point on
the chains $A$ and $B$, while each value of $[0,1]$ corresponds to a point of $C$.
For each point $a$ of $A$, we define
$f(a)$ to be the parameter of the point $b\in B$ such that the tangent ray
of $A$ at $a$ (following the designated direction of $\alpha_1$) and the tangent ray of $B$ at $b$ intersect at a point on the bisector $B(\alpha_1,\alpha_2)$ (e.g., see Fig.~\ref{fig:bisectorline} left); if the tangent ray at $a$ does not intersect $B(\alpha_1,\alpha_2)$, then define $f(a)=1$.
For each point $b$ of $B$, define $g(b)$ to be the parameter of the point $c\in C$ such that $\overline{cb}$ is tangent to $B$ at $b$ (e.g., see Fig.~\ref{fig:bcintersection} left); note that by the definition of $B$, the tangent ray from any point of $B$ must intersect $C$ and thus such a point $c\in C$ must exist.
For each point $c\in C$, define $h(c)$ to be the parameter of the point of $a\in A$ such that $\overline{ac}$ is tangent to $A$ at $a$ (e.g., see Fig.~\ref{fig:bcintersection} right); note that by the definition of $C$, such a point $a\in A$ must exist and $\overline{ac}$ does not intersect $C$ other than $c$. One can verify that all three functions are
continuous and decreasing (the tangent at an obstacle vertex of the chains is not unique but the issue can be handled~\cite{ref:KirkpatrickTe95}). The fixed-point of the
composition of the three functions $h\cdot g \cdot f$ corresponds to
the intersection $q$ of $B(\alpha_1,\alpha_2)$ and $\zeta_l$, which can be computed by applying the tentative prune-and-search algorithm of Theorem 3.9~\cite{ref:KirkpatrickTe95}.

\begin{figure}[t]
\begin{minipage}[t]{\textwidth}
\begin{center}
\includegraphics[height=1.5in]{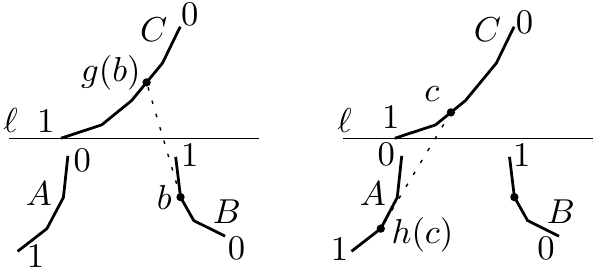}
\caption{\footnotesize Illustrating the definitions of $g(b)$ and $h(c)$.}
\label{fig:bcintersection}
\end{center}
\end{minipage}
\vspace{-0.15in}
\end{figure}

To make sure that the algorithm can be implemented in $O(\log n)$ time, since each convex chain is part of an elementary chain and thus is represented by an array, it suffices to show the following: (1) given
any $a\in A$ and any $b\in B$, whether $f(a)\geq b$ can be determined in $O(1)$ time; (2) given any $b\in B$ and any $c\in C$, whether $g(b)\geq c$ can be determined in $O(1)$ time; (3) given any $c\in C$ and any $a\in A$, whether $h(c)\geq a$ can be determined in $O(1)$ time. We prove them below.

Indeed, for (1), it can be done in the same way as in the proof of
Lemma~\ref{lem:bl-intersection}.
For (2), $g(b)\geq c$ if and only if $c$ is to the right of the tangent ray of $B$ at $b$, which can be easily determined in $O(1)$ time.
For (3), $h(c)\geq a$ if and only $c$ is to the right of the tangent ray of $A$ at $a$, which can be easily determined in $O(1)$ time.

Therefore, applying the tentative prune-and-search technique in Theorem 3.9~\cite{ref:KirkpatrickTe95} can compute $q$ in $O(\log n)$ time.

Note that the above algorithm is based on the assumption that the intersection of $B(\alpha_1,\alpha_2)$ and $\zeta_l$ exists. However, we do not know whether this is true or not. To determine that, we finally validate the solution as follows. Let $q$ be the point returned by the algorithm. We first determine whether $q\in \zeta_l$. If not, then the intersection does not exist. Otherwise, we further compute the two distances $d(\alpha_i,q)$ for $i=1,2$ in $O(\log n)$ time. If the distances are equal, then $q$ is the true intersection; otherwise, the intersection does not exist.
\end{proof}

If the intersection $q$ of $\zeta_l$ and $B(\alpha_1,\alpha_2)$ does not exist, then we compute the tangent between $\zeta_l$ and $\alpha_1$, which can be done in $O(\log n)$ time~\cite{ref:OvermarsMa81}; let $q'$ be the tangent point at $\zeta_l$.
Regardless whether $q$ exists or not, we compute the lowest bisector intersection $b$ in $c$ above $\ell$ in the same way as in the algorithm for the previous case where $c$ is an empty rectangle. Depending on whether $q$ exists or not, we proceed as follows. For any point $p$ in the plane, let $y(p)$ denote the $y$-coordinate of $p$.

\begin{enumerate}
\item
If $q$ exists, then depending on whether $y(q)\leq y(b)$, there are two subcases.
If $y(q)\leq y(b)$, then we process the bisector event $q$: remove $\alpha_1$ from $W(\ell)$ and then recompute $q$, $q'$, and $b$. Otherwise, we process the bisector event at $b$ in the same way as in the previous case and then recompute $q$, $q'$, and $b$.

\item
If $q$ does not exist, then depending on whether $y(b)\leq y(q')$, there are two subcases. If $y(b)\leq y(q')$, then we process the bisector event at $b$ in the same way as before and then recompute $q$, $q'$, and $b$.
Otherwise, we insert a new generator $\alpha=(A,q')$ to $W(\ell)$ as the leftmost generator, where $A$ is the fragment of the elementary chain containing $\zeta_l$ from $q'$ counterclockwise to the end of the chain, and $\alpha$ is designated the counterclockwise direction of $A$ and the weight of $q'$ is $d(\alpha_1,q')$; e.g., see Fig.~\ref{fig:creategenerator10}. The ray from $q'$ in the direction from $q''$ to $q'$ is the bisector of $\alpha$ and $\alpha_1$, where $q''$ is the tangent point on $\alpha_1$ of the common tangent between $\alpha_1$ and $\zeta_l$. We initialize the tracing-point $z(\alpha,\alpha_1)$ of $B(\alpha,\alpha_1)$ to $q'$. Finally, we recompute $q$, $q'$, and $b$.
\end{enumerate}

Once the sweep line $\ell$ reaches the top side $\zeta_t$ of $c$, the algorithm stops. Finally, as in the previous case, we run a wavefront
splitting step. Because the left side $\zeta_l$ consists of obstacle
edges, we split $W(\ell)$ to compute $W(e,g)$ for all transparent edges $g$ on the top side $\zeta_t$ and the right side $\zeta_r$ of $c$.
The algorithm is the same as the previous case.

The above discusses the case that only the left side $\zeta_l$ of $c$ is a
convex chain. For the general case where both the left and right
sides of $c$ are convex chains, the algorithm is similar. The
difference is that we have to compute a point $p$ corresponding to $q$
and a point $p'$ corresponding to $q'$ on the right side $\zeta_r$ of
$c$. More specifically, $p$ is the intersection of
$B(\alpha_2',\alpha_1')$ with $\zeta_r$, where $\alpha_2'$ and $\alpha_1'$
are the two rightmost generators of $W$. If $p$ does not exist, then we
compute the common tangent between $\zeta_r$ and $\alpha_1'$, and $p'$ is the tangent
point on $\zeta_r$.

In the following, if $q$ does not exist, we let $y(q)$ be
$\infty$; otherwise, $q'$ is not needed and we let $y(q')$ be
$\infty$. We apply the same convention to $p$ and $p'$. We define $b$
as the bisector event in the same way as before. In each step, we
process the lowest point $r$ of $\{q,q',b,p,p'\}$. If $r$
is $q$ or $p$, we process it in the same way as before for $q$.
If $r$ is $q'$ or $p'$, we process it in the same way as before for $q'$.
If $r$ is $b$, we process it in the same way as before.
After processing $r$, we recompute the five points. Each step
takes $O(\log n)$ time.
After the sweep line $\ell$ reaches the top side $\zeta_t$ of $c$,
$W(\ell)$ is $W(e,\zeta_t)$ for the top side $\zeta_t$ of $c$ because both the
left and right sides of $c$ are obstacle edges.
Finally, we run the wavefront splitting step on $W(\ell)$ to compute the wavefronts $W(e,g)$ for all transparent edges $g$ on $\zeta_t$.

In summary, propagating $W(e)$ through $c$ takes $O((1+h_c)\cdot
\log n+n_c)$ time, where $h_c$ is the number of bisector events (including both the bisector-bisector intersection events and the bisector-obstacle intersection events) and $n_c$ is the number of hyperbolic-arcs of bisectors that are traced out in the wavefront splitting step.

We use the following lemma to summarize the algorithm for both cases (i.e., regardless whether $c$ is an empty rectangle or not).

\begin{lemma}\label{lem:propagation}
Suppose $W(e)$ is a wavefront on a transparent edge of a cell $c$ of the subdivision $\calS'$. Then, $W(e)$ can be propagated through $c$ to all other transparent edges of $c$ in $O((1+h_c)\log n+n_c)$ time, where $h_c$ is the number of bisector events (including both the bisector-bisector intersection events and bisector-obstacle intersection events) and $n_c$ is the number of hyperbolic-arcs of bisectors that are traced out in the wavefront splitting step.
\end{lemma}

\subsection{Time analysis}
\label{sec:time}

In this section, we show that the running time of our wavefront expansion algorithm described above is bounded by $O(n+h\log h)$.
For this and also for the purpose of constructing the shortest path map $\spm'(s)$ later in Section~\ref{sec:spm}, as in the HS algorithm, we mark generators in the way that if a generator $\alpha$ is involved in a true bisector event of $\spm(s)$ (either a bisector-bisector intersection or a bisector-obstacle intersection) in a cell $c$ of the subdivision $\calS'$, then $\alpha$ is guaranteed to be in a set of marked generators for $c$ (but a marked generator for $c$ may not actually participate in a true bisector event in $c$).
The generator marking rules are presented below, which are consistent with those in the HS algorithm.

\paragraph{Generator marking rules:}
\begin{enumerate}
\item
For any generator $\alpha=(A,a)$, if its initial vertex $a$ lies in a cell $c$, then mark $\alpha$ in $c$.

\item
Let $e$ be a transparent edge and let $W(e)$ be a wavefront coming from some generator $\alpha$'s side of $e$.

\begin{enumerate}
\item
If $\alpha$ claims an endpoint $b$ of $e$ in $W(e)$, or if it would do so except for an artificial wavefront, then mark $\alpha$ in all cells $c$ incident to $b$.

\item
If $\alpha$'s claim in $W(e)$ is shortened or eliminated by an artificial wavelet, then mark $\alpha$ for $c$, where $c$ is the cell having $e$ as an edge and on $\alpha$'s side of $e$.
\end{enumerate}

\item

Let $e$ and $g$ be two transparent edges with $g\in output(e)$. Mark a generator $\alpha$ of $W(e)$ in both cells having $e$ as an edge if one of the following cases happens:

\begin{enumerate}
\item
$\alpha$ claims an endpoint of $g$ in $W(e,g)$;

\item
$\alpha$ participates in a bisector event either during the wavefront propagation procedure for computing $W(e,g)$ from $W(e)$, or during the wavefront merging procedure for computing $W(g)$. Note that $\alpha$ is also considered to participate in a bisector event if its claim on $g$ is shortened or eliminated by an artificial wavelet.
\end{enumerate}

\item

If a generator $\alpha$ of $W(e)$ claims part of an obstacle edge during the wavefront propagation procedure for propagating $W(e)$ toward $output(e)$ (this includes the case in which $\alpha$ participates in a bisector-obstacle intersection event), then mark $\alpha$ in both cells having $e$ as an edge.
\end{enumerate}

Note that each generator may be marked multiple times and each mark applies to one instance of the generator.


\begin{lemma}\label{lem:numgen}
The total number of marked generators during the algorithm is at most $O(h)$.
\end{lemma}

Because the proof of Lemma~\ref{lem:numgen} is quite technical and lengthy, we devote the entire Section~\ref{sec:lemnumgen} to it. In the rest of this subsection, we use Lemma~\ref{lem:numgen} to show that the running time of our wavefront expansion algorithm is bounded by $O(n+h\log h)$. Our goal is to prove the following lemma.

\begin{lemma}\label{lem:algocom}
The wavefront expansion algorithm runs in $O(n+h\log h)$ time and space.
\end{lemma}

First of all, by Lemma~\ref{lem:subalgo}, constructing the conforming subdivision $\calS'$ can be done in $O(n+h\log h)$ time and $O(n)$ space.

The wavefront expansion algorithm has two main subroutines: the wavefront merging procedure and the wavefront propagation procedure.

The wavefront merging procedure is to construct $W(e)$ based on $W(f,e)$ for the edges $f\in input(e)$. By Lemma~\ref{lem:merge}, this step takes $O((1+k)\log n)$ time, where $k$ is the total number of generators in all wavefronts $W(f,e)$ that are absent from $W(e)$. According to the algorithm, if a generator $\alpha$ is absent from $W(e)$, it must be deleted at a bisector event. Thus, $\alpha$ must be marked by Rule~3(b). Due to Lemma~\ref{lem:numgen}, the total sum of $k$ in the entire algorithm is $O(h)$. As such, the wavefront merging procedure in the entire algorithm takes $O(h\log n)$ time in total.

The wavefront propagation procedure is to compute $W(e,g)$ by propagating $W(e)$ to all edges $g\in output(e)$ either through $\calU(g)$ or $\calU(e)$.
By Lemma~\ref{lem:propagation}, the running time of the procedure is $O((1+h_c)\log n+n_c)$ time, where $h_c$ is the number of bisector events (including both the bisector-bisector intersection events and bisector-obstacle intersection events) and $n_c$ is the number of hyperbolic-arcs of bisectors that are traced out in the wavefront splitting step.
For each bisector-bisector intersection event, at least one involved generator is marked by Rule~3(b). For each bisector-obstacle intersection event, at least one involved generator is marked by Rule~4. Hence, by Lemma~\ref{lem:numgen}, the total sum of $h_c$ in the entire algorithm is $O(h)$. In addition, Lemma~\ref{lem:bisectorvertices} below shows that the total sum of $n_c$ in the entire algorithm is $O(n)$. Therefore, the wavefront propagation procedure in the entire algorithm takes $O(n+h\log n)$ time in total.

\begin{lemma}\label{lem:bisectorvertices}
The total number of traced hyperbolic-arcs of the bisectors in the entire algorithm is $O(n)$.
\end{lemma}
\begin{proof}
First of all, notice that each extension bisector consists of a single hyperbolic-arc, which is a straight line. As each generator is marked by Rule~1, by Lemma~\ref{lem:numgen}, the total number of generators created in the algorithm is $O(h)$. Since each generator can define at most one extension bisector, the number of hyperbolic-arcs on extension bisectors is at most $O(h)$. In the following, we focus on hyperbolic-arcs of non-extension bisectors. Instead of counting the number of traced hyperbolic-arcs, we will count the number of their endpoints.

\begin{figure}[t]
\begin{minipage}[t]{\textwidth}
\begin{center}
\includegraphics[height=1.7in]{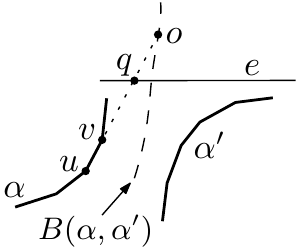}
\caption{\footnotesize Illustrating the definitions of $u$, $v$, $q$, and $o$.}
\label{fig:tracepoint}
\end{center}
\end{minipage}
\vspace{-0.15in}
\end{figure}

Consider a hyperbolic-arc endpoint $o$ that is traced out. According to our algorithm, $o$ is traced out during the wavefront propagation procedure for propagating $W(e)$ to compute $W(e,g)$ for some transparent edge $e$ and  $g\in output(e)$. Suppose $o$ belongs to a non-extension bisector $B(\alpha,\alpha')$ of two generators $\alpha$ and $\alpha'$ in $W(e)$. Then, $o$ must be defined by an obstacle edge $\overline{uv}$ of either $\alpha$ or $\alpha'$, i.e., the ray $\rho(u,v)$ emanating from $v$ along the direction from $u$ to $v$ (which is consistent with the designated direction of the generator that contains $\overline{uv}$) hits $B(\alpha,\alpha')$ at $o$ (e.g., see Fig~\ref{fig:tracepoint}). Without loss of generality, we assume that $\overline{uv}$ belongs to $\alpha$. In the following, we argue that $\overline{uv}$ can define $O(1)$ hyperbolic-arc endpoints that are traced out during the entire algorithm (a hyperbolic-arc endpoint defined by $\overline{uv}$ is counted twice is it is traced out twice), which will prove Lemma~\ref{lem:bisectorvertices} as there are $O(n)$ obstacle edges in total.

We first discuss some properties. Since $o\in B(\alpha,\alpha')$, both $\alpha$ and $\alpha'$ claim $o$. As both generators are in $W(e)$, the ray $\rho(u,v)$ must cross $e$, say, at a point $q$ (e.g., see Fig~\ref{fig:tracepoint}). Because $o$ is traced out when we propagate $W(e)$ to compute $W(e,g)$, $\overline{oq}$ must be in either $\calU(g)$ or $\calU(e)$, i.e., $\overline{oq}\subseteq \calU(e)\cup \calU(g)$. We call $(e,g)$ the {\em defining pair} of $o$ for $\overline{uv}$.
According to our algorithm, during the propagation from $W(e)$ to $g$, $\overline{uv}$ defines only one hyperbolic-arc endpoint, because it is uniquely determined by the wavefront $W(e)$. As such, to prove that $\overline{uv}$ can define $O(1)$ hyperbolic-arc endpoints that are traced out during the entire algorithm, it suffices to show that there are at most $O(1)$ defining pairs for $\overline{uv}$. Let $\Pi$ denote the set of all such defining pairs.
We prove $|\Pi|=O(1)$ below.

For each pair $(e',g')\in \Pi$, according to the above discussion, $\overline{uv}$ and $(e',g')$ define a hyperbolic-arc endpoint $o'$ such that $o'$ is on the ray $\rho(u,v)$ and $\overline{vo'}$ crosses $e'$ at a point $q'$. Without loss of generality, we assume that $(e,g)$ is a pair that minimizes the length $|\overline{vq'}|$. Let $W(e')$ refer to the wavefront at $e'$ whose propagation to $g'$ traces out $o'$.

We partition $\Pi$ into two subsets $\Pi_1$ and $\Pi_2$, where $\Pi_1$ consists of all pairs $(e',g')$ of $\Pi$ such that $\overline{qo}$ intersects $e'$ and $\Pi_2=\Pi\setminus \Pi_1$. Since $\overline{oq}\subseteq \calU(e)\cup \calU(g)$, each well-covering region contains $O(1)$ cells, and $|output(e')|=O(1)$ for each transparent edge $e'$, the size of $\Pi_1$ is $O(1)$.

For $\Pi_2$, we further partition it into two subsets $\Pi_{21}$ and
$\Pi_{22}$, where $\Pi_{21}$ consists of all pairs $(e',g')$ of
$\Pi_2$ such that $e$ is in the well-covering region $\calU(e')$ of
$e'$ or $e'\in \calU(e)\cup \calU(g)$, and $\Pi_{22}=\Pi_2\setminus \Pi_{21}$. Since $e$ is in
$\calU(e')$ for a constant number of transparent edges $e'$, each of $\calU(e)$ and $\calU(g)$ contains $O(1)$ cells, and $|output(e')|=O(1)$ for each $e'$, it holds that $|\Pi_{21}|=O(1)$.
In the following, we argue that $\Pi_{22}=\emptyset$, which will prove $|\Pi|=O(1)$.

Assume to the contrary that $|\Pi_{22}|\neq \emptyset$ and let $(e',g')$ be a pair of $\Pi_{22}$. Since $(e',g')\in \Pi_2$, by the definition of $\Pi_2$, $e'$ does not intersect $\overline{oq}$. By the definition of $e$, $\overline{vq}\setminus\{q\}$ does not intersect $e'$. Recall that $q'$ is the intersection of $e'$ and $\rho(u,v)$. Therefore, the points $v$, $q$, $o$, $q'$, and $o'$ appear on the ray $\rho(u,v)$ in this order (e.g., see Fig~\ref{fig:tracepoint10}). Further, since $(e',g')\in \Pi_{22}$, by the definition of $\Pi_{22}$, $e$ is not in $\calU(e')$ and $e'$ is not in $\calU(e)\cup \calU(g)$.

\begin{figure}[t]
\begin{minipage}[t]{0.49\textwidth}
\begin{center}
\includegraphics[height=2.2in]{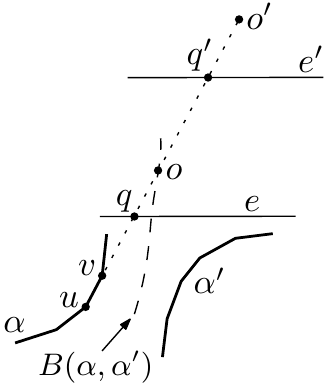}
\caption{\footnotesize Illustrating the definitions of $u$, $v$, $q$, and $o$.}
\label{fig:tracepoint10}
\end{center}
\end{minipage}
\hspace{0.02in}
\begin{minipage}[t]{0.49\textwidth}
\begin{center}
\includegraphics[height=1.8in]{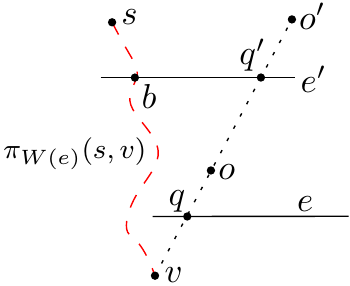}
\caption{\footnotesize The dashed (red) path is $\pi_{W(e)}(s,v)$, which crosses $e'$ at $b$.}
\label{fig:tracepoint20}
\end{center}
\end{minipage}
\vspace{-0.15in}
\end{figure}

Without loss of generality, we assume that $e'$ is horizontal and the
wavefront $W(e')$ is from below $e'$ (thus $v$ is below $e'$ while
$o'$ is above $e'$). Let $\calF'$ be the modified free space by replacing $e'$ with an opaque edge
of open endpoints.
Since the generator in $W(e')$ that contains $v$ claims $q'$, $\pi'(s,q')$ is a shortest path from $s$ to $q'$ in $\calF'$, where $\pi'(s,q')$ is the path following the
wavefront $W(e')$. Since $v$ is in the generator of $W(e')$ that
claims $q'$, $v$ is the anchor of $q'$ in $\pi'(s,q')$, i.e., the edge
of the path incident to $q'$ is $\overline{vq'}$. Let $\pi'(s,v)$ be
the sub-path of $\pi'(s,q')$ between $s'$ and $v$. Then, $\pi'(s,v)$ is a shortest path from $s$ to $v$ in $\calF'$.

Recall that $v$ is in the generator $\alpha$ of $W(e)$. Let $\pi_{W(e)}(s,v)$ be the path from $s$ to $v$ following $W(e)$. We claim that $|\pi'(s,v)|=|\pi_{W(e)}(s,v)|$. Assume to the contrary this is not true. Then, either $|\pi'(s,v)|<|\pi_{W(e)}(s,v)|$ or $|\pi_{W(e)}(s,v)|<|\pi'(s,v)|$.

\begin{itemize}
\item
If $|\pi'(s,v)|<|\pi_{W(e)}(s,v)|$, then since $\pi_{W(e)}(s,v)$ is a shortest path from $s$ to $v$ in the modified free space by considering $e$ as an opaque edge of open endpoints, $\pi'(s,v)$ must cross the interior of $e$. This means that $\pi'(s,q')=\pi'(s,v)\cup \overline{vq'}$ crosses $e$ twice. Because $\pi'(s,q')$ is a shortest path in $\calF'$ and $e\in \calF'$, it cannot cross $e$ twice, a contradiction.

\item
If $|\pi_{W(e)}(s,v)|<|\pi'(s,v)|$, then since $\pi'(s,v)$ is a shortest path from $s$ to $v$ in $\calF'$, the path $\pi_{W(e)}(s,v)$ cannot be in $\calF'$ and thus must cross the interior of $e'$, say, at a point $b$ (e.g., see Fig~\ref{fig:tracepoint10}). Let $\pi_{W(e)}(s,b)$ be the sub-path of $\pi_{W(e)}(s,v)$ between $s$ and $b$. Let $\pi_{W(e)}(b,q')$ be a shortest path from $b$ to $q'$ along $e'$ by considering it as an opaque edge with open endpoints. Note that if $b$ and $q'$ are on different sides of $e'$, then $\pi_{W(e)}(b,q')$ must be through an open endpoint of $e'$. It is not difficult to see that $|\pi_{W(e)}(b,q')|\leq |e'|$. Let $\pi_{W(e)}(s,q')$ be the concatenation of $\pi_{W(e)}(s,b)$ and $\pi_{W(e)}(b,q')$. Hence, $\pi_{W(e)}(s,q')$ is a path in $\calF'$. Since $\pi'(s,q')$ is a shortest path from $s$ to $q'$ in $\calF'$, it holds that $|\pi'(s,q')|\leq |\pi_{W(e)}(s,q')|$.

Notice that $|\pi_{W(e)}(s,q')|=|\pi_{W(e)}(s,b)|+|\pi_{W(e)}(b,q')|\leq |\pi_{W(e)}(s,v)|+|e'|< |\pi'(s,v)|+|e'|$.

On the other hand, $|\pi'(s,q')|=|\pi'(s,v)|+|\overline{vq'}|\geq |\pi'(s,v)|+|\overline{qq'}|$. We claim that $|\overline{qq'}|\geq 2|e'|$. Indeed, since $e$ is outside $\calU(e')$, $q\in e$, and $q'\in e'$, $\overline{qq'}$ must cross $\partial \calU(e')$ at a point $b'$. By the property of well-covering regions of $\calS'$, $|\overline{b'q'}|\geq 2|e'|$. Since $|\overline{qq'}|\geq |\overline{b'q'}|$, we obtain $|\overline{qq'}|\geq 2|e'|$. In light of the claim, we have $|\pi'(s,q')| \geq |\pi'(s,v)|+2|e'|>|\pi'(s,v)|+|e'|> |\pi_{W(e)}(s,q')|$. But this incurs contradiction since  $|\pi'(s,q')|\leq |\pi_{W(e)}(s,q')|$.
\end{itemize}

Therefore, $|\pi'(s,v)|=|\pi_{W(e)}(s,v)|$ holds.

As $q'\not\in \overline{qo}$, we define $p$ as a point on
$\overline{oq'}\setminus\{o\}$ infinitely close to $o$ (e.g., see Fig.~\ref{fig:tracepoint30}). Hence, $p\in
\overline{vq'}$. Since $\pi'(s,q')=\pi'(s,v)\cup \overline{vq'}$ is a shortest
path from $s$ to $q'$ in $\calF'$, $\pi'(s,v)\cup \overline{vp}$ is
a shortest path from $s$ to $p$ to $\calF'$.

Recall that $o$ is on the non-extension bisector $B(\alpha,\alpha')$ of two generators $\alpha$
and $\alpha'$ in $W(e)$ and $v$ is on $\alpha$. Let $v'$ be the anchor
of $o$ in $\alpha'$ (e.g., see Fig.~\ref{fig:tracepoint30}). Since $|\pi'(s,v)|=|\pi_{W(e)}(s,v)|$, we have
$|\pi'(s,v)|+|\overline{vo}|=|\pi_{W(e)}(s,v)|+|\overline{vo}|=|\pi_{W(e)}(s,v')|+|\overline{v'o}|$, where
$\pi_{W(e)}(s,v')$ is the path from $s$ to $v'$ following $W(e)$.
By the definition of $p$ and because $B(\alpha,\alpha')$ is a non-extension bisector, it holds that
$|\pi_{W(e)}(s,v)|+|\overline{vp}|>|\pi_{W(e)}(s,v')|+|\overline{v'p}|$. As $|\pi'(s,v)|=|\pi_{W(e)}(s,v)|$, we have $|\pi'(s,v)|+|\overline{vp}|>|\pi_{W(e)}(s,v')|+|\overline{v'p}|$.
Since $\pi'(s,v)\cup \overline{vp}$ is a shortest path from $s$ to $p$ in
$\calF'$,
$\pi_{W(e)}(s,v')\cup \overline{v'p}$ cannot be a path from $s$ to $p$ in $\calF'$.
Therefore, $\pi_{W(e)}(s,v')\cup \overline{v'p}$ must intersect the interior of $e'$.

\begin{figure}[t]
\begin{minipage}[t]{\textwidth}
\begin{center}
\includegraphics[height=2.0in]{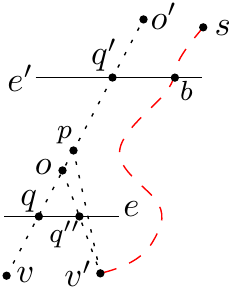}
\caption{\footnotesize The dashed (red) path is $\pi_{W(e)}(s,v')$, which crosses $e'$ at $b$.}
\label{fig:tracepoint30}
\end{center}
\end{minipage}
\vspace{-0.15in}
\end{figure}

We claim that $\overline{v'p}$ cannot intersect $e'$. Indeed, since $\overline{v'p}$ is covered by the wavefront $W(e)$ when $W(e)$ propagates to $g$ in either $\calU(e)$ or $\calU(g)$, $\overline{v'p}$ is in $\calU(e)\cup \calU(g)$. Since $(e',g')\in \Pi_{22}$ and by the definition of $\Pi_{22}$, $e'$ is not in $\calU(e)\cup \calU(g)$. Therefore, $\overline{v'p}$ cannot intersect $e'$.

The above claim implies that $\pi_{W(e)}(s,v')$ must intersect the interior of $e'$ at a point, say, $b$ (e.g., see Fig.~\ref{fig:tracepoint30}). Let $\pi_{W(e)}(s,b)$ be the sub-path of $\pi_{W(e)}(s,v')$ between $s$ and $b$.
Let $\pi_{W(e)}(b,q')$ be a path from $b$ to $q'$ along $e'$ by considering $e'$ as an opaque edge of open endpoints. As discussed above, $|\pi_{W(e)}(b,q')|\leq |e'|$. Hence, $\pi_{W(e)}(s,b)\cup \pi_{W(e)}(b,q')$ is a path in $\calF'$, whose length is at most $|\pi_{W(e)}(s,b)|+|e'|\leq |\pi_{W(e)}(s,v')|+|e'|$.
On the other hand,
$|\pi_{W(e)}(s,v')|+|\overline{v'o}|+|\overline{oq'}|=|\pi'(s,v)|+|\overline{vo}|+|\overline{oq'}|=|\pi'(s,v)|+|\overline{vq'}|$. Since $(e',g')\in \Pi_{22}$ and by the definition of $\Pi_{22}$, $e$ is outside $\calU(e')$. Because $q'\in e'$ and $\overline{ov'}$ crosses $e$ at a point $q''$, by the property of well-covering regions of $\calS'$, $|\overline{q''o}|+|\overline{oq'}|\geq 2|e'|$.
Therefore, we obtain
\begin{align*}
|\pi'(s,v)|+|\overline{vq'}| & =|\pi_{W(e)}(s,v')|+|\overline{v'o}|+|\overline{oq'}|\\
& \geq |\pi_{W(e)}(s,v')|+|\overline{q''o}|+|\overline{oq'}|\\
& \geq |\pi_{W(e)}(s,v')|+2|e'|\geq |\pi_{W(e)}(s,b)| + 2|e'|\\
& >|\pi_{W(e)}(s,b)| + |e'|\geq |\pi_{W(e)}(s,b)| + |\pi_{W(e)}(b,q')|.
\end{align*}
This means that $\pi_{W(e)}(s,b)\cup \pi_{W(e)}(b,q')$ is a path from $s$ to $q'$ in $\calF'$ that is shorter than the path $\pi'(s,v)\cup \overline{vq'}$. But this incurs a contradiction as $\pi'(s,v)\cup \overline{vq'}$ is a shortest path from $s$ to $q'$ in $\calF'$.

This completes the proof of the lemma.
\end{proof}

In summary, the total time of the wavefront expansion algorithm is $O(n+h\log n)$, which is $O(n+h\log h)$.

For the space complexity of the algorithm, since each wavefront $W(e)$ is maintained by a persistent binary tree, each bisector event costs additional $O(\log n)$ space. As discussed above, the total sum of $k+h_c$ in the entire algorithm, which is the total number of bisector events, is $O(h)$. Hence, the space cost by persistent binary trees is $O(h\log n)$. The space used by other parts of the algorithm is $O(n)$. Hence, the space complexity of the algorithm is $O(n+h\log n)$, which is $O(n+h\log h)$. This proves Lemma~\ref{lem:algocom}.

\subsection{Proving Lemma~\ref{lem:numgen}}
\label{sec:lemnumgen}

In this section, we prove Lemma~\ref{lem:numgen}. The proof follows the same strategy in the high level as that in the HS algorithm, although many details are different. We start with the following lemma.

\begin{lemma}\label{lem:pairs}
Suppose $\Pi$ is the set of pairs $(e,B)$ of transparent edges $e$ and bisectors $B$ such that $B$ crosses $e$ in the wavefront $W(e)$ of $e$ during our wavefront expansion algorithm but the same crossing does not occur in $\spm'(s)$. Then $|\Pi|=O(h)$.
\end{lemma}
\begin{proof}
Let $\alpha_1=(A_1,a_1)$ and $\alpha_2=(A_2,a_2)$ be the two generators of $B$.
Recall that the well-covering region $\calU(e)$ of $e$ is the union of $O(1)$ cells of $\calS'$ and each cell has $O(1)$ elementary chain fragments on its boundary. Hence, $\calU(e)$ has $O(1)$ convex chains on its boundary.
Recall also that $W(e)$ is computed by merging all contributing wavefronts $W(f,e)$ for $f\in input(e)$ in the wavefront merging procedure, and $W(f,e)$ is computed by propagating $W(f)$ to $e$ through $\calU(e)$ in the wavefront propagation procedure.

We first argue that the number of pairs $(e,B)$ of $\Pi$ in the case where at least one of $\alpha_1$ and $\alpha_2$ is in $\calU(e)$ is $O(h)$. Indeed, according to our wavefront propagation procedure, if a subcell $c\in \calU(e)$ is an empty rectangle, then no new generators will be created when $W(f)$ is propagating through $c$; otherwise, although multiple generators may be created, at most two (one on each side of $c$) are in the wavefront existing $c$. As $\calU(e)$ has $O(1)$ cells and each cell may be partitioned into $O(1)$ subcells during the wavefront propagation procedure, only $O(1)$ generators of $W(e)$ are inside $\calU(e)$. Since each generator of $W(e)$ may define at most two bisectors in $W(e)$ and the total number of transparent edges of $\calS'$ is $O(h)$, the number of pairs $(e,B)$ of $\Pi$ such that at least one generator of $B$ is in $\calU(e)$ is at most $O(h)$.

In the following, we assume that both $\alpha_1$ and $\alpha_2$ are outside $\calU(e)$, i.e., their initial vertices $a_1$ and $a_2$ are outside $\calU(e)$. Let $q$ be the intersection of $B$ and $e$. Let $\pi'(a_1,q)$ denote the path from $q$ to $a_1$ following the tangent from $q$ to $A_1$ and then to $a_1$ along $A_1$; define $\pi'(a_2,q)$ similarly. Clearly, both $\alpha_1$ and $\alpha_2$ claim $q$ in $W(e)$.
Since $\alpha_1$ is outside $\calU(e)$, $\alpha_1$ must be in $W(f_1)$ for some transparent edge $f_1$ of $\partial \calU(e)$ so that $\alpha_1$ of $W(e)$ is from $W(f_1,e)$.
Since $\alpha_1$ is outside $\calU(e)$, $\pi'(a_1,q)$ must intersects $f_1$, say, at point $q_1$ (e.g., see Fig.~\ref{fig:pair10}), and $\pi'(q_1,q)$ is inside $\calU(e)$, where $\pi'(q_1,q)$ is the sub-path of $\pi'(a_1,q)$ between $q_1$ and $q$. Also, $f_1$ is processed (for the wavefront propagation procedure) earlier than $e$ because $W(f_1)$ contributes to $W(e)$.

We claim that $q_1$ is claimed by $\alpha_1$ in $\spm'(s)$. Assume to the contrary this is not true. Then,
let $\pi(s,q_1)$ be a shortest path from $s$ to $q_1$ in the free space $\calF$.

\begin{itemize}
\item
If $\pi(s,q_1)$ does not intersect $e$, then $\pi(s,q_1)$ is in the modified free space $\calF'$ by replacing $e$ with an opaque edge of open endpoints. Since $\alpha_1$ claims $q$ on $e$, $\pi=\pi_{W(e)}(s,a_1)\cup \pi'(a_1,q)$ is a shortest path from $s$ to $q$ in the modified free space $\calF'$, where $\pi_{W(e)}(s,a_1)$ is the path from $s$ to $a_1$ following the wavefront $W(e)$. Since $q_1$ is in $\pi$ and $q_1$ is not claimed by $\alpha_1$ in $\spm'(s)$, if we replace the portion between $s$ and $q_1$ in $\pi$ by $\pi(s,q_1)$, we obtain a shorter path from $s$ to $q$ in $\calF'$ than $\pi$. But this incurs a contradiction since $\pi$ is a shortest path from $s$ to $q$ in $\calF'$.

\item
If $\pi(s,q_1)$ intersects $e$, say, at a point $b$, then since $f_1\in \partial \calU(e)$ and thus $d(f_1,e)\geq 2\cdot \max\{|f_1|,|e|\}$ by the properties of the well-covering regions of $\calS'$, we show below that $e$ must be processed earlier than $f_1$, which incurs a contradiction because $f_1$ is processed earlier than $e$.

Indeed, $covertime(e)=\td(s,e)+|e|\leq d(s,b)+\frac{1}{2}\cdot |e| + |e|=d(s,b)+\frac{3}{2}\cdot |e|$. On the other hand, $covertime(f_1)=\td(s,f_1)+|f_1|\geq d(s,q_1)-\frac{1}{2}\cdot |f_1| + |f_1|=d(s,b)+d(b,q_1)+\frac{1}{2}\cdot |f_1|$. Since $f_1\in \partial \calU(e)$, $b\in e$, and $q_1\in f_1$, $d(b,q_1)\geq d(f_1,e)\geq 2\cdot |e|$. Hence, we obtain $covertime(f_1)>covertime(e)$ and thus $e$ must be processed earlier than $f_1$.
\end{itemize}

Therefore, $q_1$ must be claimed by $\alpha_1$ in $\spm(s)$.

We define $f_2$, $q_2$, and $\pi'(q,q_2)$ analogously with respect to $\alpha_2$. Similarly, we can show that $q_2$ is claimed by $\alpha_2$ in $\spm'(s)$.

\begin{figure}[t]
\begin{minipage}[t]{0.49\textwidth}
\begin{center}
\includegraphics[height=2.0in]{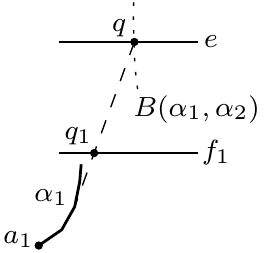}
\caption{\footnotesize Illustrating $q$ and $q_1$.}
\label{fig:pair10}
\end{center}
\end{minipage}
\hspace{0.02in}
\begin{minipage}[t]{0.49\textwidth}
\begin{center}
\includegraphics[height=2.2in]{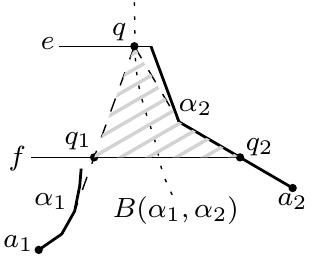}
\caption{\footnotesize Illustrating the psuedo-triangle $\triangle(q,q_1,q_2)$ (the shaded region).}
\label{fig:pair20}
\end{center}
\end{minipage}
\vspace{-0.15in}
\end{figure}

For each $f\in \partial \calU(e)$, the generators of $W(f)$ that are also in $W(e)$ may not form a single subsequence of the generator list of $W(e)$, but they must form a constant number of (maximal) subsequences. Indeed, since $\calU(e)$ is the union of $O(1)$ cells of $\calS'$, the number of islands in $\calU(e)$ is $O(1)$. Thus, the number of topologically different paths from $f$ to $e$ in $\calU(e)$ is $O(1)$; each such path will introduce at most one subsequence of generators of $W(e)$ that are also from $W(f)$. Therefore, the generators of $W(f)$ that are also in $W(e)$ form $O(1)$ subsequences of the generator list of $W(e)$.

Since $\partial\calU(e)$ has $O(1)$ transparent edges, the generator list of $W(e)$ can be partitioned into $O(1)$ subsequences each of which is from $W(f)$ following topologically the same  path for a single transparent edge $f$ of $\partial\calU(e)$. Due to this property, we have the following observation: the number of pairs of adjacent generators of $W(e)$ that are from different subsequences is $O(1)$.

Due to the above observation, the number of pairs of edges $f_1$ and $f_2$ on $\partial \calU(e)$ in the case  where $f_1\neq f_2$, or $f_1=f_2$ but $\pi'(q,q_1)$ and $\pi'(q,q_2)$ are topologically different in $\calU(e)$ is only $O(1)$. Therefore, the total number of pairs $(e,B)$ of $\Pi$ in that case is $O(h)$. In the following, it suffices to consider the case where $f_1= f_2$, and $\pi'(q,q_1)$ and $\pi'(q,q_2)$  are topologically the same in $\calU(e)$; for reference purpose, we use $\Pi'$ to denote the subset of pairs $(e,B)$ of $\Pi$ with the above property. Below we prove $|\Pi'|=O(h)$.

Let $f=f_1=f_2$. Since $\pi'(q,q_1)$ and $\pi'(q,q_2)$ are topologically the same in $\calU(e)$,
the region bounded by $\pi'(q,q_1)$, $\pi'(q,q_2)$, and $\overline{q_1q_2}$ must be in $\calU(e)$; we call the above region a {\em pseudo-triangle}, for both $\pi'(q,q_1)$ and $\pi'(q,q_2)$ are convex chains, and we use $\triangle(q,q_1,q_2)$ to denote it (e.g., see Fig.~\ref{fig:pair20}).
Because $(e,B)$ is not an incident pair in $\spm'(s)$, the point $q$ must be claimed by a different generator $\alpha$ in $\spm'(s)$, which must be from the side of $e$ different than $\alpha_1$ and $\alpha_2$. We have proved above that $q_1$ is claimed by $\alpha_1$ and $q_2$ is claimed by $\alpha_2$ in $\spm'(s)$. Hence, there must be at least one bisector event in $\spm'(s)$ that lies in the interior of the pseudo-triangle $\triangle(q,q_1,q_2)$\footnote{This is also due to that for any point $p\in \overline{q_1q_2}$, the shortest path $\pi(s,p)$ cannot intersect $e$; this can be proved by a similar argument to the above for proving that $q_1$ is claimed by $\alpha_1$ in $\spm'(s)$. This observation implies that $\alpha$ cannot claim any points on $f$ in $\spm'(s)$.}. We charge the early demise of $B$ to any one of these bisector events in $\triangle(p,q_1,q_2)$.

Note that the path $\pi'(q,q_1)$ (resp., $\pi'(q,q_2)$) is in a shortest path from $s$ to $q$ following the wavefront $W(e)$ in the modified free space by replacing $e$ with an opaque edge of open endpoints. Hence, if $\Pi'$ has other pairs $(e,B')$ whose first element is $e$, then the corresponding paths $\pi'(q,q_1)$ and $\pi'(q,q_2)$ of all these pairs are disjoint, and thus the corresponding pseudo-triangles $\triangle(q,q_1,q_2)$ are also disjoint in $\calU(e)$. Hence, each bisector event of $\spm'(s)$ inside $\triangle(q,q_1,q_2)$ is charged at most once for all pairs of $\Pi'$ that have $e$ as the first element. Since each cell of $\calS'$ belongs to the well-covering region $\calU(e)$ of $O(1)$ transparent edges $e$, each bisector event of $\spm'(s)$ is charged $O(1)$ times for all pairs of $\Pi'$. Because $\spm'(s)$ has $O(h)$ bisector events by Corollary~\ref{coro:complexity}, the number of pairs of $\Pi'$ is at most $O(h)$.

This completes the proof of the lemma.
\end{proof}

Armed with Lemma~\ref{lem:pairs}, we prove the subsequent five lemmas, which together lead to Lemma~\ref{lem:numgen}.

\begin{lemma}\label{lem:rule3b}
The total number of marked generators by Rule~2(a) and Rule~3 is $O(h)$.
\end{lemma}
\begin{proof}
Suppose $\alpha$ is a generator marked by Rule~2(a) for a transparent edge $e$. Then, $\alpha$ must be the first or last non-artificial generator of $W(e)$. Hence, at most two generators of $W(e)$ can be marked by Rule~2(a). As $\calS'$ has $O(h)$ transparent edges, the total number of generators marked by Rule~2(a) is $O(h)$.

Suppose $\alpha$ is a generator marked by Rule~3(a) for the two transparent edges $(e,g)$ with $g\in output(e)$. Since $\alpha$ claims an endpoint of $g$, $\alpha$ must be the first or last generator of $W(e,g)$. Hence, at most two generators of $W(e,g)$ can be marked by Rule~3(a). Since $|output(e)|=O(1)$, at most $O(1)$ generators can be marked for the pairs of transparent edges with $e$ as the first edge.  Since $\calS'$ has $O(h)$ transparent edges, the total number of generators marked by Rule~3(a) is $O(h)$.

Suppose $\alpha$ is a generator marked by Rule~3(b) for the two transparent edges $(e,g)$ with $g\in output(e)$. Note that $\alpha$ is a generator in $W(e)$. We assume that $\alpha$ is not the first or last non-artificial generators of $W(e)$; the first and last non-artificial generators of $W(e)$, countered separately, sum to $O(h)$ for all transparent edges $e$ of $\calS'$.
Let $\alpha_1$ and $\alpha_2$ be the two neighboring generators of $\alpha$ in $W(e)$. Thus, both $\alpha_1$ and $\alpha_2$ are non-artificial.
We assume that the bisectors $B(\alpha_1,\alpha)$ and $B(\alpha,\alpha_2)$ intersect $e$ in $\spm'(s)$; by Lemma~\ref{lem:pairs}, there are only $O(h)$ bisector and transparent edge intersections that appear in some approximate wavefront but not in $\spm'(s)$.

Since $\alpha$ is marked by Rule~3(b), at least one of $B(\alpha_1,\alpha)$ and $B(\alpha,\alpha_2)$ fails to reach the boundary of $D(e)$ during the algorithm, where $D(e)$ is the union of cells through which $W(e)$ is propagated to all edges $g'\in output(e)$. Note that $D(e)\subseteq \calU(e)\cup \bigcup_{g'\in output(e)}\calU(g')$, which contains $O(1)$ cells of $\calS'$ as $|output(e)|=O(1)$ and the well-covering region of each transparent edge is the union of $O(1)$ cells of $\calS'$; also, each cell of $D(e)$ is within a constant number of cells of $e$. Without loss of generality, we assume that $B(\alpha_1,\alpha)$ does not reach the boundary of $D(e)$ and a bisector event on $B(\alpha_1,\alpha)$ is detected during the algorithm. The detected bisector event on $B(\alpha_1,\alpha)$ also implies that an actual bisector event in $\spm'(s)$ happens to $B(\alpha_1,\alpha)$ no later than the detected event; we charge the marking of $\alpha$ to that bisector actual endpoint (i.e., a vertex) in $\spm'(s)$, and the endpoint is in $D(e)$ because $B(\alpha_1,\alpha)$ intersects $e$ in $\spm'(s)$.
Since each cell of $D(e)$ is within a constant number of cells of $e$, each cell of $\calS'$ belongs to $D(e')$ for a constant number of transparent edges $e'$ of $\calS'$. Therefore, each vertex of $\spm'(s)$ is charged $O(1)$ times. Since $\spm'(s)$ has $O(h)$ vertices by Lemma~\ref{lem:sizespm}, the total number of generators marked by Rule~3(b) is $O(h)$.
\end{proof}


\begin{lemma}\label{lem:rule1}
The total number of marked generators by Rule 1 is $O(h)$.
\end{lemma}
\begin{proof}
According to our algorithm, generators are only created during the wavefront propagation procedure, which is to propagate $W(e)$ to compute $W(e,g)$ for all edges $g\in output(e)$ through $\calU$, where $\calU$ is $\calU(e)$ or $\calU(g)$. The procedure has two cases depending on whether a cell $c$ of $\calU$ is an empty rectangle. If yes, then no generators will be created in $c$. Otherwise, $c$ may be partitioned into $O(1)$ subcells, each of which may have a convex chain on its left side and/or its right side. Let $c$ be such a subcell. Without loss of generality, we assume that the current wavefront $W$ is propagating in $c$ from bottom to top. We consider the generators created on the left side $\zeta_l$ of $c$.

According to our algorithm, a generator on $\zeta_l$ is created by the leftmost generator of the current wavefront $W$. As such, if the leftmost wavelet of $W$ does not create a generator, then no generator on $\zeta_l$ will be created. Otherwise, let $\alpha$ be a generator created on $\zeta_l$, which becomes the leftmost generator of $W$ at the point of creation. Let $\alpha'$ be the right neighbor of $\alpha$ in $W$. According to our algorithm, if $\alpha$ does not involve in any bisector event in $c$, i.e., the bisector $B(\alpha,\alpha')$ does not intersect any other bisector in $c$ during the propagation of $W$ in $c$, then no more new generators will be created on $\zeta_l$. Otherwise, we charge the creation of $\alpha$ to the bisector event involving $B(\alpha,\alpha')$; each such event can be charged at most twice (one for a generator on $\zeta_l$ and the other for a generator created on the right side of $c$). Recall that for each bisector event a generator is marked by Rule~3(b).

In light of the above discussion, since each cell of $\calS'$ belongs to the well-covering region $\calU(e')$ of $O(1)$ transparent edges $e'$, the total number of generators marked by Rule~1 is $O(h)+O(k)$, where $k$ is the total number of generators marked by Rule~3(b), which is $O(h)$ by Lemma~\ref{lem:rule3b}. The lemma thus follows.
\end{proof}

\begin{lemma}\label{lem:rule4}
The total number of marked generators by Rule 4 is $O(h)$.
\end{lemma}
\begin{proof}
Let $\alpha$ be the generator and $e$ be the transparent edge specified in the rule statement. According to Rule~4, $\alpha$ is marked because it claims part of an obstacle edge during the wavefront propagation procedure for propagating $W(e)$ to compute $W(e,g)$ for edges $g\in output(e)$.
As in the proof of Lemma~\ref{lem:rule3b}, define $D(e)$ as the union of cells through which $W(e)$ is propagated to all edges $g\in output(e)$. Recall that $D(e)$ contains $O(1)$ cells of $\calS'$, and each cell of $D(e)$ is within a constant number of cells of $e$, which implies that each cell of $\calS'$ belongs to $D(e')$ for a constant number of transparent edges $e'$ of $\calS'$.

First of all, for the case where $\alpha$ is a generator created during the wavefront propagation procedure, the total number of such marked generators is no more than the total number of marked generators by Rule~1, which is $O(h)$ by Lemma~\ref{lem:rule1}.

Second, for the case where $\alpha$ is the first or last generator in $W(e)$, the total number of such generators is clearly $O(h)$ since $\calS'$ has $O(h)$ transparent edges $e$.

Third, for the case where $\alpha$ claims a rectilinear extreme vertex, the total number of such marked generators is $O(h)$. To see this, since $D(e)$ has $O(1)$ cells of $\calS'$ and each cell contains at most one rectilinear extreme vertex, $D(e)$ contains $O(1)$ rectilinear extreme vertices. Therefore, at most $O(1)$ generators will be marked in $D(e)$. Since each cell of $\calS'$ belongs to $D(e')$ for $O(1)$ transparent edges $e'$ of $\calS'$ and $\calS'$ contains $O(h)$ transparent edges $e$, the total number of marked generators in this case is $O(h)$.

Finally, any Rule~4 marked generator $\alpha$ that does not belong to any of the above cases has the following property: $\alpha$ does not claim a rectilinear extreme vertex, and is not the first or last non-artificial generator in $W(e)$, and is not a generator created in $D(e)$. Let $\alpha'$ be a neighbor of $\alpha$ in $W(e)$. Due to the above property, $\alpha'$ is a non-artificial generator. We can assume that $B(\alpha',\alpha)$ intersects $e$ in $\spm'(s)$; by Lemma~\ref{lem:pairs}, there are only $O(h)$ bisector and transparent edge intersections that appear in some approximate wavefront but not in $\spm'(s)$. Hence, in $\spm'(s)$, $B(\alpha',\alpha)$ terminates in $D(e)$, either on an obstacle edge or in a bisector event before reaching an obstacle edge. In either case, we charge the mark of $\alpha$ at $e$ to this endpoint of $B(\alpha',\alpha)$ in $\spm'(s)$, which is vertex of $\spm'(s)$. Because each cell of $\calS'$ belongs to $D(e')$ for a constant number of transparent edges $e'$ of $\calS'$, each vertex of $\spm'(s)$ is charged at most $O(1)$ times. As $\spm'(s)$ has $O(h)$ vertices by Lemma~\ref{lem:sizespm}, the total number of marked generators by Rule~4 is $O(h)$.
\end{proof}

\begin{lemma}
The total number of marked generators by Rule 2(b) is $O(h)$.
\end{lemma}
\begin{proof}
Let $\alpha$ be the generator and $e$ be the transparent edge specified in the rule statement.

First of all, the total number of marked generators in the case where $\alpha$ is in $\calU(e)$ is $O(h)$, because the number of Rule~1 marked generators is $O(h)$ by Lemma~\ref{lem:rule1} and $\calU(e)$ has $O(1)$ cells. In the following, we consider the other case where $\alpha$ is outside $\calU(e)$.

Since $\alpha$ is outside $\calU(e)$, there is a transparent edge $f$ on $\partial \calU(e)$ (i.e., $f\in input(e)$), such that $\alpha$ is in $W(f)$ and it is in $W(e)$ because $W(f)$ is propagated to $e$ through $\calU(e)$. Since $\alpha$'s claim is shortened or eliminated by an artificial wavefront, there must be a bisector event involving $\alpha$ during the computation of $W(f,e)$ from $W(f)$. Hence, $\alpha$ is also marked by Rule~3(b). We charge $\alpha$ to this Rule~3(b) mark for $f$. Since there are $O(h)$ generators marked by Rule~3(b) by Lemma~\ref{lem:rule3b}, the total number of marked generators of Rule~2(b) is $O(h)$.
\end{proof}

\subsection{Computing the shortest path map $\spm(s)$}
\label{sec:spm}

In this section, we compute the shortest path map $\spm(s)$. To this end, we need to mark generators following our rules during our wavefront merging and propagation procedures. As the total number of all marked generators is $O(h)$, marking these generators do not change the running time of our algorithm asymptotically. In the following, we show that $\spm(s)$ can be constructed in $O(n+h\log h)$ time with the help of the marked generators. In light of Lemma~\ref{lem:spmconstruction}, we will focus on constructing $\spm'(s)$.

We first show in Section~\ref{sec:markenough} that the marked generators are sufficient in the sense that if a generator participates in a true bisector event of $\spm'(s)$ in a cell $c$ of $\calS'$, then $\alpha$ must be marked in $c$. Then, we present the algorithm for constructing $\spm'(s)$, which consists of two main steps.
The first main step is to identify all vertices of $\spm'(s)$ and the second one is to compute the edges of $\spm'(s)$ and assemble them to obtain $\spm'(s)$. The first main step is described in Section~\ref{sec:first} while the second one is discussed in Section~\ref{sec:second}.

\subsubsection{Correctness of the marked generators}
\label{sec:markenough}

In this section we show that if a generator participates in a bisector event of $\spm'(s)$ in a cell $c$ of $\calS'$, then $\alpha$ must be marked in $c$. Our algorithm in the next section for computing $\spm'(s)$ relies on this property. Our proof strategy again follows the high-level structure of the HS algorithm. We first prove the following lemma.

\begin{lemma}\label{lem:markcell}
Let $\alpha$ be a generator in an approximate wavefront $W(e)$ for some transparent edge $e$. Suppose there is a point $p\in e$ that is claimed by $\alpha$ in $W(e)$ but not in $\spm'(s)$ (because the approximate wavefront from the other side of $e$ reaches $p$ first). Then, $\alpha$ is marked in the cell $c$ on $\alpha$'s side of $e$.
\end{lemma}
\begin{proof}
Assume to the contrary that $\alpha$ is not marked in $c$. Then, $\alpha$ must have two neighboring non-artificial generators $\alpha_1$ and $\alpha_2$ in $W(e)$ since otherwise Rule~2 would apply.
Also, because all transparent edges of $\partial \calU(e)$ are in $output(e)$, the two bisectors $B(\alpha_1,\alpha)$ and $B(\alpha,\alpha_2)$ must exist the well-covering region $\calU(e)$ of $e$ through the same transparent edge $g\in \partial \calU(e)$, since otherwise Rule~3 or 4 would apply.
Further, the region $R$ bounded by $B(\alpha_1,\alpha)$, $B(\alpha,\alpha_2)$, $e$, and $g$ must be a subset of $\calU(e)$. Indeed, if $R$ contains an island not in $\calU(e)$, then $\alpha$ would claim an endpoint of a boundary edge of the island, in which case Rule~3 would apply.

Let $\alpha'=(A,a)$ be the true predecessor of $p$ in $\spm'(s)$.
Without loss of generality, we assume that $e$ is horizontal and
$\alpha$ is below $e$ and thus $\alpha'$ is above $e$.
Let $\pi'(a,p)$ be the path from
$p$ along its tangent to $A$ and then following $A$ to $a$.

We first consider the case where $\alpha'$ is not in the interior of
the well-covering region $\calU(e)$, i.e., the initial vertex $a$ is
not in the interior of $\calU(e)$. Since $R$ is a subset of of
$\calU(e)$ without non-$\calU(e)$ islands, $a$ is not in the
interior of $R$ and thus $\pi'(a,p)$ intersects $\partial R$, say,
at a point $q$. In the following, we argue that $\alpha$ has been
involved in a bisector event detected by our algorithm, and thus
marked in $c$. Let $g'$ be the portion of $g$ between its intersections with
$B(\alpha_1,\alpha)$ and $B(\alpha,\alpha_2)$. Depending on whether
$q\in g'$, there are two subcases.

\begin{itemize}
\item
If $q\in g'$ (e.g., see Fig.~\ref{fig:markcell}), then $\td(s,g)\leq d(s,a)+|\pi'(a,q)|+|g|/2$, where
$\pi'(a,q)$ is the sub-path of $\pi'(a,p)$ between $a$ and $q$. Hence,
the time $\tau_g$ when artificial wavefronts originating from the
endpoints of $g$ cover $g$ is no later than $\td(s,g)+|g|\leq
d(s,a)+|\pi'(a,q)|+3|g|/2$. Because $g\in \partial \calU(e)$,
$|\overline{pq}|\geq 2\cdot \max\{|g|,|e|\}$. Hence, $\tau_g\leq
d(s,a)+|\pi'(a,q)|+3|g|/2<
d(s,a)+|\pi'(a,q)|+|\overline{qp}|=d(s,a)+|\pi'(a,p)|<d_{W(e)}(s,p)$, where
$d_{W(e)}(s,p)$ is the length of the path from $s$ to $p$ following the
wavefront $W(e)$; the last inequality holds because $\alpha'$ is the
true predecessor of $p$ while $\alpha$ is not.

\begin{figure}[t]
\begin{minipage}[t]{0.49\textwidth}
\begin{center}
\includegraphics[height=2.1in]{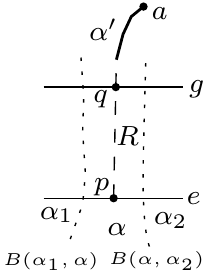}
\caption{\footnotesize Illustrating the case where $q\in g'$.}
\label{fig:markcell}
\end{center}
\end{minipage}
\hspace{0.02in}
\begin{minipage}[t]{0.49\textwidth}
\begin{center}
\includegraphics[height=2.1in]{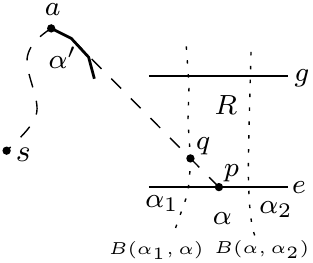}
\caption{\footnotesize Illustrating the case where $q\in B(\alpha_1,\alpha)$.}
\label{fig:markcell10}
\end{center}
\end{minipage}
\vspace{-0.15in}
\end{figure}

On the other hand, the time $\tau_e$ when the wavefront $W(e)$
reaches an endpoint of $e$ cannot be earlier than $d_{W(e)}(s,p)-|e|$. Hence,
the wavelet from $\alpha$ cannot reach $g$ earlier than
$d_{W(e)}(s,p)-|e|+d(e,g)\geq  d_{W(e)}(s,p)-|e|+2|e|\geq  d_{W(e)}(s,p)+|e|$,
which is larger than $\tau_g$ since $\tau_g<d_{W(e)}(s,p)$ as proved
above. Therefore, by the time the wavelet from
$\alpha$ reaches $g$, the artificial wavelets of $g$ have already
covered $g$, eliminating the wavelet from $\alpha$ from reaching $g$.
Thus, $\alpha$ must be marked by Rule~3(b).

\item

If $q\not\in g'$, then $q$ is on one of the bisectors
$B(\alpha_1,\alpha)$ and $B(\alpha,\alpha_2)$. Let $\pi(s,a)$ be a shortest path from $s$ to $a$ and let $\pi(s,p)=\pi(s,a)\cup \pi'(a,p)$, which is a shortest path from $s$ to $p$.
If $\pi(s,p)$ intersects $g$, then we can use the same analysis as above to show that the wavelet from $\alpha$ will be eliminated from reaching $g$ by the artificial wavelets of $g$ and thus $\alpha$ must be marked by Rule~3(b).
In the following, we assume that $\pi(s,p)$ does not intersect $g$.

Without of loss of generality, we assume that $q\in B(\alpha_1,\alpha)$ (e.g., see Fig.~\ref{fig:markcell10}). Since
$\pi'(a,p)$ is a subpath of the shortest path $\pi(s,p)$, every point of
$\pi'(a,p)$ has $\alpha'$ as its predecessor. As $q\in
\pi'(a,p)$, the predecessor of $q$ is $\alpha'$.
Let $\pi(s,q)$ be the subpath of $\pi(s,p)$ between $s$ and $q$, which is a shortest path. Since $\pi(s,p)$ does not intersect $g$, $\pi(s,q)$ does not intersect $g$. Hence, $\pi(s,q)$ is a shortest path in the modified free space $\calF'$ by replacing $g$ with an opaque edge of open endpoints. Therefore, during the wavefront propagation procedure for computing $W(e,g)$ or during the wavefront merging procedure for computing $W(g)$, the point $q$, which is on the bisector $B(\alpha_1,\alpha)$, must be claimed by $\alpha'$. Hence, a bisector event must be detected for
$B(\alpha_1,\alpha)$ during the computation of $W(e,g)$ or $W(g)$. In
either case, $\alpha$ must be marked by Rule~3(b).
\end{itemize}

We next consider the case where $\alpha'$ lies inside $\calU(e)$, i.e., its initial vertex $a$ is inside $\calU(e)$.
If $a$ is not between the two bisectors $B(\alpha_1,\alpha)$ and $B(\alpha,\alpha_2)$, then $\pi'(a,p)$ must intersect one of the bisectors and thus we can use a similar argument as above second case to show that $\alpha$ must be marked. Otherwise, $a$ is in the region $R$.
This implies that not all points in $R$ are claimed by $\alpha$ when $W(e)$ is propagating to $g$, and therefore, a bisector event involving $\alpha$ must happen during the wavefront propagation procedure to propagate $W(e)$ to compute $W(e,g)$ and thus $\alpha$ is marked by Rule~3(b).
\end{proof}

With the help of the preceding lemma, we prove the following lemma.

\begin{lemma}\label{lem:markcorrect}
If a generator $\alpha$ participates in a bisector event of $\spm'(s)$ in a cell $c$ of $\calS'$, then $\alpha$ must be marked in $c$.
\end{lemma}
\begin{proof}
If a bisector has an endpoint on an obstacle edge of $c$, it either
emanates from an obstacle vertex $a$ on the edge (i.e., the bisector
is an extension bisector) or defined by two
generators that claim part of the opaque edge. In the first case, $a$
is the initial vertex of a new created generator in $c$ and thus the generator
is marked by Rule~1. In the second case, both generators are marked by
Rule~4.

Let $\alpha$ be a generator that participates in a bisector event of
$\spm'(s)$ in a cell $c$ of $\calS'$. Assume to the contrary that $\alpha$ is not marked for
$c$. Then by Rule~2(a) there must be transparent edges $e$ and $f$ on
the boundary of $c$ such that both $W(e)$ and $W(f)$ contain the three
consecutive generators $\alpha_1$, $\alpha$, and $\alpha_2$. Without
loss of generality, we assume that $W(e)$ enters $c$ and $W(f)$ leaves
$c$. Let $R$ be the region of $c$ bounded by $e$, $f$, and the two bisectors
$B(\alpha_1,\alpha)$ and $B(\alpha,\alpha_2)$; e.g., see Fig.~\ref{fig:markcorrect}. The region $R$ must be a subset of $c$. Indeed, if $R$ contains an island not in $c$, then $\alpha$ would claim an endpoint of a boundary edge of the island, in which case Rule~3 would apply.

\begin{figure}[t]
\begin{minipage}[t]{\textwidth}
\begin{center}
\includegraphics[height=2.2in]{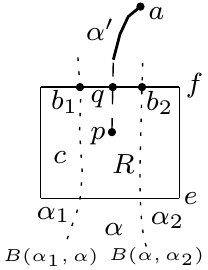}
\caption{\footnotesize Illustrating the proof of Lemma~\ref{lem:markcorrect}.}
\label{fig:markcorrect}
\end{center}
\end{minipage}
\vspace{-0.15in}
\end{figure}

Since $\alpha$ participates in a bisector event of $\spm'(s)$ in $c$, at least one
point $p$ in $R$ is not claimed by $\alpha$ in $\spm'(s)$. Let
$\alpha'=(A,a)$ be the true predecessor of $p$. Note that the initial vertex $a$ must be outside $R$ since otherwise a bisector event involving $\alpha$ must happen when $W(e)$ is propagating through $c$ and thus $\alpha$ would be marked by Rule~3(b). Let $\pi'(a,p)$ be the path from
$p$ along its tangent to $A$ and then following $A$ to $a$.
Since $a$ is outside $R$ and $p$ is inside $R$, $\pi'(a,p)$ must intersect the boundary of $R$.

Because $\alpha_1$, $\alpha$, and $\alpha_2$ are three consecutive
generators of $W(e)$, no generator other than $\alpha$ on the same side of $e$ as $\alpha$
claims any point of $R$. Thus, $\pi'(a,p)$ does not cross $e$. Let
$b_1$ and $b_2$ be the intersections of $f$ with $B(\alpha_1,\alpha)$ and $B(\alpha,\alpha_2)$, respectively. Then, if $\alpha$ claims both $b_1$ and $b_2$ in $\spm'(s)$, then $\pi'(a,p)$ cannot cross either bisector on $\partial R$, and thus it must cross $\overline{b_1b_2}$, say, at a point $q$ (e.g., see Fig.~\ref{fig:markcorrect}). Note that $q$ satisfies the
hypothesis of Lemma~\ref{lem:markcell} and thus $\alpha$ is marked for
$c$. It $\alpha$ does not claim either $b_1$ or $b_2$, then that point
satisfies the hypothesis of Lemma~\ref{lem:markcell} and thus $\alpha$
is marked for $c$.
\end{proof}

\subsubsection{Computing the vertices of $\spm'(s)$}
\label{sec:first}

In this section, we compute the vertices of $\spm'(s)$. The next section will compute the edges of $\spm'(s)$ and assemble them to obtain $\spm'(s)$.

\paragraph{Computing active regions.}
Because the approximate wavefronts are represented by persistent trees, after the above wavefront expansion algorithm finishes, the approximate wavefronts $W(e)$ for all transparent edges $e$ of $\calS'$ are still available. Also, for each cell $c$ and each transparent edge $e$ of $c$, a set of marked generators in $W(e)$ are known. Using these marked generators, we first break $c$ into {\em active} and {\em inactive} regions such that no vertices of $\spm'(s)$ lie in the inactive regions. Note that each unmarked generator does not participate in a bisector event in $\spm'(s)$ by Lemma~\ref{lem:markcorrect}. If $\alpha_1$ and $\alpha_2$ are neighboring generators on $\partial c$ such that one of them is marked while the other is unmarked, their bisector belongs to $\spm'(s)$. Therefore, all such generators are disjoint and they together partition $c$ into regions such that each region is claimed only by marked generators or only by unmarked generators; the former regions are active while the latter are inactive. We can compute these active regions as follows.

Since the wavefronts $W(e)$ for all transparent edges $e$ of $c$ are available, the list of all generators ordered along the boundary of $c$ is also available. For each pair of adjacent generators $\alpha_1$ and $\alpha_2$, if one of them is marked and the other is unmarked, then we explicitly compute the portion of their bisector $B(\alpha_1,\alpha_2)$ in $c$. We describe the details of this step below.

First of all, observe that $\alpha_1$ and $\alpha_2$ must be from $W(e)$ for the same transparent edge $e$ of $c$, since otherwise each generator must claim an endpoint of their own transparent edge thus must have been marked by Rule~2(a).
Without loss of generality, we assume that $e$ is horizontal and both $\alpha_1$ and $\alpha_2$ are below $e$. Note that we cannot afford computing the entire bisector $B(\alpha_1,\alpha_2)$ and then determining its portion in $c$ because the running time would be proportional to the size of $B(\alpha_1,\alpha_2)$, i.e., the number of hyperbolic-arcs of $B(\alpha_1,\alpha_2)$. Instead, we first compute the intersection $b$ of $B(\alpha_1,\alpha_2)$ and $e$, which can be done in $O(\log n)$ time by the bisection-line intersection operation in Lemma~\ref{lem:bl-intersection}. Note that $b$ is one endpoint of the portion of $B(\alpha_1,\alpha_2)$ in $c$. To determine the other endpoint $b'$, we do the following. Our goal is to compute $b'$ in $O(|B_c|+\log n)$ time, with $B_c=B(\alpha_1,\alpha_2)\cap c$. To this end, we could trace $B_c$ in $c$ from $b$, and for each hyperbolic-arc of $B_c$, we determine whether it intersects $\partial c$. Recall that $\partial c$ consists of $O(1)$ transparent edges and at most $O(1)$ convex chains. To achieve the desired runtime, we need to determine whether each hyperbolic-arc of $B_c$ intersects $\partial c$ in $O(1)$ time. However, it is not clear to us whether this is possible as the size of each convex chain may not be of $O(1)$ size. To circumvent the issue, we use the following strategy. Before tracing $B_c$, we first compute the intersection between $B(\alpha_1,\alpha_2)$ and each convex chain on $\partial c$, which can be done in $O(\log n)$ time by the bisector-chain intersection operation in Lemma~\ref{lem:bc-intersection}. Among all intersections, let $b''$ be the closest one to $\alpha$. Then, we start to trace $B_c$ from $b$, for each hyperbolic-arc $e'$, we determine whether $e'$ intersects each of the transparent edges of $\partial c$. If not, we further check whether $e'$ contains $b''$. If $b''\not\in e'$, then $e'$ is in $c$ and we continue the tracing; otherwise, $b'$ is $b''$ and we stop the algorithm. If $e'$ intersects at least one of the transparent edges of $\partial c$, then among all such intersections as well as $b''$, $b'$ is the one closest to $\alpha$, which can be determined in $O(\log n)$ time by computing their tangents to $\alpha$.
In this way, computing the portion of $B(\alpha_1,\alpha_2)$ in $c$ can be done in $O(\log n+n_c(\alpha_1,\alpha_2))$ time, where $n_c(\alpha_1,\alpha_2)$ is the number of hyperbolic-arcs of $B(\alpha_1,\alpha_2)$ in $c$.

In this way, all active regions of $c$ can be computed in $O(h_c\log
n+n_c)$ time, where $h_c$ is the total number of marked generators in the
wavefronts $W(e)$ of all transparent edges $e$ of $c$ and $n_c$ is the
total number of hyperbolic-arcs of the bisector boundaries of these
active regions.

\paragraph{Computing the vertices of $\spm'(s)$ in each active region.}
In what follows, we compute the vertices of $\spm'(s)$ in each active region $R$ of $c$.

The boundary $\partial R$ consists of $O(1)$ pieces, each of which is
a transparent edge fragment, an elementary chain fragment, or a bisector in
$\spm'(s)$. Unlike the HS algorithm, where each of these pieces is of
$O(1)$ size, in our problem both an elementary chain fragment and a bisector
may not be of constant size. Let $e$ be a transparent edge of $R$.
Without loss of generality, we assume that $e$ is horizontal and $R$
is locally above $e$. Without considering other wavefronts, we use
$W(e)$ to partition $R$ into subregions, called {\em Voronoi faces}, such
that each face has a unique predecessor in $W(e)$. We use $Vor(e)$
to denote the partition of $R$. Note that if $\alpha$ is the
predecessor of a Voronoi face, then for each point $p$ in the face, its
tangent to $\alpha$ must cross $e$ and we assign $p$ a {\em weight}
that is equal to $d(\alpha,p)$.
Also, it is possible that these faces together may not cover the entire $R$; for those points
outside these faces, their weights are $\infty$.

We now discuss how to compute the partition $Vor(e)$. To this
end, we can use our wavefront propagation procedure to propagate
$W(e)$ inside $R$. To apply the algorithm, one issue
is that the boundary of $R$ may contain bisectors, which consists of hyperbolic-arcs instead of polygonal segments. To circumvent the issue, an observation is that the bisectors of $W(e)$ do not intersect
the bisectors on $\partial R$. Indeed, for each bisector $B$ of $\partial R$, one of its defining generators is unmarked and thus does not participate in any bisector event in $c$, and therefore, $B$ does not intersect any bisector in $c$. In light of the observation, we can simply apply the
wavefront propagation procedure to propagate $W(e)$ to $c$ instead of
$R$ to partition $c$ into Voronoi faces, denoted by $Vor_c(e)$. The above
observation implies that each bisector on the boundary of $R$ must lie
in the same face of $Vor_c(e)$. Hence, we can simply cut $Vor_c(e)$ to
obtain $Vor(e)$ using the bisectors on the boundary of $R$. When we
apply the wavefront propagation procedure to propagate $W(e)$ to $c$,
here we add an initial step to compute the intersection of $e$ and $B(\alpha,\alpha')$ for each
pair of neighboring generators $\alpha$ and $\alpha'$ of $W(e)$ and
use it as the initial tracing-point $z(\alpha,\alpha')$. Since each such
intersection can be computed in $O(\log n)$ time by the bisector-line
intersection operation in Lemma~\ref{lem:bl-intersection}, this initial step takes $O(h_{e}\log n)$ time, where $h_e$ is the number of generators of $W(e)$.
Hence, the total time for propagating $W(e)$ into $c$ to obtain
$Vor_c(e)$ is $O(h_e\log n+n_B(e))$, where $n_B(e)$ is the number of hyperbolic-arcs of the bisectors of $W(e)$ in $c$.
After having $Vor_c(e)$, cutting it to obtain $Vor(e)$ can be easily done in $O(h_e+n_B(e)+n_R)$ time, where $n_R$ is the number of hyperbolic-arcs on the bisectors of $\partial R$.
Therefore, the total time for computing the partition $Vor(e)$ is
$O(h_e\log n+n_B(e)+n_R)$ time.
Note that since the bisectors of $\partial R$ do not intersect any bisectors involving the generators of $W(e)$, all bisector events of $W(e)$ in $c$ are actually in $R$.

After having the partition $Vor(e)$ for all transparent edges $e$ of
$R$, we can now compute vertices of $\spm'(s)$ in $R$. Consider a transparent edge $e$ of $\partial R$. We process it as
follows. For each transparent edge $f$ of $\partial R$ other than $e$,
we merge the two partitions $Vor(e)$ and $Vor(f)$ by using the merge step
from the standard divide-and-conquer Voronoi diagram algorithm to
compute the sub-region of $R$ closer to $W(e)$ than to $W(f)$. This can be
done in $O(h_e+n_B(e)+n_R+h_f+n_B(f))$ time by finding a finding a curve
$\gamma$ in $R$ that consists of points equal to $W(e)$ and $W(f)$,
and the algorithm is similar to the HS algorithm.

Intersecting the results for all such $f$ produces the region $R(e)$ claimed by $W(e)$ in $R$.
Intersecting $R(e)$ with $Vor(e)$ gives the vertices of $\spm'(s)$ in $R$ that $W(e)$ contributes.
Repeating the above for all transparent edges $e$ of $\partial R$ gives the vertices of $\spm'(s)$ in $R$.
Since $\partial R$ has $O(1)$ transparent edges, the total time is $O(h_R\log n+n_B(R)+n_R)$, where $h_R$ is the number of generators in all wavefronts of all transparent edges of $\partial R$ and $n_B(R)$ is the total number of hyperbolic-arcs of the bisectors all wavefronts of all transparent edges of $\partial R$.


We do the above for all active regions $R$ of $c$, and then all
vertices of $\spm'(s)$ in $c$ can be computed. The total time is $O(h_c\log
n+n_R(c)+n_B(c))$, where $h_c$ is the total number of marked generators in the
wavefronts $W(e)$ of all transparent edges $e$ of $c$, $n_R(c)$ is the
total number of hyperbolic-arcs of the bisector boundaries of the
active regions of $c$, and $n_B(c)$ is the total number of hyperbolic-arcs of the bisectors of the wavefronts of all transparent edges of $c$.
Processing all cells $c$ of $\calS'$ as above gives all vertices of
$\spm'(s)$. For the running time, the total sum of $n_R(c)$ among all cells $c\in \calS'$ is $O(n)$ because each hyperbolic-arc of a bisector on the boundary of any active region also appears in $\spm'(s)$, whose total size is $O(n)$ by Corollary~\ref{coro:size}. The total sum of $h_c$ among all cells $c$ is $O(h)$ by Lemma~\ref{lem:numgen}.
For $n_B(c)$, Lemma~\ref{lem:bisectorvertices} is essentially a proof that the total sum of $n_B(c)$ over all cells $c$ is $O(n)$. To see this, since $c$ is a cell incident to $e$, the wavefront $W(e)$ will be propagated into $c$ during the wavefront propagation procedure to propagate $W(e)$ to edges of $output(e)$, and thus the hyperbolic-arcs of the bisectors of $W(e)$ are counted in the proof analysis of Lemma~\ref{lem:bisectorvertices}.
In summary, the total time for computing all vertices of $\spm'(s)$ is $O(n+h\log n)$, which is $O(n+h\log h)$.

\subsubsection{Constructing $\spm'(s)$}
\label{sec:second}

With all vertices of $\spm'(s)$ computed above, as in the HS
algorithm, we next compute the edges of $\spm'(s)$ separately and then assemble
them to obtain $\spm'(s)$. One difference is that the HS algorithm
uses a standard plane
sweep algorithm to assemble all edges to obtain $\spm(s)$, which takes
$O(n\log n)$ time as there are $O(n)$ edges in $\spm(s)$. In order to
achieve the desired $O(n+h\log h)$ time bound, here instead we propose a
different algorithm.

We first discuss how to compute the edges of $\spm'(s)$, given the vertices of $\spm'(s)$.
Recall that each vertex $v$ of $\spm'(s)$ is either an intersection
between a bisector and an obstacle edge, or an intersection of
bisectors. By our general position assumption, in the former case, $v$
is in the interior of an obstacle edge; in the latter case, $v$ is the
intersection of three bisectors, i.e., a {\em triple point}.
During our above algorithm for computing vertices of $\spm'(s)$, we
associate each vertex $v$ with the two generators of the bisector that
contains $v$ (more specifically, we associate $v$ with the initial
vertices of the generators).
We create a list of all vertices of $\spm'(s)$, each identified by a
{\em key} consisting of its two generators. If a vertex $v$ is a
triple point, then we put it in the list for three times (each time
with a different pair of generators); if $v$ is a bisector-obstacle
intersection, then it is put in the list once.
We now sort all vertices of $\spm'(s)$ with their keys; by traversing the sorted list, we can group together all vertices belong to the same bisector.
This sorting takes $O(h\log n)$ time as there are $O(h)$ vertices in
$\spm'(s)$.

We take all $\spm'(s)$-vertices in the same bisector and sort them
along the bisector determined by their weighted distances from the
generators of the bisector (each of these distances can be computed in
additional $O(\log n)$ time by computing the tangents from the vertex
to the generators). These sorting takes $O(h\log n)$ time altogether. The
above computes the bisector edges $e$ of $\spm'(s)$ that connect two adjacent
vertices. In fact, each edge $e$ is only implicitly determined in the
sense that the hyperbolic-arcs of $e$ are not explicitly computed.
In addition, for each obstacle $P$, we sort all vertices of $\spm'(s)$
on $\partial P$ to compute the convex-chain edges of $\spm'(s)$
connecting adjacent vertices of $\spm'(s)$ on $\partial P$.
This sorting can be easily done in $O(n+h\log n)$ time for all
obstacles.

For each vertex $v$ of $\spm'(s)$, the above computes its adjacent
vertices. As
discussed above, due to the general position assumption, $v$ has at
most three adjacent vertices. We next define a set $E(v)$ of at most three points for $v$.
For each adjacent vertex $u$ of $v$ in $\spm'(s)$, let $e(v,u)$ be the edge of $\spm'(s)$ connecting them.
$e(v,u)$ is either a bisector edge or a convex-chain edge. In the former case, we refer to each hyperbolic-arc of $e(v,u)$ as a {\em piece} of $e(v,u)$; in the latter case, we refer to each obstacle edge of $e(v,u)$ as a {\em piece}. If we traverse $e(v,u)$ from $v$ to $u$, the endpoint $u'$ of the first piece (incident to $v$) is added to $E(v)$; we define $h(u')$ to be $u$. Since $v$ is adjacent to at most three vertices in $\spm'(s)$, $|E(v)|\leq 3$. In fact, $E(v)$ is exactly the set of vertices adjacent to $v$ in the shortest path map $\spm(s)$.

With all edges of $\spm'(s)$ and the sets $E(v)$ of all vertices $v$ of $\spm'(s)$ computed above,
we construct the doubly-connected-edge-list (DCEL) of $\spm'(s)$, as follows.
For convenience, we assume that there is a bounding box that contains all obstacles so that no bisector edge of $\spm'(s)$ will go to infinity and thus each face of $\spm'(s)$ is bounded.


\begin{figure}[t]
\begin{minipage}[t]{\textwidth}
\begin{center}
\includegraphics[height=1.5in]{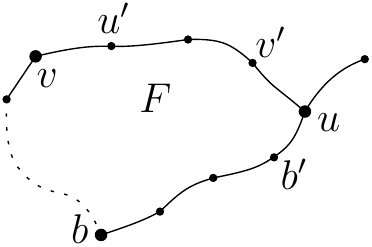}
\caption{\footnotesize Illustrating $v$, $u$, $u'$, $v'$, $b'$, and $b$.}
\label{fig:traceedge}
\end{center}
\end{minipage}
\vspace{-0.15in}
\end{figure}

For each vertex $v$ of $\spm'(s)$, we store its coordinate in the DCEL data structure. As
$|E(v)|\leq 3$, there are at most three faces in $\spm'(s)$ incident
to $v$, we construct them one by one. For each point $u'$ in $E(v)$, we
construct the face $F$ clockwise (with respect to $v$) incident to the
edge $e(v,u)$ of $\spm'(s)$ connecting $v$ and $u=h(u')$, as follows (e.g., see Fig.~\ref{fig:traceedge}).
We first trace out the edge $e(v,u)$, which is either a bisector edge or a convex chain edge.
In the former case, we compute the hyperbolic-arcs of $e(v,u)$, one
at a time, until we reach $u$, and add them to the DCEL data structure. Each
hyperbolic-arc can be computed in constant time using the two
generators of the bisector. In the latter case, we trace out the
obstacle edges of $e(v,u)$, one at a time, and add them to the DCEL data
structure. Let $v'$ be the last endpoint of the piece of $e(v,u)$ containing $u$ (e.g., see Fig.~\ref{fig:traceedge}). Note that $v'$ is in the set $E(u)$.
Let $b'$ be the first point of $E(u)$ counterclockwise around $u$ after $v'$
(e.g., see Fig.~\ref{fig:traceedge}). Let $b=h(b')$, which is a vertex
of $\spm'(s)$ adjacent to $u$.
Hence, the edge $e(u,b)$ of $\spm'(s)$ connecting $u$ to $b$ is incident to the
face $F$. We trace out the edge $e(u,b)$ in the same way as above. When we
reach $b$, we continue to trace the next edge of $F$.
Since each face of $\spm'(s)$ is bounded, eventually we
will arrive back to the vertex $v$ again\footnote{We could lift the assumption that each face of $\spm'(s)$ is bounded in the following way. During the above algorithm for constructing $F$, if $F$ is unbounded, then we will reach a bisector edge that extends to the infinity. If that happens, then we construct other edges of $F$ from the other direction of $v$. More specifically, starting from the first point of $E(v)$ clockwise around $v$ after $u'$, we trace out the edges of $F$ in the same way as above, until we reach a bisector edge that extends to the infinity, at which moment all edges of $F$ are constructed.}. This finishes the construction of the
face $F$. We do the same for all other faces of $\spm'(s)$, after
which the DCEL data structure for $\spm'(s)$ is constructed.
For the running time, since each edge of $\spm'(s)$ is traced at
most twice, by Corollary~\ref{coro:size}, the total time of the above procedure for constructing the DCEL data structure is $O(n)$.

In summary, $\spm'(s)$ can be computed in $O(n+h\log h)$ time. By
Lemma~\ref{lem:spmconstruction}, the shortest path map $\spm(s)$ can be built
in additional $O(n)$ time.

\subsection{Reducing the space to $O(n)$}
\label{sec:space}

The above provides an algorithm for computing $\spm(s)$ in $O(n+h\log h)$ time and $O(n+h\log h)$ space. In this subsection, we discuss how to reduce to the space to $O(n)$, using the technique given in~\cite{ref:WangSh21}.

The reason that the above algorithm needs $O(n+ h\log h)$ space is two-fold. First, it uses fully persistent binary trees (with the path-copying method) to represent wavefronts $W(e)$. Because there are $O(h)$ bisector events in the wavefront expansion algorithm and each event costs $O(\log n)$ additional space on a persistent tree, the total space of the algorithm is $O(n+h\log n)$. Second, in order to construct $\spm(s)$ after the wavefront expansion algorithm, the wavefronts $W(e)$ of all transparent edges $e$ of $\calS'$ are needed, which are maintained in those persistent trees.
We resolve these two issues in the following way.

\subsubsection{Reducing the space of the wavefront expansion algorithm}

We still use persistent trees to represent wavefronts. However, as there are $O(h)$ bisector events in total in the algorithm, we divide the algorithm into $O(\log h)$ phases so that each phase has no more than $h/\log h$ events. The total additional space for processing the events using persistent trees in each phase is $O(h)$. At the end of each phase, we ``reset'' the space of the algorithm by only storing a ``snapshot'' of the algorithm (and discarding all other used space) so that (1) the snapshot contains sufficient information for the subsequent algorithm to proceed as usual, and (2) the total space of the snapshot is $O(h)$.

Specifically, we make the following changes to the wavefront propagation procedure, which is to compute the wavefronts $W(e,g)$ for all edges $g\in output(e)$ using the wavefront $W(e)$.
We now maintain a counter $count$ to record the number of bisector events that have been processed so far since the last space reset; $count=0$ initially. Consider a wavefront propagation procedure on the wavefront $W(e)$ of a transparent edge $e$. The algorithm will compute $W(e,g)$ for all edges $g\in output(e)$, by propagating $W(e)$. We apply the same algorithm as before. For each bisector event, we first do the same as before. Then, we increment $count$ by one. If $count< h/\log h$, we proceed as before (i.e., process the next event). Otherwise, we have reached the end of the current phase and will start a new phase. To do so, we first reset $count=0$ and then reset the space by constructing and storing a snapshot of the algorithm (other space occupied by the algorithm is discarded), as follows.

\begin{enumerate}
\item
Let $g$ refer to the edge of $output(e)$ whose $W(e,g)$ is currently being computed in the algorithm.
 We store the tree that is currently being used to compute $W(e,g)$ right after the above event. To do so, we can make a new tree by copying the newest version of the current persistent
tree the algorithm is operating on. The size of the tree is bounded by $O(h)$. We will use this tree to ``resume'' computing $W(e,g)$ in the subsequent algorithm.

\item
For each $g'\in output(e)\setminus\{g\}$ whose $W(e,g')$
has been computed, we store the tree for $W(e,g')$. We will use the tree to
compute the wavefronts $W(g')$ of $g'$ in the subsequent algorithm.

\item
We store the tree for the wavefront $W(e)$. Note that the tree may have many versions due to processing the events and we only keep its original version for $W(e)$. Hence, the size of the tree is $O(h)$. This tree will be used in the subsequent algorithm to compute $W(e,g')$ for those edges $g'\in output(e)\setminus\{g\}$ whose $W(e,g')$ have not been computed yet.

\item
We check every transparent edge $e'$ of $\calS'$ with $e'\neq e$. If $e'$ has been
processed (i.e., the wavefront propagation procedure has been called
on $W(e')$) and there is an edge $g'\in output(e')$ that has {\em not} been processed, we know that $W(e',g')$ has been computed and is available; we store the tree for $W(e',g')$. We will use the tree to compute the wavefronts $W(g')$ of $g'$ in the subsequent algorithm.
\end{enumerate}

We refer to the wavefronts stored in the algorithm as the {\em snapshot}; intuitively, the snapshot contains all wavelets in the forefront of the wavelet expansion.
By the same analysis as in~\cite{ref:WangSh21}, we can show that the snapshot contains sufficient information for the subsequent algorithm to proceed as usual and the total space of the snapshot is $O(h)$.

The above discusses our changes to the wavefront propagation procedure. For the wavefront merging procedure, which is to construct $W(e)$ from $W(f,e)$ for the edges $f\in input(e)$,
notice that we do not need the old versions of $W(f,e)$ anymore after $W(e)$ is constructed. Therefore, it is not necessary to use the path-copying method to process each event in the procedure.
Hence, the total space needed in the wavefront merging procedure in the entire algorithm is $O(n)$.

\subsubsection{Reducing the space of constructing $\spm'(s)$}

For the second issue of constructing $\spm'(s)$, our algorithm relies on the wavefronts $W(e)$ for all transparent edges $e$, which are maintained by persistent trees.
Due to the space-reset, our algorithm does not maintain the wavefronts anymore, and thus we need to somehow restore these wavefronts in order to construct $\spm'(s)$. To this end, a key observation is that by marking a total of $O(h)$ additional wavelet generators it is possible to restore all historical wavefronts that are needed for constructing $\spm'(s)$. In this way, $\spm'(s)$ can be constructed in $O(n + h\log h)$ time and $O(n)$ space.

More specifically, our algorithm considers each cell $c$ of $\calS'$ individually. For each cell $c$, the algorithm has two steps. First, compute the active regions of $c$. Second, for each active region $R$, compute the vertices of $\spm'(s)$ in $R$. For both steps, our algorithm utilizes the wavefronts $W(e)$ of the transparent edges $e$ on the boundary of $c$. Due to the space reset, the wavefronts $W(e)$ are not available anymore in our new algorithm. We use the following strategy to resolve the issue.

First, to compute the active regions in $c$, we need to know the bisectors defined by an unmarked generator $\alpha$ and a marked generator $\alpha'$ in the wavefronts $W(e)$ of the transparent edges $e$ of $c$. We observe that $\alpha$ is a generator adjacent to $\alpha'$ in $W(e)$. Based on this observation, we slightly modify our wavefront expansion algorithm so that it also marks the neighbors of the generators that are marked in our original algorithm and we call them {\em newly-marked} generators (the generators marked in our original algorithm are called {\em originally-marked} generators). If a generator is both newly-marked and originally-marked, we consider it as originally-marked. As each generator has two neighbors in a wavefront, the total number of marked generators is still $O(h)$. The newly-marked generators and the originally-marked generators are sufficient for computing all active regions of each cell $c$. Indeed, the active regions are decomposition of $c$ by the bisectors of adjacent generators with one originally-marked and the other newly-marked in $W(e)$ of the transparent edges $e$ of $c$.

Second, to compute the vertices of $\spm'(s)$ in each active region $R$ of $c$, we need to restore the wavefronts $W(e)$ of the transparent edges $e$ on the boundary of $R$. To this end, we observe that $W(e)$ consists of exactly the originally-marked generators that claim $e$. Consequently, the same algorithm as before can be applied to construct $\spm'(s)$. When computing the partition $Vor(e)$ of $R$, we propagate $W(e)$ using the wavefront propagation procedure, which uses persistent trees to represent $W(e)$. Since here we do not need to keep the old versions of $W(e)$ any more, we can use an ordinary tree (without the path-copying method) to represent $W(e)$. In this way, processing each bisector event only introduces $O(1)$ additional space. Hence, constructing $\spm'(s)$ takes $O(n+h\log h)$ time and $O(n)$ space.

\section{The general case}
\label{sec:general}

In this section, we extend our algorithm for the convex case in Section~\ref{sec:convex} to the general case where obstacles of $\calP$ may not be convex. To this end, we resort to an extended corridor structure of $\calP$, which has been used to solve various problems in polygonal domains~\cite{ref:ChenCo19,ref:ChenTw16,ref:ChenCo17,ref:ChenA15,ref:KapoorAn97,ref:MitchellSe95}. The structure decomposes the free space $\calF$ into three types of regions: an {\em ocean} $\calM$, $O(h)$ {\em canals}, and $O(n)$ {\em bays}. The details are given below.

\subsection{The extended corridor structure}

Let $\Tri(\calP)$ denote an arbitrary triangulation of $\calP$.
Let $G(\calP)$ be the (planar) dual graph of $\Tri(\calP)$,
i.e., each node of $G(\calF)$ corresponds to a triangle in
$\Tri(\calP)$ and each edge connects two nodes of $G(\calP)$
corresponding to two triangles sharing a triangulation diagonal of $\Tri(\calP)$.
We compute a {\em corridor graph} $G$, as follows. First, repeatedly remove every degree-one node from $G(\calP)$  until no such node remains. Second, repeatedly remove every degree-two node from
$G(\calP)$ and replace its two incident edges by a single edge
until no such node remains. The resulting graph is $G$ (e.g., see Fig.~\ref{fig:triangulation}), which has
$O(h)$ faces, nodes, and edges~\cite{ref:KapoorAn97}. Each node of
$G$ corresponds to a triangle in $\Tri(\calP)$, which is called a
{\em junction triangle} (e.g., see Fig.~\ref{fig:triangulation}).
The removal of all junction triangles results in $O(h)$
{\em corridors} (defined below), and each corridor
corresponds to an edge of $G$.

\begin{figure}[t]
\begin{minipage}[t]{0.47\linewidth}
\begin{center}
\includegraphics[totalheight=1.5in]{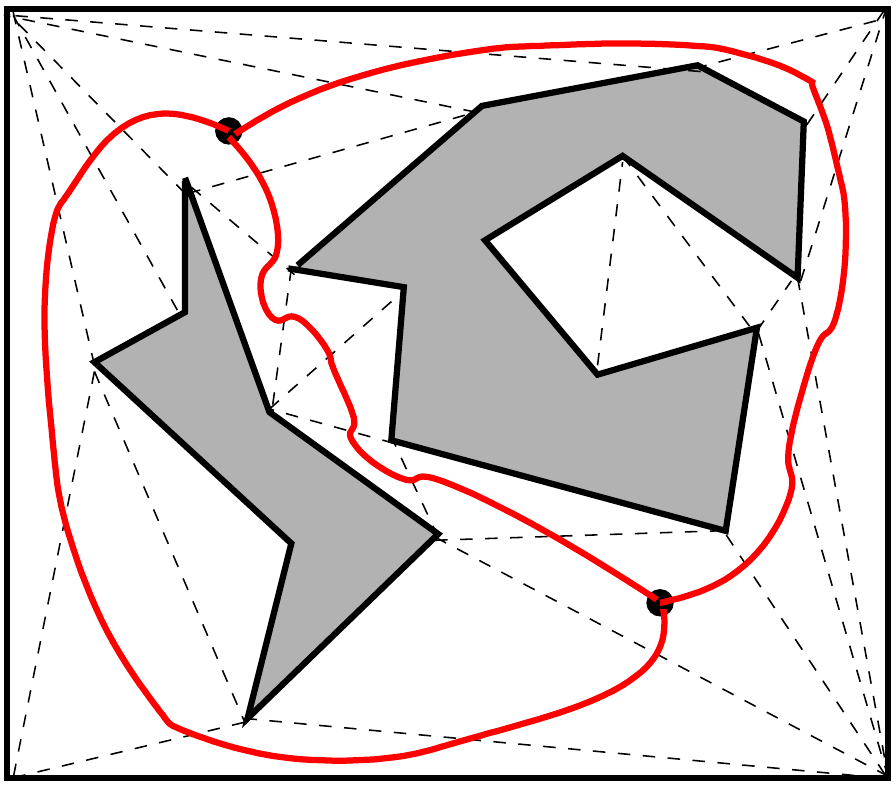}
\caption{\footnotesize
Illustrating a triangulation of $\calP$ with two obstacles.
There are two junction triangles indicated by the large dots inside
them, connected by three solid (red) curves. Removing the two
junction triangles results in three corridors.}
\label{fig:triangulation}
\end{center}
\end{minipage}
\hspace*{0.04in}
\begin{minipage}[t]{0.52\linewidth}
\begin{center}
\includegraphics[totalheight=1.6in]{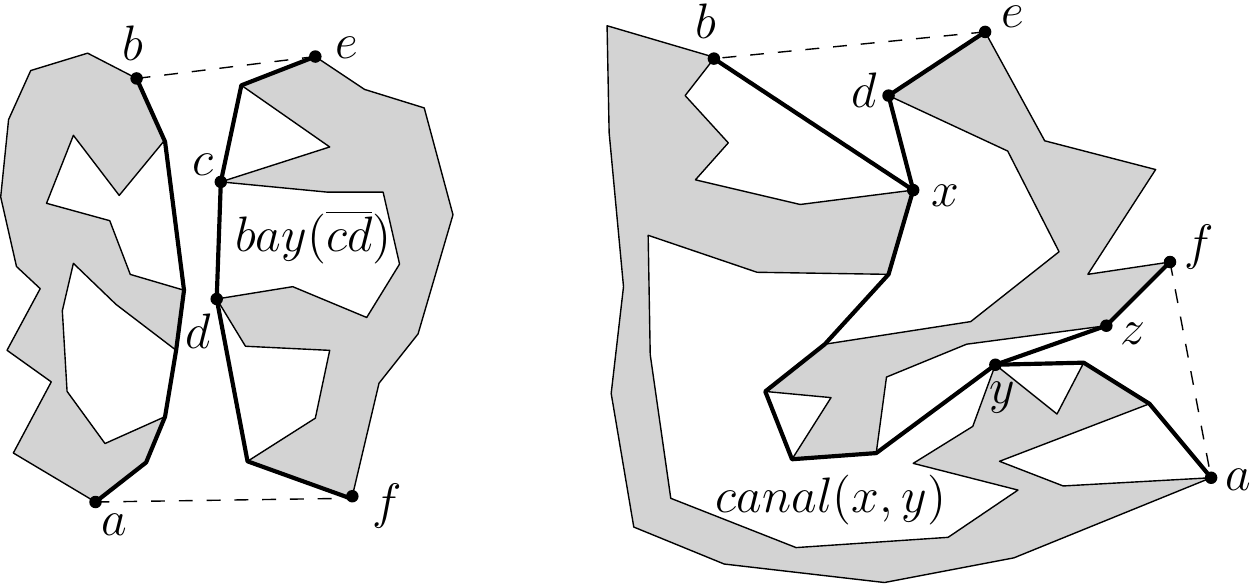}
\caption{\footnotesize
Illustrating an open hourglass (left) and a
closed one (right) with a corridor path connecting the apices
$x$ and $y$ of the two funnels. The dashed segments are diagonals.
The paths $\pi_C(a,b)$ and $\pi_C(e,f)$ are marked by thick solid
curves. A bay $bay(\overline{cd})$ with gate $\overline{cd}$ (left)
and a canal $\canal$ with gates $\overline{xd}$ and $\overline{yz}$
(right) are also shown.} \label{fig:corridor}
\end{center}
\end{minipage}
\vspace*{-0.15in}
\end{figure}

The boundary of a corridor $C$ consists of four parts (see
Fig.~\ref{fig:corridor}): (1) A boundary portion of an obstacle, from an obstacle vertex $a$ to an obstacle vertex $b$; (2) a triangulation
diagonal of a junction triangle from $b$ to an obstacle vertex $e$; (3) a boundary portion of an obstacle from $e$ to an obstacle vertex $f$; (4) a diagonal of a
junction triangle from $f$ to $a$. The two diagonals $\overline{be}$
and $\overline{af}$ are called the {\em doors} of $C$, and the other two parts of the boundary of $C$ are the two {\em sides} of $C$. Note that $C$ is a simple polygon. A point is in the {\em interior} of $C$ if it is in $C$ excluding the two doors.
Let $\pi_C(a,b)$ (resp., $\pi_C(e,f)$) be the shortest path from $a$ to $b$
(resp., $e$ to $f$) in $C$. The region $H_C$ bounded by
$\pi_C(a,b), \pi_C(e,f)$, $\overline{be}$, and
$\overline{fa}$ is called an {\em hourglass}, which is {\em open} if
$\pi_C(a,b)\cap \pi_C(e,f)=\emptyset$ and {\em closed} otherwise (see
Fig.~\ref{fig:corridor}). If $H_C$ is open, then both $\pi_C(a,b)$ and
$\pi_C(e,f)$ are convex chains and called the {\em sides} of
$H_C$; otherwise, $H_C$ consists of two ``funnels''
and a path $\pi_C=\pi_C(a,b)\cap \pi_C(e,f)$ joining the two apices of the
funnels, called the {\em corridor path} of $C$.
Each side of a funnel is also convex.

Let $\calM$ be the union of all $O(h)$ junction triangles, open hourglasses, and funnels.
We call $\calM$ the {\em ocean}, whose boundary
$\partial\calM$ consists of $O(h)$ convex chains that are sides of open hourglasses and funnels.
The other space of $\calP$, i.e., $\calP\setminus\calM$, is further
partitioned into two types of regions: {\em bays} and {\em canals}, defined as follows.
Consider the hourglass $H_C$ of a corridor $C$.

If $H_C$ is open (see Fig.~\ref{fig:corridor}), then $H_C$
has two sides. Let $S_1$ be a side of $H_C$.
The obstacle vertices on $S_1$ all lie on the same
side of $C$. Let $c$ and $d$ be any two consecutive
vertices on $S_1$ such that $\overline{cd}$ is
not an obstacle edge of $\calP$ (e.g., see Fig.~\ref{fig:corridor} left). The region enclosed
by $\overline{cd}$ and the boundary portion of $C$ between $c$ and
$d$ is called a {\em bay}, denoted by
$bay(\overline{cd})$. We call
$\overline{cd}$ the {\em gate} of $bay(\overline{cd})$.

If $H_C$ is closed, let $x$ and $y$ be the two apices
of the two funnels. Consider two consecutive vertices $c$ and $d$ on
a side of a funnel such that $\overline{cd}$ is not an obstacle edge of $\calP$.
If $c$ and $d$ are on the same side of the corridor $C$, then $\overline{cd}$ also
defines a bay. Otherwise, either $c$ or $d$ is a funnel
apex, say, $c=x$, and we call $\overline{xd}$ a {\em canal gate} at $x=c$
(e.g., see Fig.~\ref{fig:corridor} right).
Similarly, there is also a canal gate
at the other funnel apex $y$, say $\overline{yz}$.
The region of $C$ between the two canal gates
$\overline{xd}$ and $\overline{yz}$ 
is the {\it canal} of $H_C$, denoted by $canal(x,y)$.

Each bay or canal is a simple polygon.
All bays and canals together constitute the space $\calP\setminus\calM$.
Each vertex of $\partial\calM$ is a vertex of $\calP$ and
each edge of $\partial\calM$ is either an edge of $\calP$ or a gate
of a bay/canal. Gates are common boundaries between $\calM$ and bays/canals.
After $\calP$ is triangulated, $\calM$ and all bays and canals can be obtained in $O(n)$ time~\cite{ref:KapoorAn97}.

The reason that the extended corridor structure can help find a shortest path is the following. Suppose we want to find a shortest \st\ path for two points $s$ and $t$. We consider $s$ and $t$ as two special obstacles and build the extended corridor structure. If a shortest \st\ path $\pi(s,t)$ contains a point in the interior of a corridor $C$, then $\pi(s,t)$ must cross both doors of $C$ and stay in the hourglasses of $C$, and further, if the hourglass is closed, then its corridor path must be contained in $\pi(s,t)$. In fact, $\pi(s,t)$ must be in the union of the ocean $\calM$ and all corridor paths~\cite{ref:KapoorAn97}.

In light of the above properties, we propose the following algorithm.
Let $s$ be a given source point. By considering $s$ as a special obstacle of $\calP$, we construct the extended corridor structure of $\calP$. Consider any query point $t$, which may be in the ocean $\calM$, a bay $\bay$, or a canal $\canal$.

\begin{itemize}
\item
If $t\in \calM$, then the union of $\calM$ and all corridor paths contains a shortest \st\ path. To handle this case, we will build a shortest path map $\spm(\calM)$ in $\calM$ with respect to the union of $\calM$ and all corridor paths. In face, $\spm(\calM)$ is exactly the portion of $\spm(s)$ in $\calM$, i.e., $\spm(s)\cap \calM$. To build $\spm(\calM)$, a key observation is that the boundary $\partial \calM$ consists of $O(h)$ convex chains. Therefore, we can easily adapt our previous algorithm for the convex case. However, the algorithm needs to be modified so that the corridor paths should be taken into consideration. Intuitively, corridor paths provide certain kind of ``shortcuts'' for wavefronts to propagate.

\item
If $t$ is in a bay $\bay$, then any shortest \st\ path must cross its gate $\overline{cd}$. To handle this case, we will extend $\spm(\calM)$ into $\bay$ through the gate $\overline{cd}$ to construct the shortest path map in $\bay$, i.e., the portion of $\spm(s)$ in $\bay$, $\spm(s)\cap \bay$.

\item
If $t$ is in a canal $\canal$, then any shortest \st\ path must cross one of the two gates of the canal. To handle this case, we will extend $\spm(\calM)$ into $\canal$ through the two gates to construct the shortest path map in $\canal$, i.e., the portion of $\spm(s)$ in $\canal$, $\spm(s)\cap \canal$.
\end{itemize}

In the following, we first describe our algorithm for constructing $\spm(\calM)$ in Section~\ref{sec:ocean}. We then expand $\spm(\calM)$ into all bays in Section~\ref{sec:bay} and expand $\spm(\calM)$ into all canals in Section~\ref{sec:canal}. The algorithm for the canal case utilizes the bay case algorithm as a subroutine.

\subsection{Constructing the shortest path map in the ocean $\spm(\calM)$}
\label{sec:ocean}

As the boundary of $\calM$ consists of $O(h)$ convex chains, we can
apply and slightly modify our algorithm for the convex case. To do so, for each
convex chain of $\partial\calM$, we define its rectilinear extreme
vertices in the same way as before. Let $\calV$ be the set of the rectilinear extreme vertices of
all convex chains. Hence, $|\calV|=O(h)$. In addition, to incorporate the
corridor paths into the algorithm, we include the endpoints of each corridor path in $\calV$. As there are $O(h)$ corridor paths, the size of $\calV$ is still bounded by $O(h)$. Note that each corridor path endpoint is also an endpoint of a convex chain of $\partial \calM$. We construct the conforming
subdivision $\calS$ based on the points of $\calV$
and then insert the convex chains of $\calM$ into
$\calS$ to obtain $\calS'$. The algorithm is essentially the same as
before. In addition, we make the following changes to $\calS'$, which is mainly for incorporating the corridor paths into our wavefront expansion algorithm, as will be clear later.

Let $v$ be an endpoint of a corridor path $\pi$. Since $v$ is in $\calV$, $v$ is incident to $O(1)$ transparent edges in $\calS'$. For each such transparent edge $e$, if $|\pi|<2 \cdot |e|$, then we divide $e$ into two sub-edges such that the length of the one incident to $v$ is equal to $|\pi|/2$; for each sub-edge, we set its well-covering region the same as $\calU(e)$. Note that this does not affect the properties of $\calS'$. In particular, each transparent edge $e$ is still well-covered. This change guarantees the following property: for each corridor path $\pi$, $|\pi|\geq 2\cdot |e'|$ holds, where $e'$ is any transparent edge of $\calS'$ incident to either endpoint of $\pi$. For reference purpose, we refer to it as the {\em corridor path length property}.

Next we apply the wavefront expansion algorithm. Here we need to
incorporate the corridor paths into the algorithm. Intuitively, each
corridor path provides a ``shortcut'' for the wavefront, i.e., if a
wavelet hits an endpoint of a corridor path, then the wavelet will
come out of the corridor path from its other endpoint but with a delay
of distance equal to the length of the corridor path. More details are
given below.

Since all corridor path endpoints are in $\calV$, they are vertices of transparent edges of $\calS'$.
Consider an endpoint $v$ of a corridor path $\pi$. Let $u$ be the other endpoint of $\pi$.
Recall that the wavefront propagation procedure for $W(e)$ is to propagate $W(e)$ to compute $W(e,g)$ for all edges $g\in output(e)$. In addition to the previous algorithm for the procedure, we also propagate $W(e)$ through the corridor path $\pi$ to $u$. This is done as follows. Recall that when $e$ is processed, since $v$ is an endpoint of $e$, the weighted distance of $v$ through $W(e)$ is equal to $d(s,v)$. Hence, the wavefront $W(e)$ hits $u$ through $\pi$ at time $d(s,v)+|\pi|$. We then update $covertime(g)= \min\{covertime(g),d(s,v)+|\pi|+|g|\}$, for each transparent edge $g$ incident to $u$. We also set the wavefront $W(e,g)$ consisting of the only wavelet with $u$ as the generator with weight equal to $d(s,v)+|\pi|$. Since there are $O(1)$ transparent edges $g$ incident to $u$, the above additional step takes $O(1)$ time, which does not change the time complexity of the overall algorithm asymptotically. The corridor path length property assures that if $W(e)$ contributes to a wavefront $W(g)$ at $g$, then $e$ must be processed earlier than $g$. This guarantees the correctness of the algorithm.

In this way, we can first construct a decomposition $\spm'(\calM)$ of $\calM$ in $O(n+h\log h)$ time and $O(n)$ space, where $\spm'(\calM)$ is defined similarly as $\spm'(s)$ in Section~\ref{sec:convex}. Then, by a similar algorithm as that for Lemma~\ref{lem:spmconstruction}, $\spm(\calM)$ can be obtained in additional $O(n)$ time.

\subsection{Expanding $\spm(\calM)$ into all bays}
\label{sec:bay}
We now expand $\spm(\calM)$ into all bays in $O(n+h\log h)$ time and $O(n)$ space. In fact, we expand $\spm'(\calM)$ to the bays.
We process each bay individually. Consider a bay $\bay$ with gate $\overline{cd}$.
Without loss of generality, we assume that $\overline{cd}$ is
horizontal, $c$ is to the left of $d$, and $\bay$ is locally above
$\overline{cd}$.

Let $v_1,v_2,\ldots,v_{m}$ be the vertices of $\spm'(\calM)$ on
$\overline{cd}$ ordered from left to right (e.g., see Fig.~\ref{fig:bayextend}). Let $c=v_0$ and
$d=v_{m+1}$. Hence, each $\overline{v_iv_{i+1}}$ is claimed by a
generator $\alpha_i=(A_i,a_i)$ for all $i=0,1,\ldots,m$.
Let $b_i$ and $c_i$ be the tangent points on $A_{i-1}$ and $A_i$ from
$v_i$, respectively, for each $i=1,2,\ldots,m$  (e.g., see Fig.~\ref{fig:bayextend}). For $v_0$, only $c_0$
is defined; for $v_{m+1}$, only $b_{m+1}$ is defined.
Observe that for any point $p\in \overline{v_iv_{i+1}}$, which is
claimed by $\alpha_i$, its tangent point on $A_i$ must be on the
portion of $A_i$ between $c_i$ and $b_{i+1}$ and we use $A_i'$ to
denote that portion. So with respect to $\bay$, we use
$\alpha_i'=(A_i',a_i')$ to
refer to the generator, where $a_i'$ refers to the one of $c_i$ and
$b_{i+1}$ that is closer to $a_i$. Hence, for any point $t\in \bay$,
any shortest path $\pi(s,t)$ from $s$ to $t$ must be via one of the
generators $\alpha'_i$ for $i=0,1,\ldots,m$.
Consider the region $R$ bounded by $A'_i$ for all $i\in [0,m]$, the
tangents from $v_i$ to their generators for all $i\in [0,m+1]$, and
the boundary of the bay excluding its gate. Notice that $R$ is a
simple polygon. For any point $t\in \bay$, the above observation implies any shortest \st\ path $\pi(s,t)$ is the
concatenation of a shortest path from $s$ to a generator initial vertex $a_i'$ and
the shortest path from $a_i'$ to $t$ in $R$.

\begin{figure}[t]
\begin{minipage}[t]{\textwidth}
\begin{center}
\includegraphics[height=2.7in]{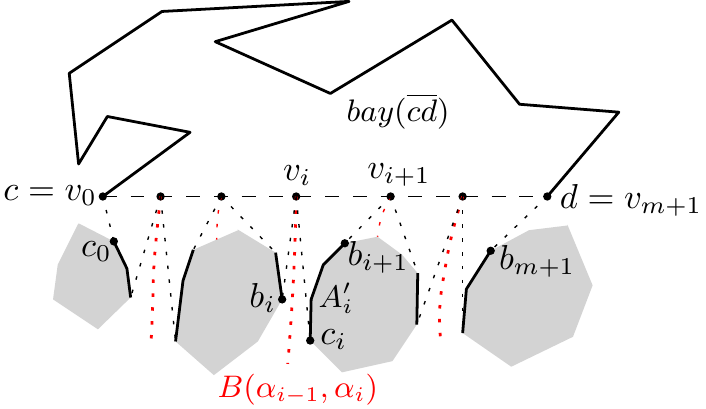}
\caption{\footnotesize Illustrating $\bay$ and the generators. The thick segments on obstacles are $A_i'$, $i=0,1,\ldots,m$.}
\label{fig:bayextend}
\end{center}
\end{minipage}
\vspace{-0.15in}
\end{figure}

According to the above discussion, expanding $\spm'(\calM)$ into $\bay$
is equivalent to the following weighted geodesic Voronoi diagram
problem in a simple polygon: Partition $R$ into cells with respect to
the point sites $a_0',a_1',\ldots,a_m'$ (with weights equal to their
geodesic distances to $s$) such that all points in the same cell have the same closest site.
Let $n_b$ be the number of vertices in $\bay$ (the subscript ``b'' represents ``bay''). Let $n_g$ be the total number of obstacle vertices in $A_i'$ for all $i\in [0,m]$ (the subscript ``g'' represents ``generator'').
Note that $v_i$ for all $i=1,\ldots,m$ are also vertices of $R$.
Hence, the number of vertices of $R$ is $n_b+n_g+m$.
The above problem can be solved in $O(m\log m+n_b+n_g)$ time by the
techniques of Oh~\cite{ref:OhOp19}. Indeed, given $m'$ point sites in a simple polygon $P'$ of $n'$ vertices, Oh~\cite{ref:OhOp19} gave an algorithm that can compute the geodesic Voronoi diagram of the sites in $P'$ in $O(n'+m'\log m')$ time and $O(n'+m')$ space. Although the point sites in Oh's problem do not have weights, our problem is essentially an intermediate step of Oh's algorithm because all weighted point sites in our problem are on one side of $\overline{cd}$. Therefore, we can run Oh's algorithm from ``the middle'' and solve our problem in $O(m\log m+n_b+n_g)$ and $O(n_b+n_g)$ space. In fact, our problem is a special case of Oh's problem because there are no sites in $\bay$. For this reason, we propose our own algorithm to solve this special case and the algorithm is much simpler than Oh's algorithm; our algorithm also runs in $O(m\log m+n_b+n_g)$ and $O(n_b+n_g+m)$ space. This also makes our paper more self-contained.

Before presenting the algorithm, we analyze the total time for processing all bays.
Since $\spm'(\calM)$ has $O(h)$ vertices, the total sum of $m$ for all bays is $O(h)$. The total sum of $n_b$ for all bays is at most $n$. Notice that the obstacle edges on $A_i'$ are disjoint for different bays, and thus the total sum of $n_g$ for all bays is $O(n)$. Hence, expanding $\spm'(\calM)$ to all bays takes $O(n+h\log h)$ time and $O(n)$ space in total.

The above actually only considers the case where the gate $\overline{cd}$ contains at least one vertex of $\spm'(\calM)$. It is possible that no vertex of $\spm'(\calM)$ is on $\overline{cd}$, in which case the entire gate is claimed by one generator $\alpha$ of $\spm'(\calM)$. We can still define the region $R$ in the same way. But now $R$ has only one weighted site and thus the geodesic Voronoi diagram problem becomes computing a shortest path map in the simple polygon $R$ for a single source point. This problem can be  solved in $O(n_b+n_g)$ time~\cite{ref:GuibasLi87}; note that $m=0$ in this case. Hence, the total time for processing all bays in this special case is $O(n)$.


\subsubsection{Solving the special weighted geodesic Voronoi diagram problem}

We present an algorithm for the above special case of the weighted geodesic Voronoi diagram problem and the algorithm runs in $O(m\log m+n_b+n_g)$ and $O(n_b+n_g+m)$ space.

If we consider all generators $\alpha_i$ for $i=0,1,\ldots,m$ as a wavefront at $\overline{cd}$, denoted by $W(\overline{cd})$, then our algorithm is essentially to propagate $W(\overline{cd})$ inside $\bay$. To this end, we first triangulate $\bay$ and will use the triangulation to guide the wavefront propagation. Each time we propagate the wavefront through a triangle. In the following, we first discuss how to propagate $W(\overline{cd})$ through the triangle $\triangle cda$ with $\overline{cd}$ as an edge and $a$ as the third vertex. This is a somewhat special case as $\triangle cda$ is the first triangle the wavefront will propagate through; later we will discuss the general case but the algorithm is only slightly different.

Recall that each convex chain $A_i'$ is represented by an array and the generator list $W(\overline{cd})$ is represented by a balanced binary search tree $T(W(\overline{cd}))$. We build a point location data structure on the triangulation of $\bay$ in $O(n_b)$ time~\cite{ref:EdelsbrunnerOp86,ref:KirkpatrickOp83}, so that given any query point $p$, we can determine the triangle that contains $p$ in $O(\log n_b)$ time.

We begin with computing the intersection of the adjacent bisectors $B(\alpha_{i-1},\alpha_i)$ and $B(\alpha_i,\alpha_{i+1})$ for all $i=1,2,\ldots, m-1$. Each intersection can be computed in $O(\log n_g)$ time by the bisector-bisector intersection operation in Lemma~\ref{lem:bb-intersection}. Computing all intersections takes $O(m\log n_g)$ time. For each intersection $q$, called a {\em bisector event}, we use the point location data structure to find the triangle of the triangulation that contains $q$ and store $q$ in the triangle.

Since all generators are outside $\bay$, by Corollary~\ref{coro:monotone}, all bisectors are monotone with respect to the direction orthogonal to $\overline{cd}$. Our algorithm for propagating the wavefront through $\triangle cda$ is based on this property.
We sort all bisector events in $\triangle cda$ according to their perpendicular distances to the supporting line of $\overline{cd}$. Then, we process these events in the same way as in our wavefront propagation procedure. Specifically, for each bisector event $q$ of two bisectors $B(\alpha',\alpha)$ and $B(\alpha,\alpha'')$, we remove $\alpha$ from the generator list. Then, we compute the intersection $q'$ of $B(\alpha',\alpha'')$ with $B(\alpha_1',\alpha')$, where $\alpha_1'$ is the other neighboring bisector of $\alpha'$ than $\alpha$, and we use the point location data structure to find the triangle that contains $q'$ and store $q'$ in the triangle. If $q'\in \triangle cda$, then we insert it to the bisector event sorted list of $\triangle cda$. We do the same for $B(\alpha',\alpha'')$ and $B(\alpha'',\alpha''_1)$, where $\alpha''_1$ is the other neighbor of $\alpha''$ than $\alpha$. After all events in $\triangle cda$ are processed, we split the current wavefront $W$ at the vertex $a$. To this end, we first find the generator $\alpha^*$ of $W$ that claims $a$. For this, we use the same algorithm as before, i.e., binary search plus bisector tracing. So we need to maintain a tracing-point for each bisector as before (initially, we can set $v_i$ as the tracing-point for $B(\alpha_{i-1},\alpha_i)$, $i=1,2,\ldots,m$).
The correctness of the above algorithm for finding the generator $\alpha^*$ relies on the property of the following lemma.

\begin{lemma}\label{lem:intersectionbay}
The bisector $B(\alpha,\alpha')$ intersects $\overline{ad}\cup \overline{ac}$ at most once for any two bisectors of $\alpha$ and $\alpha'$ of $W(\overline{cd})$.
\end{lemma}
\begin{proof}
Note that we cannot apply the result of Lemma~\ref{lem:intersection} since to do so we need to have $\overline{ad}$ and $\overline{ac}$ parallel to $\overline{cd}$. But the proof is somewhat similar to that for Lemma~\ref{lem:intersection}.

\begin{figure}[t]
\begin{minipage}[t]{\textwidth}
\begin{center}
\includegraphics[height=2.0in]{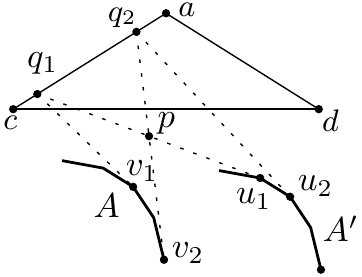}
\caption{\footnotesize Illustrating the proof of Lemma~\ref{lem:intersectionbay}.}
\label{fig:intersectionbay}
\end{center}
\end{minipage}
\vspace{-0.15in}
\end{figure}

Assume to the contrary that $B(\alpha,\alpha')$ intersects $\overline{ad}\cup \overline{ac}$ at two points, $q_1$ and $q_2$. Let $A$ and $A'$ be the underlying arcs of $\alpha$ and $\alpha'$, respectively. Let $v_1$ and $u_1$ be the tangents points of $q_1$ on $A$ and $A'$, respectively (e.g., see Fig.~\ref{fig:intersectionbay}).
Let $v_2$ and $u_2$ be the tangents points of $q_2$ on $A$ and $A'$, respectively.
Since both $A$ and $A'$ are on one side of $\overline{cd}$ while $\overline{ad}\cup \overline{ac}$ is on the other side, if we move a point $q$ from $q_1$ to $q_2$ on $\overline{ad}\cup \overline{ac}$, the tangent from $q$ to $A$ will continuously change from $\overline{q_1v_1}$ to $\overline{q_2v_2}$ and the tangent from $q$ to $A'$ will continuously change from $\overline{q_1u_1}$ to $\overline{q_2u_2}$. Therefore, either $\overline{q_1u_1}$ intersects $\overline{q_2v_2}$ in their interiors or $\overline{q_1v_1}$ intersects $\overline{q_2u_2}$ in their interiors; without loss of generality, we assume that it is the former case. Let $p$ be the intersection of $\overline{q_1u_1}$ and $\overline{q_2v_2}$ (e.g., see Fig.~\ref{fig:intersectionbay}). Since $q_1\in B(\alpha,\alpha')$, points of $\overline{q_1u_1}$ other than $q_1$ have only one predecessor, which is $\alpha'$. As $p\in \overline{q_1u_1}$ and $p\neq q_1$, $p$ has only one predecessor $\alpha'$. Similarly, since $q_2\in B(\alpha,\alpha')$ and $p\in \overline{q_2v_2}$, $\alpha$ is also $p$'s predecessor. We thus obtain a contradiction.
\end{proof}


After $\alpha^*$ is found, depending on whether $\alpha^*$ is the first or last generator of $W$, there are three cases.

\begin{enumerate}
\item
If $\alpha^*$ is not the first or last generator of $W$, then we split $W$ into two wavefronts, one for $\overline{ca}$ and the other for $\overline{ad}$. To do so, we first split the binary tree $T(W)$ that represents the current wavefront $W$ at $\alpha^*$. Then,
we do binary search on $A^*$ to find the tangent point from $a$, where $A^*$ is the underlying chain of $\alpha^*$. We also split $\alpha^*$ into two at the tangent point of $A^*$, i.e., split $A^*$ into two chains that form two generators, one for $\overline{ac}$ and the other for $\overline{ad}$ (e.g., see Fig.~\ref{fig:generatorsplit}). As $A^*$ is represented by an array, splitting $\alpha^*$ can be performed in $O(1)$ time by resetting the end indices of the chains in the array.
This finishes the propagation algorithm in $\triangle acd$. The above splits $W$ into two wavefronts, one for $\overline{ac}$ and the other for $\overline{ad}$; we then propagate the wavefronts through $\overline{ac}$ and $\overline{ad}$ recursively.

\begin{figure}[t]
\begin{minipage}[t]{\textwidth}
\begin{center}
\includegraphics[height=1.6in]{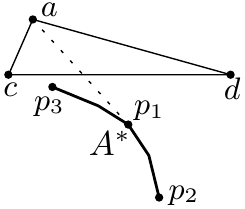}
\caption{\footnotesize Splitting the generator $\alpha^*$. Assume that $A^*$ is the convex chain from $p_2$ to $p_3$ with $p_2$ as the initial vertex of the generator. $\overline{ap_1}$ is tangent to $A^*$ at $p_1$. After the split, the chain from $p_3$ to $p_1$ becomes a generator with $p_1$ as the initial vertex and the chain from $p_1$ to $p_2$ becomes another generator with $p_2$ as the initial vertex. }
\label{fig:generatorsplit}
\end{center}
\end{minipage}
\vspace{-0.15in}
\end{figure}

\item
If $\alpha^*$ is the first generator of $W$, then $\alpha^*$ must be $\alpha_0$, i.e., the leftmost generator of $W(\overline{cd})$. In this case, we do not need to split $W$. But we still need to split the generator $\alpha_0$ at the tangent point $p_0$ of $\alpha_0$ from $a$. To find the tangent point $p_0$, however, this time we do not use binary search as it is possible that we will need to do this for $\Omega(n_b)$ vertices of $\bay$, which would take $\Omega(n_b\log n_g)$ time in total. Instead, we use the following approach. Recall that the vertex $c=v_0$ connects to $\alpha_0$ by a tangent with tangent point $c_0$ (e.g., see Fig.~\ref{fig:bayextend}), and $c_0$ is an endpoint of $A_0'$. We traverse $A'_0$ from $c_0$ to $b_0$, i.e., the other endpoint of $A'_0$; for each vertex, we check whether it is the tangent from $a$. In this way, $p_0$ can be found in time linear in the number of vertices of $A_0'$ between $c_0$ and the tangent point $p_0$ (e.g., see Fig.~\ref{fig:simplecase}). After that, we still split $\alpha_0$ into two generators at $p_0$; one is the only generator for $\overline{ac}$ and the other becomes the first generator of the wavefront for $\overline{ad}$.

\begin{figure}[t]
\begin{minipage}[t]{\textwidth}
\begin{center}
\includegraphics[height=1.9in]{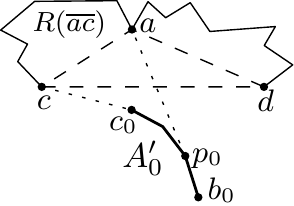}
\caption{\footnotesize Illustrating the case for propagating the one-generator wavefront to $R(\overline{ac})$. }
\label{fig:simplecase}
\end{center}
\end{minipage}
\vspace{-0.15in}
\end{figure}

For $\overline{ad}$, we propagate its wavefront through $\overline{ad}$ recursively. For $\overline{ac}$, to propagate its wavefront through $\overline{ac}$, since the wavefront has only one generator, we can simply apply the linear time shortest path map algorithm for simple polygons~\cite{ref:GuibasLi87}; we refer to this as the {\em one-generator case}. Indeed, $\overline{ac}$ partitions $\bay$ into two sub-polygons and let $R(\overline{ac})$ denote the one that does not contain $\triangle cda$ (e.g., see Fig.~\ref{fig:simplecase}). Hence, all points of $R(\overline{ac})$ are claimed by the only generator for $\overline{ac}$, whose underlying chain $A'$ is the sub-chain of $A_0'$ from $c_0$ to $p_0$ and whose initial vertex is $p_0$.
Consider the region $R'(\overline{ac})$ bounded by $\overline{cc_0}$, $A'$, $\overline{p_0a}$, and the boundary of $R(\overline{ac})$ excluding $\overline{ac}$. It is a simple polygon with a single weighted source $p_0$. Therefore, our problem is equivalent to computing the shortest path map in $R'(\overline{ac})$ with respect to the source point $p_0$, which can be done in $O(|R(\overline{ac})|+|A'|)$ time~\cite{ref:GuibasLi87}, where $|R(\overline{ac})|$ and $|A'|$ are the numbers of vertices of $R(\overline{ac})$ and $|A'|$, respectively.

\item
If $\alpha^*$ is the last generator of $W$, the algorithm is symmetric to the above second case.
\end{enumerate}

The above describes our algorithm for propagating the wavefront $W(\overline{cd})$ through the first triangle $\overline{cda}$. Next, we discuss the general case where we propagate a wavefront $W$ of more than one generator through an arbitrary triangle of the triangulation of $\bay$. For the sake of notational convenience, we consider the problem of propagating the wavefront $W(\overline{ad})$ at $\overline{ad}$ through $\overline{ad}$ into the region $R(\overline{ad})$, where $R(\overline{ad})$ is the one of the two sub-polygons of $\bay$ partitioned by $\overline{ad}$ that does not contain $\triangle cda$. Let $\triangle adb$ be the triangle in $R(\overline{ad})$ with $\overline{ad}$ with as an edge, i.e., $b$ is the third vertex of the triangle. We describe the algorithm for propagating $W(\overline{ad})$ through $\triangle adb$. The algorithm is actually quite similar as before, with an additional {\em event-validation} step.

We first sort all bisector events in $\triangle adb$, following their perpendicular distances to the supporting line of $\overline{ad}$. Then we process these events as before. One difference is that we now need to check whether each event is still valid. Specifically, for each event $q$, which is associated with three generators $\alpha$, $\alpha'$, and $\alpha''$, i.e., $q$ is the intersection of the bisectors $B(\alpha,\alpha')$ and $B(\alpha',\alpha'')$, we check whether all three generators are still in the current wavefront $W$ \footnote{Note that the order of the generators in $W$ must be consistent with their initial index order, $\alpha_0,\alpha_1,\ldots,\alpha_m$. Hence, checking whether a generator is in $W$ can be easily done in $O(\log m)$ time by using the generator indices.}. If not (this is possible if one of the three generators was deleted before, e.g., when the wavefront was propagated through $\triangle cda$), then $q$ is not valid and we ignore this event; otherwise $q$ is still valid and we process it in the same way as before.

The above algorithm is based on the assumption that $\overline{ad}$ is a triangulation diagonal. If it is an obstacle edge, then the wavefront $W(\overline{ad})$ stops at $\overline{ad}$. Notice that each bisector of the wavefront $W(\overline{ad})$ must intersect $\overline{ad}$. For each bisector of the wavefront, starting from its current tracing-point, we trace it out until the current traced hyperbolic-arc intersects $\overline{ad}$.

The algorithm stops once all triangles are processed as above.

\paragraph{Time analysis.}
We now analyze the time complexity. As discussed before, the initial step for computing the intersections of adjacent bisectors of the wavefront $W(\overline{cd})$ and locating the triangles containing them together takes $O(m\log (n_g+n_b))$ time.
During the entire algorithm, each traced bisector hyperbolic-arc belongs to the shortest path map in $\bay$, i.e., the portion of $\spm(s)$ in $\bay$, whose size is $O(n_b+n_g+m)$.
Hence, the total time on the bisector tracing in the entire algorithm is $O(n_b+n_g+m)$.
For the one-generator case where only one generator is available for $\overline{ac}$, the time for processing the sub-polygon $R(\overline{ac})$ is $O(|R(\overline{ac})|+|A'|)$. Notice that all such sub-polygons $R(\overline{ac})$ in the one-generator case are interior disjoint. Hence, the total sum of their sizes is $O(n_b)$. Also, all such generator underlying chains $A'$ are also interior disjoint, and thus the total sum of their sizes is $O(n_g)$. Therefore, the overall time for processing the one-generator case sub-polygons is $O(n_g+n_b)$.

For the general case of processing a triangle $\triangle$ of the triangulation, the total time for processing all events is $O(m''\log (n_g+n_b))$\footnote{More specifically, sorting all events takes $O(m''\log m'')$ time. For each event, removing a generator from $W$ takes $O(\log m)$, finding the intersections of new adjacent bisectors takes $O(\log n_g)$ time, and then locating the triangles containing the intersections takes $O(\log n_b)$ time. Note that $m''\leq m\leq n_g$.}, where $m''$ is the number of events in $\triangle$, both valid and invalid.
Each valid or invalid event is computed either in the initial step or after a generator is deleted. The total number of bisector events in the former case in the entire algorithm is at most $m-1$. The total number in the latter case in the entire algorithm is no more than the number of generators that are deleted in the algorithm, which is at most $m$ because once a generator is deleted it will never appear in any wavefront again. Hence, the total time for processing events in the entire algorithm is $O(m\log (n_g+n_b))$.
Once all events in $\triangle$ are processed, we need to find the generator $\alpha^*$ of the current wavefront $W$ that claims the third vertex $b$ of the triangle, by binary search plus bisector tracing.
The time is $O(\log m)$ plus the time for tracing the bisector hyperbolic-arcs. Hence, excluding the time for tracing the bisector hyperbolic-arcs, which has been counted above, the total time for this operation in the entire algorithm is $O(m'\log m)$, where $m'$ is the number of triangles the algorithm processed for the case where $\alpha^*$ is not the first or last generator. We will show later in Lemma~\ref{lem:mprime} that $m'=O(m)$.

After $\alpha^*$ is found, depending on whether $\alpha^*$ is the first or last generator of $W$, there are three cases. If $\alpha^*$ is not the first or last generator of $W$, then we find the tangent from $b$ to $\alpha^*$ by binary search in $O(\log n_g)$ time and split both $W$ and $\alpha$; otherwise, we find the tangent by a linear scan on $\alpha^*$ and only need to split $\alpha^*$. Splitting $W$ takes $O(\log m)$ time while splitting a generator only takes $O(1)$ time as discussed before. Therefore, the total time for splitting generators in the entire algorithm is $O(n_b)$, as there are $O(n_b)$ triangles in the triangulation.
If $\alpha^*$ is either the first or last generator of $W$, then a one-generator case subproblem will be produced and the time of the linear scan for finding the tangent is $O(|A'|)$, where $A'$ is the sub-chain of $\alpha^*$ that belongs to the one-generator case subproblem. As discussed above, all such $A'$ in all one-generator case subproblems are interior disjoint, and thus the total time on the linear scan in the entire algorithm is $O(n_g)$. Therefore, the total time for finding the tangent point and splitting $W$ is $O(m'\log n_g+n_b)$ as $m\leq n_g$.

\begin{lemma}\label{lem:mprime}
$m'\leq m-1$.
\end{lemma}
\begin{proof}
Initially $W=W(\overline{cd})$, which consists of $m$ generators. Each
split operation splits a
wavefront $W$ at a generator $\alpha^*$ into two wavefronts each of which has at least
two generators (such that both wavefronts contain $\alpha^*$). More
specifically, if a wavefront of size $k$ is split, then one wavefront
has $k'$ generators and the other has $k-k'+1$ generators with $k'\geq 2$ and $k-k'+1\geq 2$. The value $m'$ is equal to the number of all
split operations in the algorithm.

We use a tree structure $T$ to
characterize the split operations. The root of $T$ corresponds to
the initial generator sequence $W(\overline{cd})$. Each split on a
wavefront $W$ corresponds to an internal node of $T$ with two children
corresponding to the two subsequences of $W$ after the split. Hence, $m'$ is equal to the number of internal nodes of $T$. In the worst case $T$ has $m$ leaves, each
corresponding to two generators $\alpha_i$ and $\alpha_{i+1}$ for
$i=0,1,\ldots,m$. Since each internal node of $T$ has two children and
$T$ has at most $m$ leaves, the number of internal nodes of $T$ is at
most $m-1$. Therefore, $m'\leq m-1$.
\end{proof}

With the preceding lemma, the total time of the algorithm is
$O(m\log (n_g+n_b)+n_g+n_b)$, which is $O(m\log m+n_g+n_b)$ by similar analysis as Observation~\ref{obser:hlogh}. The space is
$O(n_g+n_b+m)$ as no persistent data structures are used.

\subsection{Expanding $\spm'(\calM)$ into all canals}
\label{sec:canal}

Consider a canal $\canal$ with two gates $\overline{xd}$ and $\overline{yz}$. The goal is to expand the map $\spm'(\calM)$ into $\canal$ through the two gates to obtain the shortest path map in the canal, denoted by $\spm(\canal)$.

The high-level scheme of the algorithm is similar in spirit to that for the $L_1$ problem~\cite{ref:ChenCo19}.
The algorithm has three main steps. First, we
expand $\spm'(\calM)$ into $\canal$ through the gate $\overline{xd}$, by applying our algorithm for bays. Let $\spm_1(\canal)$
denote the map of $\canal$ obtained by the algorithm. Second,
we expand $\spm(\calM)$ into $\canal$ through the gate
$\overline{yz}$ by a similar algorithm as above; let $\spm_2(\canal)$ denote the map of $\canal$
obtained by the algorithm. Third, we merge the two maps
$\spm_1(\canal)$ and $\spm_2(\canal)$ to obtain $\spm(\canal)$. This
is done by using the merge step from the standard divide-and-conquer
Voronoi diagram algorithm to compute the region closer to the
generators at $\overline{xd}$ than those at $\overline{yz}$.
We will provide more details for this step below and we will show that this step can be done in time linear in the total size of the two maps $\spm_1(\canal)$ and $\spm_2(\canal)$. Before doing so, we analyze the complexities of the algorithm.

The first step takes $O(n+h\log
h)$ time for all canals. So is the second step. The third step takes
linear time in the total size of $\spm_1(\canal)$ and
$\spm_2(\canal)$. Since the total size of the two maps over all canals
is $O(n)$, the total time of the third step for all canals is $O(n)$.
In summary, the time for computing the shortest path maps in all
canals is $O(n+h\log h)$ and the space is $O(n)$.

In the following, we provide more details for the third step of the algorithm.

Recall that $x$ and $y$ are the two endpoints of the corridor path $\pi$ in $\canal$. It is possible that the shortest $s$-$x$ path $\pi(s,x)$ contains $y$ or the shortest $s$-$y$ path $\pi(s,y)$ contains $x$. To determine that, we can simply check whether $d(s,x)+|\pi|=d(s,y)$ and whether $d(s,y)+|\pi|=d(s,x)$. Note that both $d(s,x)$ and $d(s,y)$ are available once $\spm(\calM)$ is computed.

We first consider the case where neither $\pi(s,x)$ contains $y$ nor $\pi(s,y)$ contains $x$.
In this case, there must be a point $p^*$ in $\pi$ such that $d(s,x)+|\pi(x,p^*)|=d(s,y)+|\pi(y,p^*)|$, where $\pi(x,p^*)$ (resp., $\pi(y,p^*)$) is the subpath of $\pi$ between $x$ (resp., $y$) and $p^*$. We can easily find $p^*$ in $O(|\pi|)$ time. To merge the two maps $\spm_1(\canal)$ and $\spm_2(\canal)$ to obtain $\spm(\canal)$, we find a dividing curve $\Gamma$ in $\canal$ such that $W(\overline{xd})$ claims all points of $\canal$ on one side of $\Gamma$ while $W(\overline{yz})$ claims all points of $\canal$ on the other side of $\Gamma$, where $W(\overline{xd})$ is the set of generators of $\spm'(\calM)$ claiming $\overline{xd}$ (one may consider $W(\overline{xd})$ is a wavefront) and $W(\overline{yz})$ is the set of generators of $\spm'(\calM)$ claiming $\overline{yz}$. The curve $\Gamma$ consists of all points in $\canal$ that have equal weighted distances to $W(\overline{xd})$ and $W(\overline{yz})$. Therefore, the point $p^*$ must be on $\Gamma$. Starting from $p^*$, we can trace $\Gamma$ out by walking simultaneously in the cells of the two maps $\spm_1(\canal)$ and $\spm_2(\canal)$.
The running time is thus linear in the total size of the two maps.

We then consider the case where either $\pi(s,x)$ contains $y$ or $\pi(s,y)$ contains $x$. Without loss of generality, we assume the latter case. We first check whether $\pi(s,p)$ contains $x$ for all points $p$ on the gate $\overline{yz}$. To do so, according to the definitions of corridor paths and gates of canals, for any point $p\in \overline{yz}$, its shortest path to $x$ in $\canal$ is the concatenation of the corridor path $\pi$ and $\overline{yp}$.
Hence, it suffices to check whether $d(s,z)=d(s,y)+|\overline{yz}|$.

\begin{itemize}
\item
If yes, then all points of $\overline{yz}$ are claimed by the generators of $W(\overline{xd})$ (in fact, they are claimed by $x$ because for any point $q\in \overline{xd}$ and any point $p\in \overline{yz}$, their shortest path in $\canal$ is the concatenation of $\overline{qx}$, $\pi$, and $\overline{yp}$). Hence, all points in $\canal$ are claimed by $W(\overline{xd})$ and $\spm(\canal)$ is $\spm_1(\canal)$.

\item
Otherwise, some points of $\overline{yz}$ are claimed by $W(\overline{xd})$ while others are claimed by $W(\overline{yz})$. As in the above case, we need to find a dividing curve $\Gamma$ consisting of all points with equal weighted distances to $W(\overline{xd})$ and $W(\overline{yz})$. To this end, we again first find a point $p^*$ in $\Gamma$. For this, since $\pi(s,y)$ contains $x$, $y$ is claimed by $W(\overline{xd})$. On the other hand, since $d(s,z)\neq d(s,y)+|\overline{yz}|$, $z$ is claimed by $W(\overline{yz})$. Therefore, $\overline{yz}$ must contains a point $p^*\in \Gamma$. Such a point $p^*$ can be found by traversing $\overline{yz}$ simultaneously in the cells of both $\spm_1(\canal)$ and $\spm_2(\canal)$. After $p^*$ is found, we can again trace $\Gamma$ out by walking simultaneously in the cells of $\spm_1(\canal)$ and $\spm_2(\canal)$. The running time is also linear in the total size of the two maps.
\end{itemize}

\subsection{Wrapping things up}


\begin{theorem}
Suppose $\calP$ is a set of $h$ pairwise disjoint polygonal obstacles with a total of $n$ vertices in the plane and $s$ is a source point. Assume that a triangulation of the free space is given. The shortest path map $\spm(s)$ with respect to $s$ can be constructed in $O(n+h\log h)$ time and $O(n)$ space.
\end{theorem}
\begin{proof}
Using the triangulation, we decompose the free space $\calF$ into an ocean $\calM$, canals, and bays in $O(n)$ time~\cite{ref:KapoorAn97}. Then, the shortest path map $\spm(\calM)$ in the ocean $\calM$ can be constructed in $O(n+h\log h)$ time and $O(n)$ space. Next, $\spm(\calM)$ can be expanded into all bays and canals in additional $O(n+h\log h)$ time and $O(n)$ space. The shortest path map $\spm(s)$ is thus obtained.
\end{proof}

The current best algorithms can compute a triangulation of the free space in $O(n\log n)$ time or in $O(n+h\log^{1+\epsilon}h)$ time for any small $\epsilon>0$~\cite{ref:Bar-YehudaTr94}. If all obstacles of $\calP$ are convex, then the triangulation can be done in $O(n+h\log h)$ time~\cite{ref:HertelFa85}.

After $\spm(s)$ is computed, by building a point location data structure~\cite{ref:EdelsbrunnerOp86,ref:KirkpatrickOp83} on $\spm(s)$ in additional $O(n)$ time, given a query point $t$, the shortest path length from $s$ to $t$ can be computed in $O(\log n)$ time and a shortest \st\ path can be produced in time linear in the number of edges of the path.

\begin{corollary}
Suppose $\calP$ is a set of $h$ pairwise disjoint polygonal obstacles with a total of $n$ vertices in the plane and $s$ is a source point. Assume that a triangulation of the free space is given. A data structure of $O(n)$ space can be constructed in $O(n+h\log h)$ time and $O(n)$ space, so that given any query point $t$, the shortest path length from $s$ to $t$ can be computed in $O(\log n)$ time and a shortest \st\ path can be produced in time linear in the number of edges of the path.
\end{corollary}




\footnotesize
 \bibliographystyle{plain}
\bibliography{reference}

\end{document}